\documentclass[a4paper,11pt]{article}
\usepackage[top=25mm,left=22mm]{geometry}\textwidth17cm\textheight24cm
\usepackage{xcolor}
\usepackage{amssymb}
\usepackage{amsthm}
\usepackage{graphicx}
\usepackage{authblk}
\usepackage[]{listings}

\usepackage{mathtools}
\usepackage{subfig}
\usepackage{mymath}
\usepackage{natbib}
\usepackage{diffusion_instability}
\usepackage[colorlinks=true]{hyperref}
\usepackage[capitalise]{cleveref}

\begin{document}
\title{From 0D to 1D spatial models using \OCMAT}
\author[1]{D.~Gra\ss}
\affil[1]{\myadress. E-mail: \myemail}
\maketitle
\begin{abstract}
We show that the standard class of optimal control models in \OCMAT{}\footnote{\OCMAT{} is a \MATL{} package that provides tools for the numerical analysis of (non-distributed) \OCPROs, specifically over an infinite time horizon. It can be downloaded from \httpOCMAT.}can be used to analyze 1D spatial distributed systems. This approach is an intermediate step on the way to the FEM discretization approach presented in \citet{grassuecker2015}. Therefore, the spatial distributed model is transformed into a standard model by a finite difference discretization. This (high dimensional) standard model is then analyzed using \OCMAT. As an example we apply this method to the spatial distributed shallow lake model formulated in \citet{brockxepapadeas2008}. The results are then compared with those of the FEM discretization in \citet{grassuecker2015}. 
\end{abstract}
%
\keywords{spatial distributed optimal control model, finite difference discretization, shallow lake model, patterned indifference threshold point}
\section{Introduction}
\label{sec:introduction}
The analysis of spatial distributed optimal control problems (over an infinite time horizon) become an important issue in economic modeling, see e.g., \citet{brockxepapadeas2008,brocketal2014}. In technical terms this means that the time evolution of the space distributed states are described by parabolic \PDE s. In \citet{brockxepapadeas2008} the authors provide a local stability analysis of the equilibria, i.e. the solutions of the elliptic \PDE s associated to the derived canonical system. They derive a set of conditions that cause a Turing instability of these equilibria and call the bifurcating patterned equilibria \hoss{} solutions, i.e.~Patterned/Heterogeneous Optimal Steady States in contrast to Homogeneous or Flat Optimal Steady States (\foss).

But, as is shown in \citet{grassuecker2015}, a local stability analysis is not sufficient to prove the optimality of heterogeneous equilibria solutions. To answer the question of optimality the objective values of the paths that converge to the different equilibria have to be compared. This analysis also sheds new light on the discussion about indifference threshold and threshold points, cf.~\citet{kiselevawagener2008a,kiseleva2011}. We therefore introduce a new terminology of defective and non-defective equilibria. This properties distinguishes between optimal equilibria that are stable or unstable.

Since in general the \PDEs{} cannot be solved analytically we have to resort to numerical methods. In \citet{grassuecker2015} we present a numerical procedure relying on a \FEM{} discretization of the derived canonical system combined with a continuation strategy, analogous to the approach in \OCMAT, cf.~\citet{grass2012}. As an example we used the distributed shallow lake model \cref{eq:sldiffusion_model}. 

This allowed us to identify parameter regions where non-defective \hoss{} exist, and where defective \hoss{} and the according stable manifolds separates the regions of attractions of \foss{} and \hoss{} (threshold points/manifolds). Additionally we were able to show the existence of homogeneous/heterogeneous indifference threshold (Skiba) points and calculated the connecting manifold of indifference threshold points.

In this paper we take an intermediate step to the full \FEM{} approximation of the canonical system, to demonstrate the capabilities of the standard optimal control class of \OCMAT, i.e. infinite time horizon problems, where the time evolution of the states is described by \ODE s. To apply the \OCMAT{} processes for the initialization and file generation directly we discretize the \PDE s of the state equations by a finite difference scheme (FD). This yields a number of \ODE s that can be handled by \OCMAT. The numerical results of this approach are then compared with the analogous results of \citet{grassuecker2015}.

The article is structured in the following way. We start with a discussion and introduction of general terms and shortly present the finite difference discretization in \cref{sec:GeneralDefinitions}. In \cref{sec:sl_model_withoutdiff} we summarize important properties and results of the 0D (non-distributed) shallow lake model, introduced and analyzed in \citet{scheffer1998,maeleretal2003,carpenterbrock2004} and \citet{wagener2003}. In the next \cref{sec:sl_model} the 1D spatial distributed shallow lake model is formulated together with its discretized counterpart. The latter model is then numerically analyzed in detail.

\section{General Definitions}
\label{sec:GeneralDefinitions}

\subsection{Models of spatial dimension 0 (0D model) }
\label{sec:ModelsOfSpatialDimension00DModel}

\begin{subequations}
\label{eq:gen_0Dmodel}
\begin{align}
&\max_{u(\cdot)}\int_0^\infty\E^{-\rho t}g(x(t),u(t))\!\Dt\label{eq:gen_0Dmodel_obj}\\
\text{s.t.}\quad&\dot x(t)=f(x(t),u(t))\label{eq:gen_0Dmodel_dyn}\\
&x(0)=x_0\in\R^n.\label{eq:gen_0Dmodel_ic}
\end{align}
\end{subequations}
with $f\in\Cset[2](\R^n\times\R^m,\R^n),\ g\in\Cset[2](\R^n\times\R^m,\R)$.

Let $(x^\opt(\cdot),u^\opt(\cdot))$ be an optimal solution of \cref{eq:gen_0Dmodel}. Then there exists $\lambda(\cdot)$ such that $(x^\opt(\cdot),u^\opt(\cdot),\lambda(\cdot))$ is a solution of the canonical system
\begin{subequations}
\label{eq:gen_0Dmodel_cansys}
\begin{align}
&\dot x(t)=f(x^\opt(t),u(t))\label{eq:gen_0Dmodel_cansys1}\\
&\dot\lambda(t)=\rho\lambda(t)-\Ha(x^\opt(t),u^\opt(t),\lambda(t),\lambda_0)\label{eq:gen_0Dmodel_cansys2}\\
&x^\opt(0)=x_0\label{eq:gen_0Dmodel_cansys4}
\shortintertext{with}
&u^\opt(t)=\argmax_{u}\Ha(x^\opt(t),u,\lambda(t),\lambda_0)\label{eq:gen_0Dmodel_cansys5}
\shortintertext{and}
&\Ha(x,u,\lambda,\lambda_0)\defin\lambda_0 g(x,u)+\lambda^\transp f(x,u),
\end{align}
\end{subequations}
To ease the notation we make the following assumptions
\begin{assumption}
\label{ass:normal}
Problem \cref{eq:gen_0Dmodel} is normal, i.e., $\lambda_0=1$. Therefore, we omit the argument $\lambda_0$.
\end{assumption}
\begin{assumption}
\label{ass:analyticcontrl}
Let $(x^\opt(\cdot),u^\opt(\cdot))$ be an optimal solution and $\lambda(\cdot)$ the according costate. Then, there exists an explicit function 
\begin{align*}
	&u^\ext(x,\lambda)\in\Cset[2]\left(\R^n\times\R^n,\R^m\right),
	\shortintertext{such that for every $t$}
	&\Ha(x^\opt(t),u^\ext(x^\opt(t),\lambda(t)),\lambda(t))=\max_{u}\Ha(x^\opt(t),u,\lambda(t)).
\end{align*}
\end{assumption}
Then an optimal solution $(x^\opt(\cdot),u^\opt(\cdot))$ can be found among the solutions satisfying
\begin{subequations}
\label{eq:cansys0D}
\begin{align}
&\dot x(t)=\bar f(x(t),\lambda(t))\label{eq:cansys0D1}\\
&\dot\lambda(t)=\rho\lambda(t)-\bar\Ha_x(x(t),\lambda(t))\label{eq:cansys0D2}\\
&x(0)=x_0.\label{eq:cansys0D3}
\shortintertext{with}
&\bar f(x(t),\lambda(t))\defin f(x(t),u^\ext(x(t),\lambda(t)))\notag
\shortintertext{and}
&\bar\Ha(x(t),\lambda(t))\defin\Ha(x(t),u^\ext(x(t),\lambda(t)),\lambda(t),1).\label{eq:cansys0D5}
\end{align}
\end{subequations}
Subsequently we will omit the bar sign.
\begin{definition}[\oss]
\label{def:oss}
Let $(x^\opt(\cdot),u^\opt(\cdot))$ with $x^\opt(\cdot)\equiv\hat x\in\R^n$ and $u^\opt(\cdot)\equiv\hat u\in\R^m$ be an optimal solution of problem \labelcref{eq:gen_0Dmodel} with $x(0)=\hat x$ . Then $(\hat x,\hat u)$ is called an \emph{optimal equilibrium} and denoted as \oss.
\end{definition}
Let $(\hat x,\hat\lambda)\in\R^{2n}$ be an equilibrium of the canonical system \cref{eq:cansys0D}. Then $(\hat x,\hat\lambda)\in\R^{2n}$ is denoted as \css{} and 
\begin{equation}
\label{eq:cansys0Djac}
	J(\hat x,\hat\lambda)\defin\evalat{\begin{pmatrix}
	\dfrac{\D\!f(x,\lambda)}{\Dx} &\dfrac{\D\!f(x,\lambda)}{\Dx[\lambda]}\\
	-\dfrac{\D\!\Ha^\ext(x,\lambda)}{\Dx} &r-\dfrac{\D\!\Ha^\ext(x,\lambda)}{\Dx[\lambda]}
\end{pmatrix}}_{(\hat x,\hat\lambda)}
\end{equation}
is the according \emph{Jacobian matrix}, and if there is no ambiguity we simply write $\hat J$. The eigenspaces corresponding to $J(\hat x,\hat\lambda)$ are denoted as
\begin{subequations}
\label{eq:hatjeigspc}
\begin{align}
	&\eigspace{s}(\hat x,\hat\lambda)\defin\{\eigval\in\C:\ J(\hat x,\hat\lambda)\eigvec=\eigval\eigvec\ \text{with}\ \Re\eigval<0\},\ n_s\defin\dim\eigspace{s}(\hat x,\hat\lambda)\\
	&\eigspace{u}(\hat x,\hat\lambda)\defin\{\eigval\in\C:\ J(\hat x,\hat\lambda)\eigvec=\eigval\eigvec\ \text{with}\ \Re\eigval>0\},\ n_u\defin\dim\eigspace{u}(\hat x,\hat\lambda)\\
	&\eigspace{c}(\hat x,\hat\lambda)\defin\{\eigval\in\C:\ J(\hat x,\hat\lambda)\eigvec=0\},\ n_c\defin\dim\eigspace{c}(\hat x,\hat\lambda).
\end{align}
\end{subequations}
\begin{definition}[Saddle Point Property]
\label{def:spp}
Let $(\hat x,\hat\lambda)\in\R^{2n}$ be an equilibrium of \cref{eq:cansys0D}. If 
\begin{equation}
\label{eq:sppdim}
	\dim\eigspace{s}(\hat x,\hat\lambda)=n
\end{equation}
then it is said, that the equilibrium satisfies the \emph{saddle point property (\spp)}. The equilibrium $(\hat x,\hat\lambda)$ is denoted as \cssspp{}. Otherwise it is denoted as \cssnspp{}. The number
\begin{equation}
\label{eq:sppdefect}
	\defect(\hat x,\hat\lambda)\defin n_s-n_u-n_c
\end{equation}
is called the \emph{defect} of $(\hat x,\hat\lambda)$. An equilibrium with defect $\defect(\hat x,\hat\lambda)<0$ is called \emph{defective}, otherwise it is called \emph{non-defective}. If $(\hat x,\hat\lambda)$ is defective and the according $(\hat x,u^\ext(\hat x,\hat\lambda))$ is \oss, then $(\hat x,u^\ext(\hat x,\hat\lambda))$ is called \emph{defective} otherwise it is called \emph{non-defective}.
\end{definition}
\begin{proposition}
\label{prop:equivspp}
Let $(\hat x,\hat\lambda)\in\R^{2n}$ be an equilibrium of \cref{eq:cansys0D} and $\rho>0$. Then $(\hat x,\hat\lambda)$ satisfies the saddle point property $\iff$ every eigenvalue $\eigval$ of the according Jacobian $J(\hat x,\hat\lambda)$ satisfies
\begin{equation}
\label{eq:spp}
	\norm{\Re\eigval-\frac{\rho}{2}}>\frac{\rho}{2}
\end{equation}
\end{proposition}
\begin{proof}
In \citet{grassetal2008} it is proved that there exist $n$ (not necessarily distinct) complex numbers $\bar\eigval\in\C$ with $\Re\bar\eigval\ge0$ such that any eigenvalue $\eigval$ of the according Jacobian, satisfies 
\begin{align*}
	&\eigval=\frac{\rho}{2}+\bar\eigval\quad\text{or}\quad\eigval=\frac{\rho}{2}-\bar\eigval.
	\intertext{This symmetry together with \spp{} yields}
	&\norm{\Re\eigval-\frac{\rho}{2}}=\Re\bar\eigval>\frac{\rho}{2}
	\shortintertext{and the last inequality is identical to}
	&\Re\bar\eigval>\frac{\rho}{2}\quad\iff\quad\Re\eigval<0.
\end{align*}
Therefore, $n$ eigenvalues have a negative real part finishing the proof.
\end{proof}
\begin{remark}
\cref{prop:equivspp} allows us to formulate the \spp{} by the equivalent \cref{eq:spp}. A definition of \spp{} claiming \cref{eq:spp} has the advantage that it also can be applied to equilibria of distributed systems. Where a definition relying on the dimension of the stable manifold fails.
\end{remark}
\subsection{Models of spatial dimension 1 (1D model) }
\label{sec:ModelsOfSpatialDimension11DModel}
We assume that the spatially distributed model is derived by the introduction of a diffusion term, whereas the functions $f$ and $g$ are the same as for model \labelcref{eq:gen_0Dmodel}. This yields
\begin{align*}
&\max_{u(\cdot,\cdot)}\int_0^\infty\int_{-L}^L\E^{-\rho t}g(x(z,t),u(z,t))\!\Dx[z]\!\Dt\\
\text{s.t.}\quad&\frac{\partial}{\partial t}x(z,t)=f(x(z,t),u(z,t))+D\frac{\partial^2x(z,t)}{\partial z^2}\\
&\evalat{\frac{\partial x(z,t)}{\partial z}}_{\pm L}=0\\
&x(z,0)=x_0(z),\quad z\in[-L,L].
\end{align*}
Or transforming $[-L,L]$ into $[0,1]$ yields
\begin{subequations}
\label{eq:gen_1Dmodel}
\begin{align}
&\max_{u(\cdot,\cdot)}\int_0^\infty\int_0^1\E^{-\rho t}g(x(z,t),u(z,t))\!\Dx[z]\!\Dt\label{eq:gen_1Dmodel_obj}\\
\text{s.t.}\quad&\frac{\partial}{\partial t}x(z,t)=f(x(z,t),u(z,t))+\frac{D}{(2L)^2}\frac{\partial^2x(z,t)}{\partial z^2}\label{eq:gen_1Dmodel_dyn}\\
&\evalat{\frac{\partial x(z,t)}{\partial z}}_{0,1}=0\label{eq:gen_1Dmodel_bc}\\
&x(z,0)=x_0(z),\quad z\in[0,1].\label{eq:gen_1Dmodel_ic}
\end{align}
\end{subequations}
Applying \PMAXP\ for \PDEs, see e.g., \citet{troeltzsch2009}, we can derive, analogous to \cref{eq:cansys0D}, the canonical system for \labelcref{eq:gen_1Dmodel} as
\begin{subequations}
\label{eq:gen_1Dmodel_model_cansys}
\begin{align}
&\frac{\partial}{\partial t} x(z,t)=f(x(z,t),\lambda(z,t))+D\frac{\partial^2x(z,t)}{\partial x^2}\label{eq:gen_1Dmodel_model_cansys1}\\
&\frac{\partial}{\partial t}\lambda(z,t)=\rho\lambda(z,t)-\frac{\partial\Ha(x(z,t),\lambda(z,t))}{\partial x}-D\frac{\partial^2\lambda(z,t)}{\partial x^2}\label{eq:gen_1Dmodel_model_cansys2}\\
&\evalat{\partial_n x(z,t)}_{0,1}=0\label{eq:gen_1Dmodel_model_cansys3}\\
&\evalat{\partial_n\lambda(z,t)}_{0,1}=0\label{eq:gen_1Dmodel_model_cansys4}\\
&x(z,0)=x_0(z),\quad z\in[0,1].\label{eq:gen_1Dmodel_model_cansys5}
\end{align}
\end{subequations}
For the numerical analysis we can than e.g. use a finite element method (\FEM) for the discretization of \cref{eq:gen_1Dmodel_model_cansys}. For the distributed shallow lake model (cf.~\cref{sec:sl_model}) this has been carried out in \citet{grassuecker2015}. In this article we want to demonstrate the capabilities of \OCMAT{} and transform it into a high dimensional 0D model. 

For the \FDM\ discretization we use an equidistant grid of length $h$ with $Nh=1$ and $N\in\N$.
\begin{subequations}
\label{eq:discretization}
\begin{align}
	&\evalat{\frac{\D}{\Dx[z]}x(z)}_{z_i}\approx\frac{x(z_i+h)-x(z_i-h)}{2h}\\
	&\evalat{\frac{\D^2}{\Dx[z]^2}x(z)}_{z_i}\approx\frac{x(z_i-h)-2x(z_i)+x(z_i+h)}{h^2}\\
	&\int_0^1g(z)\!\Dx\approx h\left(\sum_{i=1}^{N-1}g(z_i)+\frac{g(z_0)+g(z_N)}{2}\right)
\end{align}
\end{subequations}
Thus, model \labelcref{eq:gen_1Dmodel} is approximated by
\begin{subequations}
\label{eq:dgen_1Dmodel}
\begin{align}
&\max_{u_0(\cdot),\ldots,u_N(\cdot)}\frac{1}{N}\left\{\int_0^\infty\E^{-\rho t}\left(\sum_{i=1}^{N-1}g(x_i(t),u_i(t))+\frac{g(x_0(t),u_0(t))+g(x_N(t),u_N(t))}{2}\right)\!\Dt\right\}\label{eq:dgen_1Dmodel_obj}\\
\text{s.t.}\quad&\dot x_i(t)=f(x_i(t),u_i(t))+\frac{DN^2}{(2L)^2}\left(x_{i-1}(t)-2x_i(t)+x_{i+1}(t)\right)\label{eq:dgen_1Dmodel_dyn}\\
&x_1(t)-x_{-1}(t)=x_{N+1}(t)-x_{N-1}(t)=0,\quad t\ge0\label{eq:dgen_1Dmodel_bc}\\
&x_i(0)=x_{i,0}.\label{eq:dgen_1Dmodel_ic}
\shortintertext{with}
&z_i=\frac{i}{N},\quad i=0,\ldots,N\notag\\
&x_i(t)\defin x(z_i,t),\quad u_i(t)\defin u(z_i,t)\notag\\
&x_i\defin x(z_i,\cdot),\quad u_i\defin u(z_i,\cdot)\notag
\end{align}
\end{subequations}
Note that we do not consider the problem under which conditions model \labelcref{eq:dgen_1Dmodel} approximates \labelcref{eq:gen_1Dmodel}. We rather take, in the specific case of the spatial shallow lake model, model \labelcref{eq:dgen_1Dmodel} for granted to see if \OCMAT\ can handle such a problem and compare the results with those from \citet{grassuecker2015}.

The canonical system for model \labelcref{eq:dgen_1Dmodel} becomes
\begin{subequations}
\label{eq:dcansys1D}
\begin{align}
&\dot x_i(t)=f(x_i(t),\lambda_i(t))+\diffcoeff^{(x)}_i(t)\label{eq:dcansys1D1}\\
&\dot\lambda_i(t)=\rho\lambda_i(t)-\Ha_x(x_i(t),\lambda_i(t))-\diffcoeff^{(\lambda)}_i(t)\label{eq:dcansys1D2}\\
&x_i(0)=x_{i,0}.\label{eq:dcansys1D3}
\shortintertext{and}
&u_i^\ext=u_i(x_i,\lambda_i)\notag\\
&\tilde D\defin\frac{DN^2}{(2L)^2}\notag\\
&\diffcoeff^{(x)}_i\defin\begin{cases}
	2\tilde D(x_1-x_0) & i=0\\
	\tilde D(x_{i-1}-2x_i+x_{i+1}) & i=1,\ldots,N-1\\
	2\tilde D(x_{N-1}-x_N) & i=N
\end{cases}\notag\\
&\diffcoeff^{(\lambda)}_i\defin\begin{cases}
	\tilde D(\lambda_1-2\lambda_0) & i=0\\
	\tilde D(2\lambda_0-2\lambda_1+\lambda_2) & i=1\\
	\tilde D(\lambda_{i-1}-2\lambda_i+\lambda_{i+1}) & i=2,\ldots,N-2\\
	\tilde D(\lambda_{N-2}-2\lambda_{N-1}+2\lambda_N) & i=N-1\\
	\tilde D(\lambda_{N-1}-2\lambda_N & i=N)
\end{cases}\notag
\end{align}
\end{subequations}
To abbreviate notation we introduce
\begin{align*}
	&x^d\defin(x_0^\transp,\ldots,x_N^\transp)^\transp\in\R^{n(N+1)}\\
	&\lambda^d\defin(\lambda_0^\transp,\ldots,\lambda_N^\transp)^\transp\in\R^{n(N+1)}\\
	&u^d\defin(u_0^\transp,\ldots,u_N^\transp)^\transp\in\R^{m(N+1)}.
\end{align*}
\begin{definition}[\foss{} and \hoss]
\label{def:foss}
Let $(x^{d,\opt}(\cdot),u^{d,\opt}(\cdot))$ with $x^{d,\opt}(\cdot)\equiv\hat x^d\in\R^{n(N+1)}$ and $u^{d,\opt}(\cdot)\equiv\hat u^d\in\R^{m(N+1)}$ be an optimal solution of problem \labelcref{eq:gen_0Dmodel} with $x^d(0)=\hat x^d$ . If 
\begin{equation*}
	\hat x_0=\hat x_1=\cdots=\hat x_N=\hat x\in\R^n
\end{equation*}
then $(\hat x^d,\hat u^d)$ is called a \emph{flat (homogeneous) optimal steady state (\foss)}, otherwise it is called an \emph{patterned (heterogeneous) optimal steady state (\hoss)}.
\end{definition}
\begin{definition}[\fcss{} and \hcss]
\label{def:fssnhss}
Let $(\hat x^d, \hat\lambda^d)\in\R^{2n(N+1)}$ be an equilibrium of the canonical system \cref{eq:dcansys1D}. Then $(\hat x^d, \hat\lambda^d)$ is called a \emph{flat (homogeneous) steady state (\fcss)} $\iff$ 
\begin{equation}
\label{eq:fss}
	\hat x_0=\hat x_1=\cdots=\hat x_N=\hat x
\end{equation}
otherwise it is called a \emph{patterned (heterogeneous) steady state (\hcss)}. If a \fcss{} (\hcss) satisfies \spp{} it is denoted as \fcssspp{} (\hcssspp) otherwise it is denoted as \fcssnspp{} (\hcssnspp).
\end{definition}
\begin{definition}[State-Costate space]
\label{def:statecostatespace}
Let $(x^d(\cdot), \lambda^d(\cdot))$ be a solution of the canonical system \cref{eq:dcansys1D}. Then the representation $(\Norm{x^d}(\cdot), \Norm{\lambda^d}(\cdot))$ with
\begin{equation}
\label{eq:normx}
	\Norm{y^d_j}(t)\defin\frac{1}{N}\left(\sum_{i=1}^{N-1}\Norm{y^i_j(t)}+\frac{\Norm{y^0_j(t)}+\Norm{y^N_j(t)}}{2}\right),\quad y=x,\ \text{or}\ y=\lambda,\ j=1,\ldots,n.
\end{equation}
is called a \emph{solution path in the state-costate space}.
\end{definition}
\begin{remark}
In \citet{brockxepapadeas2008} \foss{} (\hoss) were introduced as \emph{flat (patterned) optimal equilibria} of the canonical system satisfying \spp. We enhanced this terminology for two reasons
\begin{enumerate}
	\item For better clearness we decided to make a further distinction between the canonical and optimal system. 
	\item There exist optimal equilibria that do not satisfy \spp, cf.~\cref{sec:HOSSNotSatisfyingSPP}.
\end{enumerate}
\end{remark}

\section{The shallow lake model without spatial diffusion}
\label{sec:sl_model_withoutdiff}
A well known version of the shallow lake model, see e.g.~\citet{wagener2003}, can be formulated as
\begin{subequations}
\label{eq:sl_model}
\begin{align}
&\max_{u(\cdot)}\int_0^\infty\E^{-\rho t}\left(\ln(u(t))-cP(t)^2\right)\!\Dt\label{eq:sl_model_obj}\\
\text{s.t.}\quad&\dot P(t)=u(t)-bP(t)+\frac{P(t)^2}{1+P(t)^2}\label{eq:sl_model_dyn}\\
&P(0)=P_0>0.\label{eq:sl_model_ic}
\end{align}
\end{subequations}
We sometimes refer to \modelcref{eq:sl_model} as \emph{0D shallow lake model}. 

By \PMAXP\ we find the canonical system
\begin{subequations}
\label{eq:sl_model_cansys}
\begin{align}
&\dot P(t)=u^\ext(t)-bP(t)+\frac{P(t)^2}{1+P(t)^2}\label{eq:sl_model_cansys1}\\
&\dot \lambda(t)=2cP(t)+\lambda(t)\left(\rho+b-\frac{2P(t)}{\left(1+P(t)^2\right)^2}\right)\label{eq:sl_model_cansys2}\\
&P(0)=P_0\label{eq:sl_model_cansys4}
\shortintertext{with}
&u^\ext(t)=-\frac{1}{\lambda(t)}.
\end{align}
\end{subequations}
Let $(P^\opt(\cdot),u^\opt(\cdot))$ be the optimal solution of \cref{eq:sl_model}; then $P^\opt(\cdot)$ is the unique solution of the \IVP
\begin{subequations}
\label{eq:sl_model_optsys}
\begin{align}
&\dot P(t)=u^\opt(t)-bP(t)+\frac{P(t)^2}{1+P(t)^2}\label{eq:sl_model_optsys1}\\
&P(0)=P_0.\label{eq:sl_model_optsys2}
\end{align}
\end{subequations}
The \odecref{eq:sl_model_optsys1} is called the \emph{optimal system} for $u^\opt(t)$. In \citet{wagener2003} it is proved that every optimal solution $(P^\opt(\cdot),u^\opt(\cdot))$ for arbitrary $P_0>0$ converges to an equilibrium $(\hat P,\hat u)$ of the optimal system with $\hat u>0$, usually depending on $P_0$.

A detailed bifurcation analysis in the parameter space $(b,c)$, see, e.g.,~\citet{wagener2003}, reveals the existence of regions in the parameter space where the optimal system consists of
\begin{itemize}
	\item One globally stable optimal equilibrium $(\hat P,\hat u)$.
	\item Two locally stable equilibria $(\hat P_o,\hat u_o)$ and $(\hat P_e,\hat u_e)$. These are separated by one of the following state values
	\begin{itemize}
		\item The state value $\hat P_u$ of an unstable optimal equilibrium $(\hat P_u,\hat u_u)$.
		\item An indifference threshold point $P_I$ also called Skiba point. 
	\end{itemize}
\end{itemize}
Thus, we can give a full classification of the optimal solutions. In the case that three optimal equilibria exist we choose $\hat P_o$ and $\hat P_e$ such that $\hat P_o<\hat P_u<\hat P_e$. There are intermediate cases (bifurcation cases) where equality holds that are not specifically mentioned. We also refer to $\hat P_e$ as the \emph{eutrophic} and to $\hat P_o$ as the \emph{oligotrophic} equilibrium.
\begin{description}
	\item[Global stable:] For any initial state $P_0>0$ there exists a unique solution $(P^\opt(\cdot),u^\opt(\cdot))$ that converges to $(\hat P,\hat u)$, which is independent of $P_0$. See \cref{fig:sl_multipleequc}.
	\item[Local stable I:] For any initial state $0<P_0<\hat P_u$ there exists a unique solution $(P^\opt(\cdot),u^\opt(\cdot))$ that converges to $(\hat P_o,\hat u_1)$, which is independent of $P_0$. For $P_0>\hat P_u$ there exists a unique solution $(P^\opt(\cdot),u^\opt(\cdot))$ that converges to $(\hat P_e,\hat u_2)$, which is independent of $P_0$. For $P_0=\hat P_u$ the optimal solution is $(P^\opt(\cdot),u^\opt(\cdot))\equiv(\hat P_u,\hat u_u)$. See \cref{fig:sl_optmultipleequc}. $\hat P_u$ is called a \emph{threshold point}.
	\item[Local stable II:] The first two statements of the previous case remain true, replacing $\hat P_u$ by $P_I$. For $P_0=P_I$ there exist two optimal solutions $(P_i^\opt(\cdot),u_i^\opt(\cdot)),\ i=1,2$ that converge to $(\hat P_o,\hat u_1)$ and $(\hat P_e,\hat u_2)$ with $\hat P_o<P_I<\hat P_e$. See \cref{fig:sl_optmultipleuniqueequc}. $P_I$ is called an \emph{indifference threshold point} or \emph{Skiba point}.
\end{description}
From the perspective of optimal control theory the last two cases are of specific interest. The first case is often referred to as \emph{history dependence}, i.e., the optimal solution depends on its initial starting point. In the second case we additionally observe the non-uniqueness of the optimal solution, i.e., \emph{indifference}. For a detailed discussion and description of the underlying bifurcations we refer to \citet{kiselevawagener2008a,kiseleva2011}.

The parameter values for these two prototypical cases are specified in \cref{tb:base_par}. In the first case we find an indifference threshold point and in the second case a threshold point.

\begin{table}[tbp] \centering%
\begin{tabular}{c|c|c|c|c|c|c|c|}
\multicolumn{4}{c}{Model \labelcref{eq:sl_model}} &\multicolumn{3}{|c}{Model \labelcref{eq:sldiffusion_model} specific} \\\hline
Scenario & $\rho$ & $c$ & $b$ & $D$ & $L$ & $N$\\ \hline
I & $0.03$ & $0.5$ & $0.65^\ast$ & 0.5 & $2\pi/0.44$ & 50\\ \hline
II & $0.3$ & $3.5^\ast$ & $0.55$ & 0.5 & $2\pi/0.44$ & 50\\ \hline
\end{tabular}%
\caption{The parameter values for the two considered scenarios for the 0D and 1D model. The values with the superscript $^\ast$ denote the free parameter.}
\label{tb:base_par}%
\end{table}%

\paragraph{Bifurcation-analysis}
\label{sec:BifurcationAnalysis}
Anyhow, this classification is the result of an intensive numerical analysis of the canonical system \cref{eq:sl_model_cansys}. This analysis covers a bifurcation analysis of its equilibria (see \cref{fig:sl0D_bifan}) and the calculation of the related stable paths and their objective value. The numerical computation is necessary since the local properties of an equilibrium $(\hat P,\hat\lambda)$ of \cref{eq:sl_model_cansys} let us not deduce that the corresponding equilibrium $(\hat P,\hat u)$ with $\hat u=1/\hat\lambda$ is an optimal equilibrium of the optimal system \cref{eq:sl_model_optsys1}. This is specifically true in the case that multiple equilibria of the canonical system exist. Therefore, there is no one-to-one correspondence between the bifurcations of the optimal and canonical system.

\paragraph{Unique optimal equilibrium}
\label{sec:UniqueOptimalEquilibrium}
The importance of a numerical analysis of the stable paths comes specifically clear in scenario I with $b=0.75$. In \cref{fig:sl_multipleequ} we see that there exist three equilibria in the canonical system (two saddles and one unstable focus), whereas the optimal system only consists of one globally stable equilibrium. 

\paragraph{Indifference threshold point}
\label{sec:IndifferenceThresholdPoint}
For $b=0.65$ the number and properties of the equilibria of the canonical system remain the same, but the optimal system consists of two locally stable equilibria separated by an indifference threshold point $P_I$, cf.~\cref{fig:sl_optmultipleequ}. Thus, a local stability analysis of the equilibria has to be supported by the global analysis the according stable manifolds. Even if the first task can be realized analytically, the stable paths can only be calculated analytically in very rare cases. Usually and specifically in our case we have to resort to numerical methods to solve the latter task.

\paragraph{Indifference point}
\label{sec:IndifferencePoint}
For scenario II with $c=3.5$ the canonical system exhibits two saddles and one unstable node (cf.~\cref{fig:sl_optmultipleuniqueequa}). Calculating the stable paths and comparing the objective values we find that the unstable equilibrium is a threshold point (cf.~\cref{fig:sl_optmultipleuniqueequb}). This unstable equilibrium separates the regions of attraction for the two locally stable equilibria corresponding to the two saddles(cf.~\cref{fig:sl_optmultipleuniqueequc}). The second case with $c=3.0825$ yields qualitatively the same result and is not depicted.
\begin{figure}
\centering
	\subfloat[Equilibria bifurcation diagram]{\includegraphics[width=0.45\textwidth]{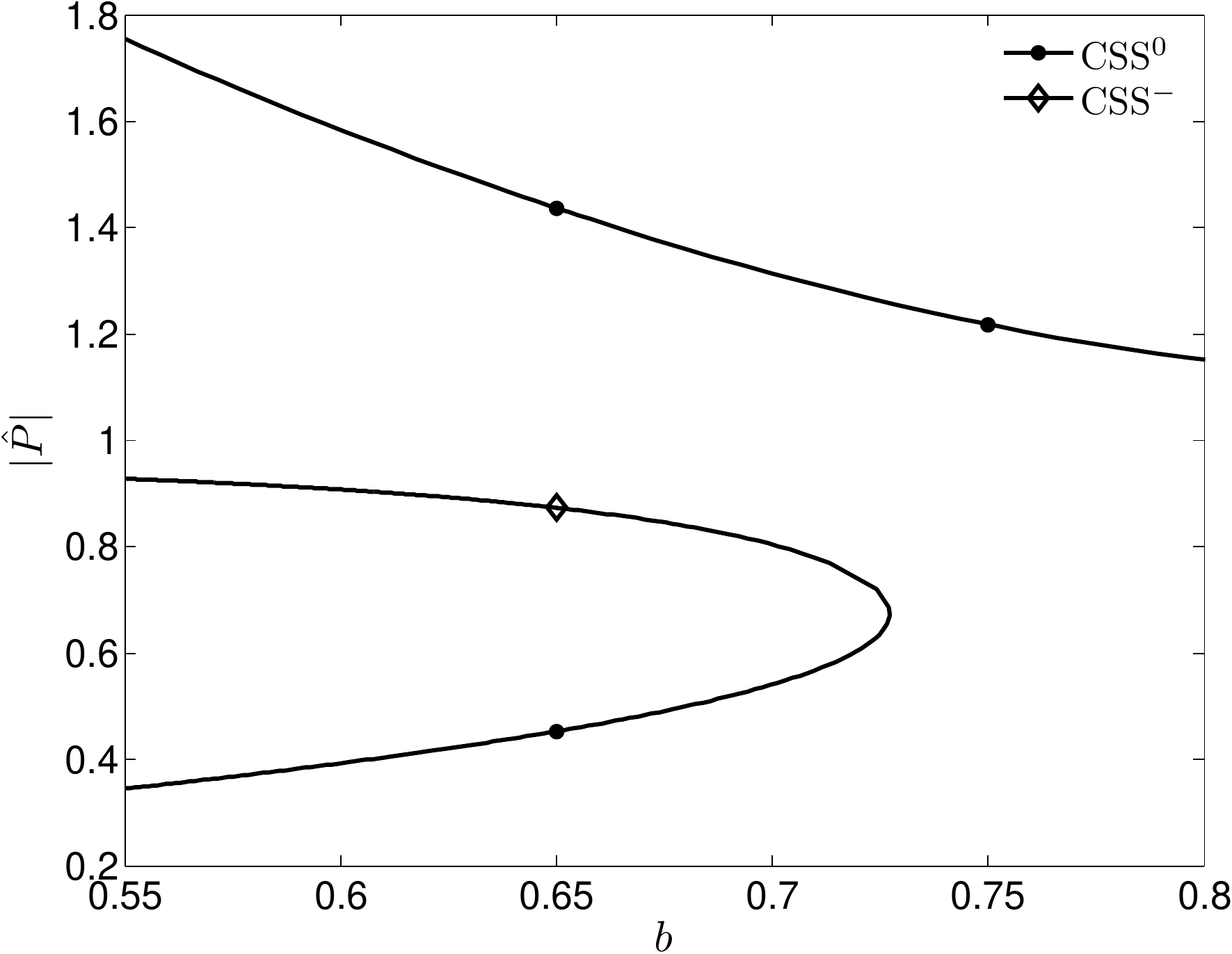}\label{fig:sl0D_bifan1}}\hfill
	\subfloat[Equilibria bifurcation diagram]{\includegraphics[width=0.45\textwidth]{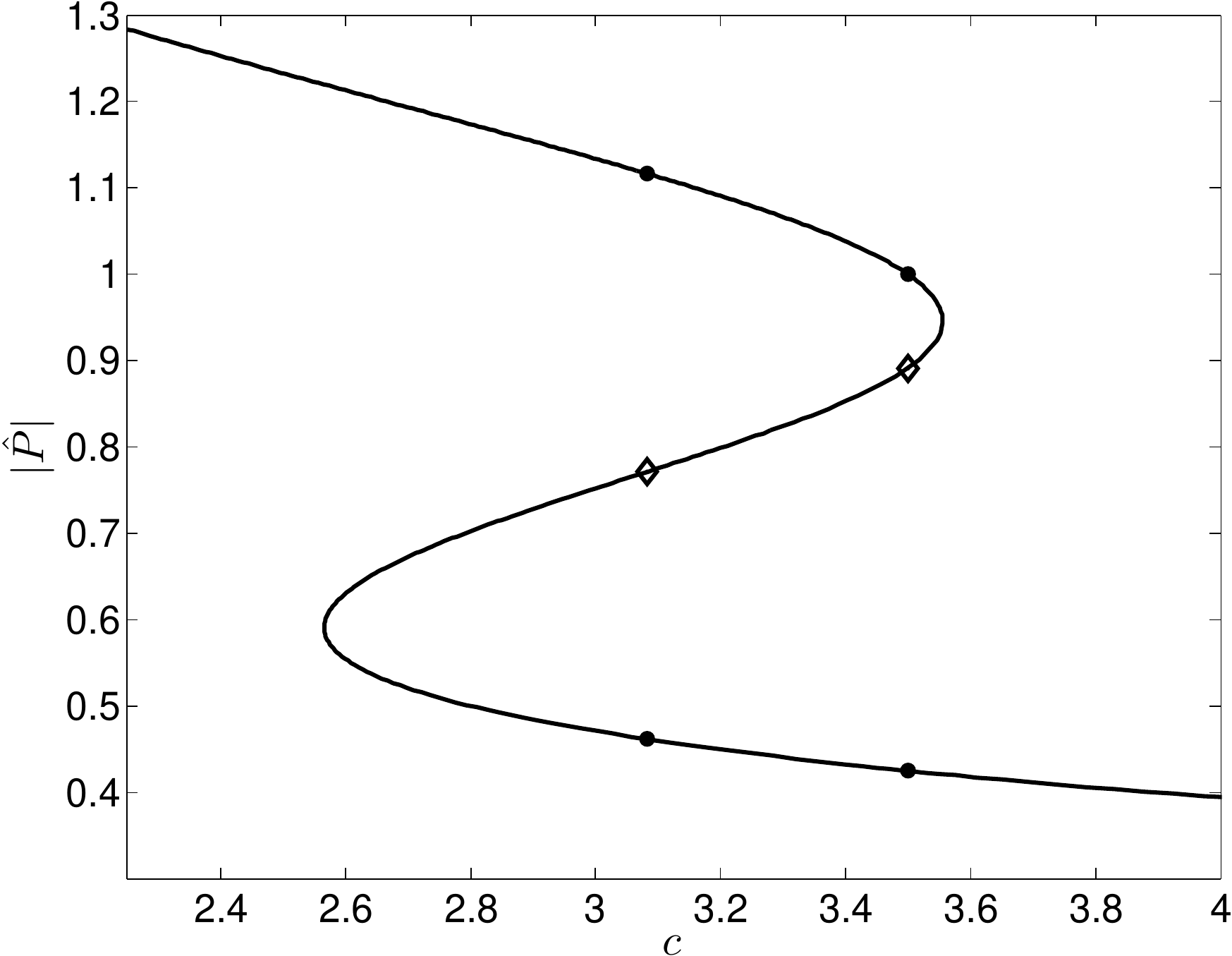}\label{fig:sl0D_bifan2}}
\caption[]{The $\bullet$ denote equilibria satisfying the \spp{} and $\diamond$ denote equilibria not satisfying \spp. The figures in the first row show the bifurcation parameter versus the (absolute) state value. In the second row the norm of the equilibrium is plotted versus the bifurcation parameter. Panel \subref{fig:sl0D_bifan1} ($\rho=0.03, c=0.5$ and varying $b$) show the existence of two separated branches of equilibria. In the interval $[0, 0.727]$ there exist three equilibria. Panel \subref{fig:sl0D_bifan2}  ($\rho=0.3, b=0.55$ and varying $c$) shows the existence of one branch of equilibria. In the interval $[2.566, 3.556]$ there exist three equilibria.}
\label{fig:sl0D_bifan}
\end{figure}

\begin{figure}
\centering
	\subfloat[State-Costate]{\includegraphics[width=0.33\textwidth]{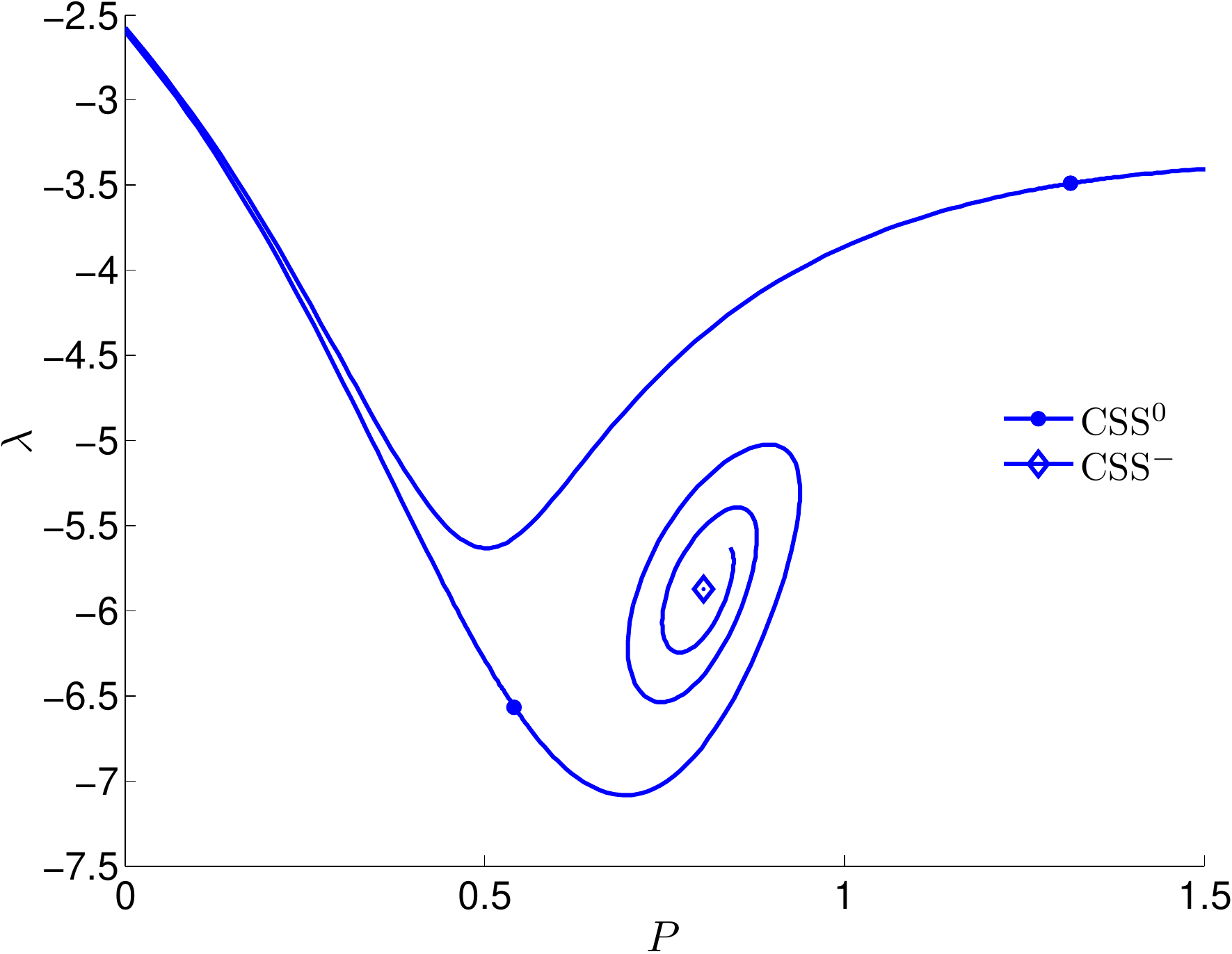}\label{fig:sl_multipleequa}}
	\subfloat[State-Objectivevalue]{\includegraphics[width=0.33\textwidth]{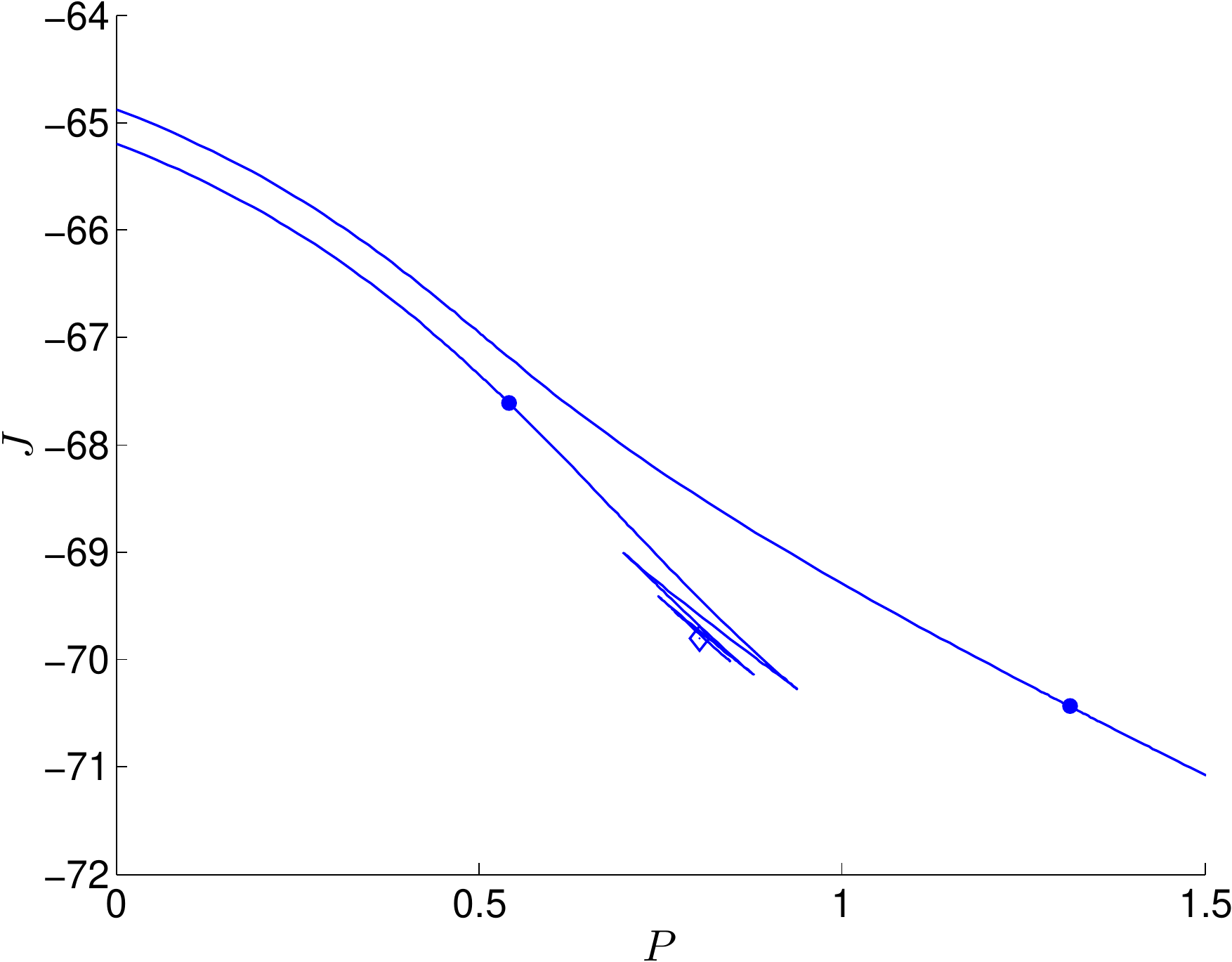}\label{fig:sl_multipleequb}}
	\subfloat[Optimal System]{\includegraphics[width=0.33\textwidth]{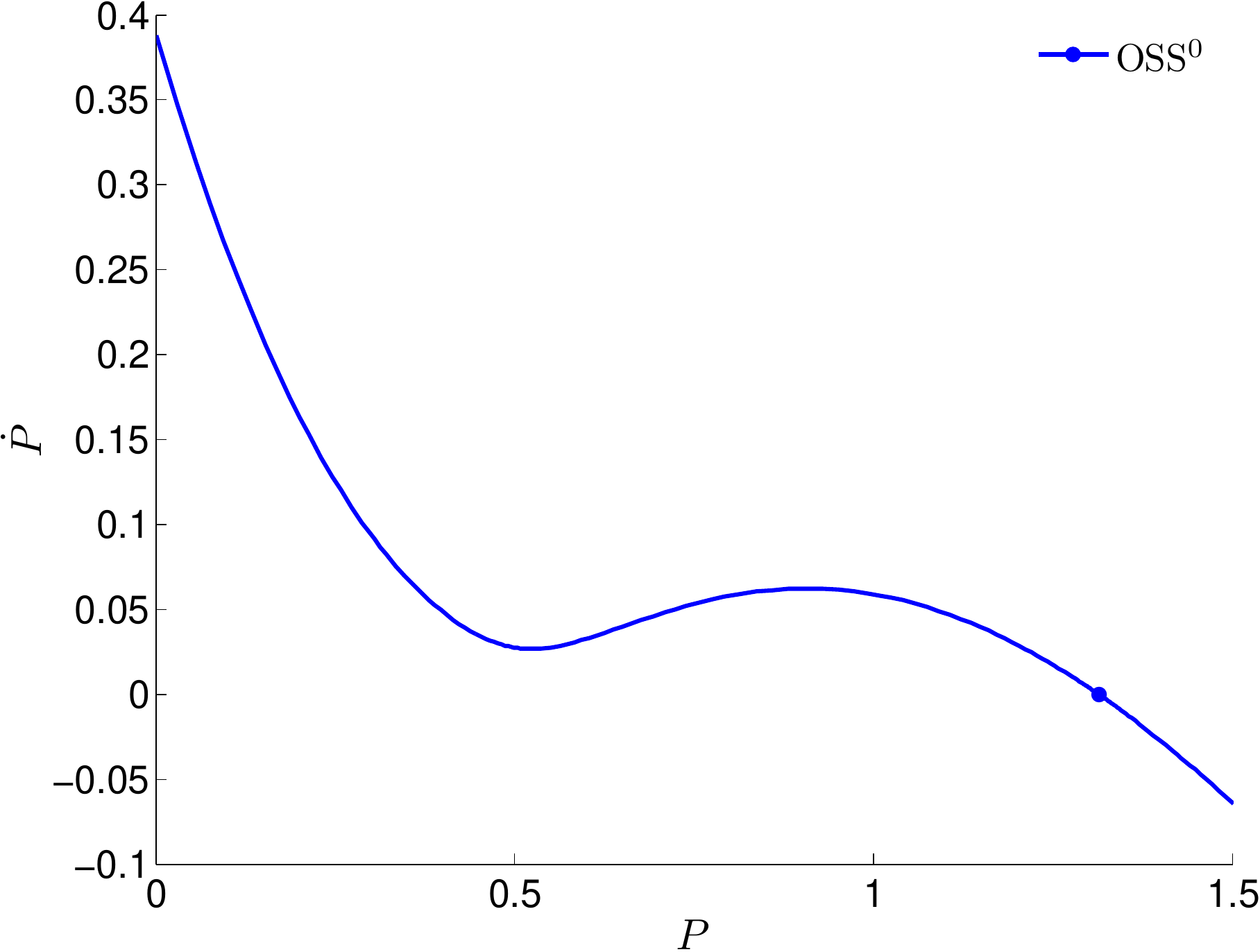}\label{fig:sl_multipleequc}}
\caption[]{For $\rho=0.03,\ c=0.5$ and $b=0.75$ there exist three equilibria in the canonical system \subref{fig:sl_multipleequa}. The optimal system \subref{fig:sl_multipleequc} only consist of one globally optimal equilibrium. In panel \subref{fig:sl_multipleequa} the $\bullet$ denote saddles of the canonical system, the $\diamond$ an unstable focus. In panel \subref{fig:sl_multipleequc} the $\bullet$ denotes the globally stable equilibrium.}
\label{fig:sl_multipleequ}
\end{figure}

\begin{figure}
\centering
	\subfloat[State-Costate]{\includegraphics[width=0.33\textwidth]{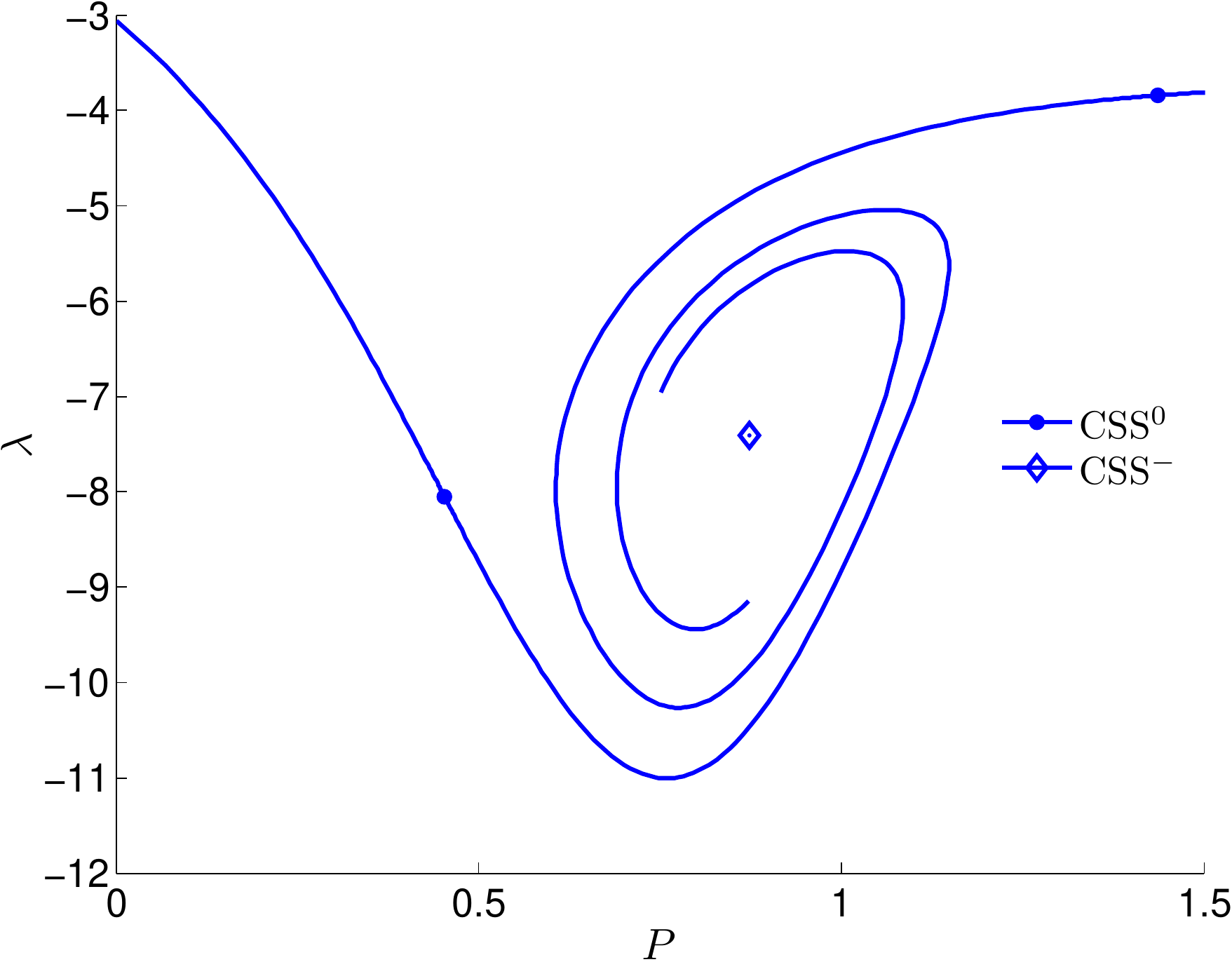}\label{fig:sl_optmultipleequa}}
	\subfloat[State-Objectivevalue]{\includegraphics[width=0.33\textwidth]{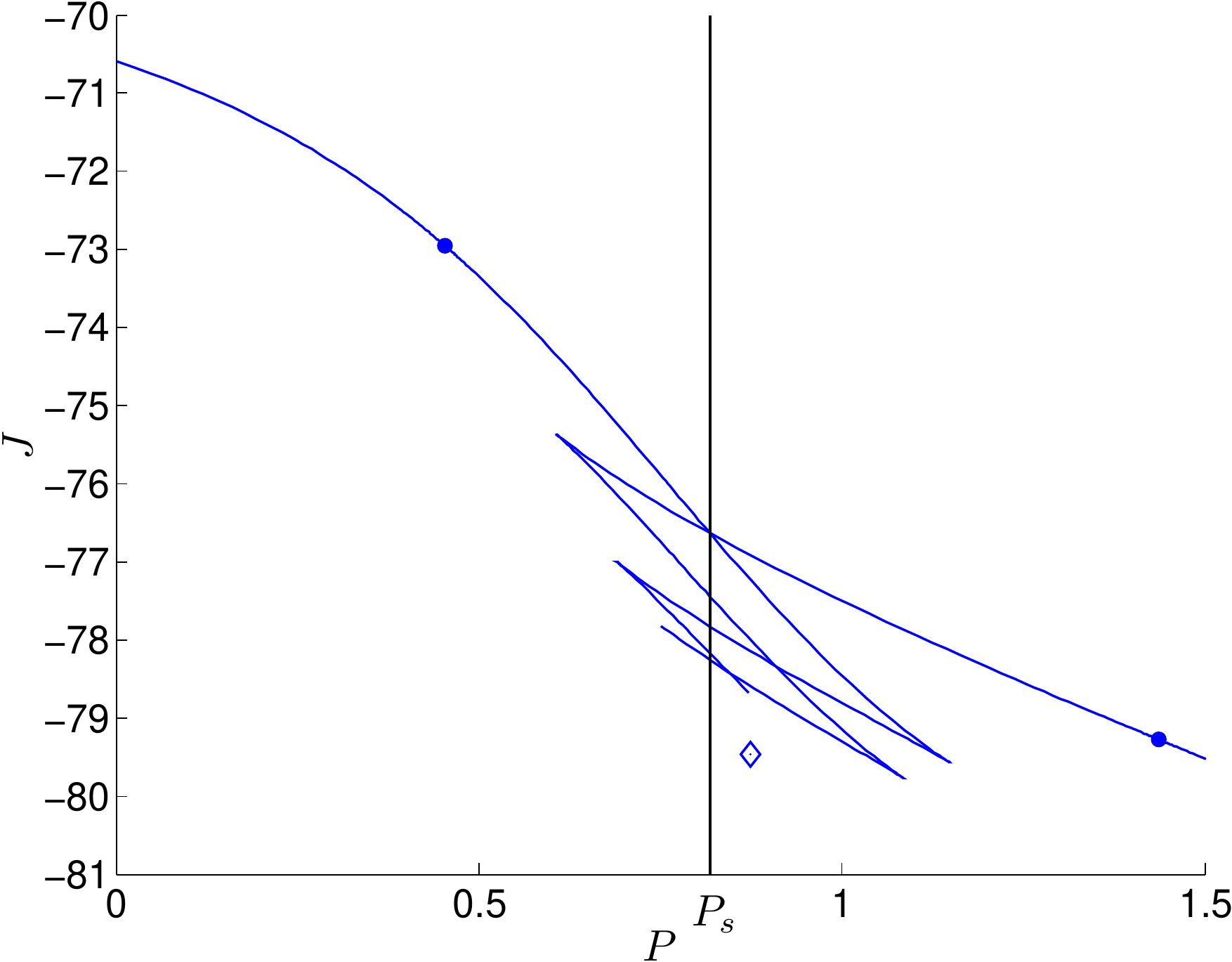}\label{fig:sl_optmultipleequb}}
	\subfloat[Optimal System]{\includegraphics[width=0.33\textwidth]{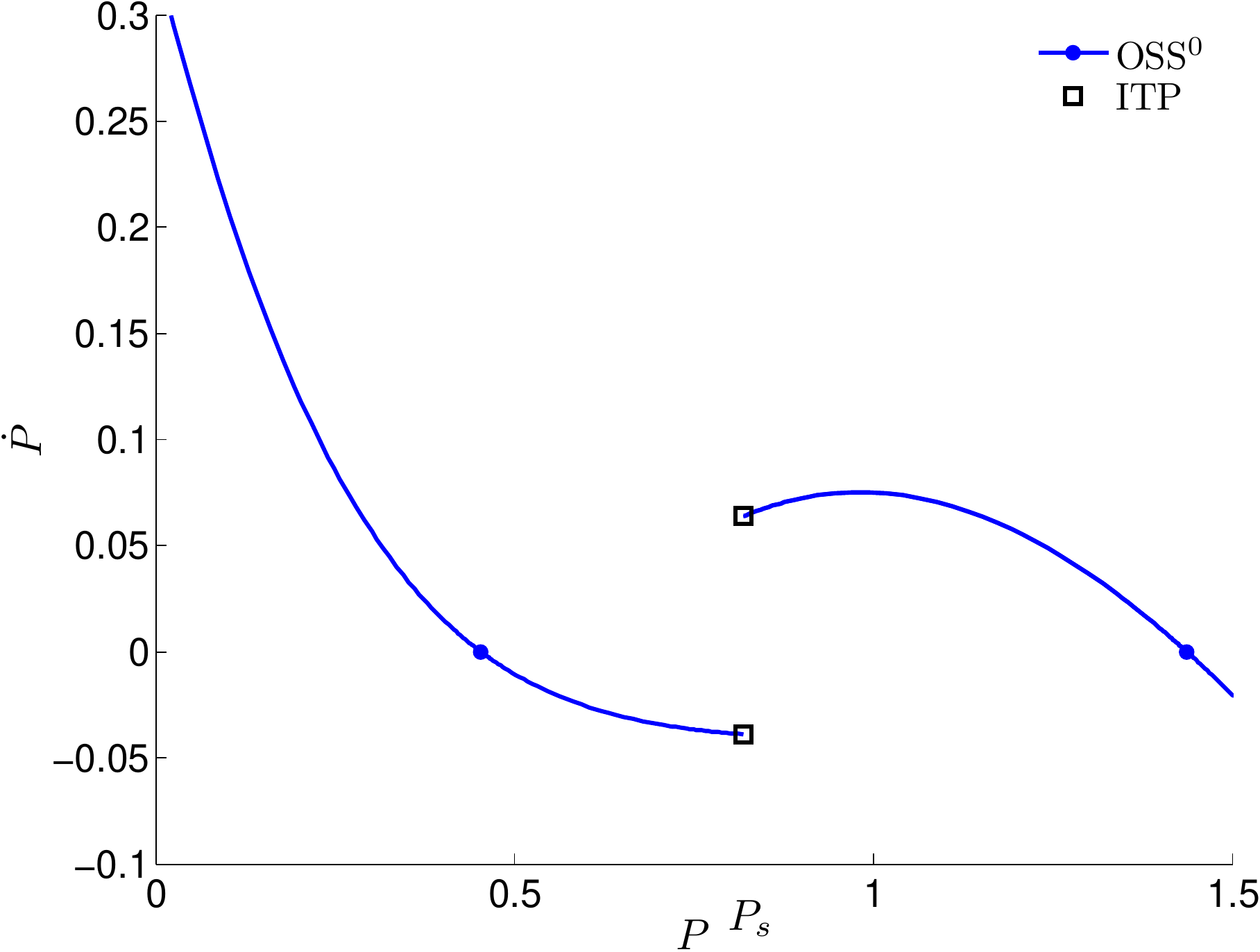}\label{fig:sl_optmultipleequc}}
\caption[]{For $\rho=0.03,\ c=0.5$ and $b=0.65$ there exist three equilibria in the canonical system \subref{fig:sl_optmultipleequa}. The optimal system \subref{fig:sl_optmultipleequc} consists of two locally optimal equilibria. The basins of attraction are separated by an indifference threshold point (Skiba point) $P_I$, with a discontinuous dynamics at $P_I$. In panel \subref{fig:sl_optmultipleequa} the $\bullet$ denote saddles of the canonical system, the $\diamond$ an unstable focus. In panel \subref{fig:sl_optmultipleequc} the $\bullet$ denote the locally stable equilibria.}
\label{fig:sl_optmultipleequ}
\end{figure}

\begin{figure}
\centering
	\subfloat[State-Costate]{\includegraphics[width=0.33\textwidth]{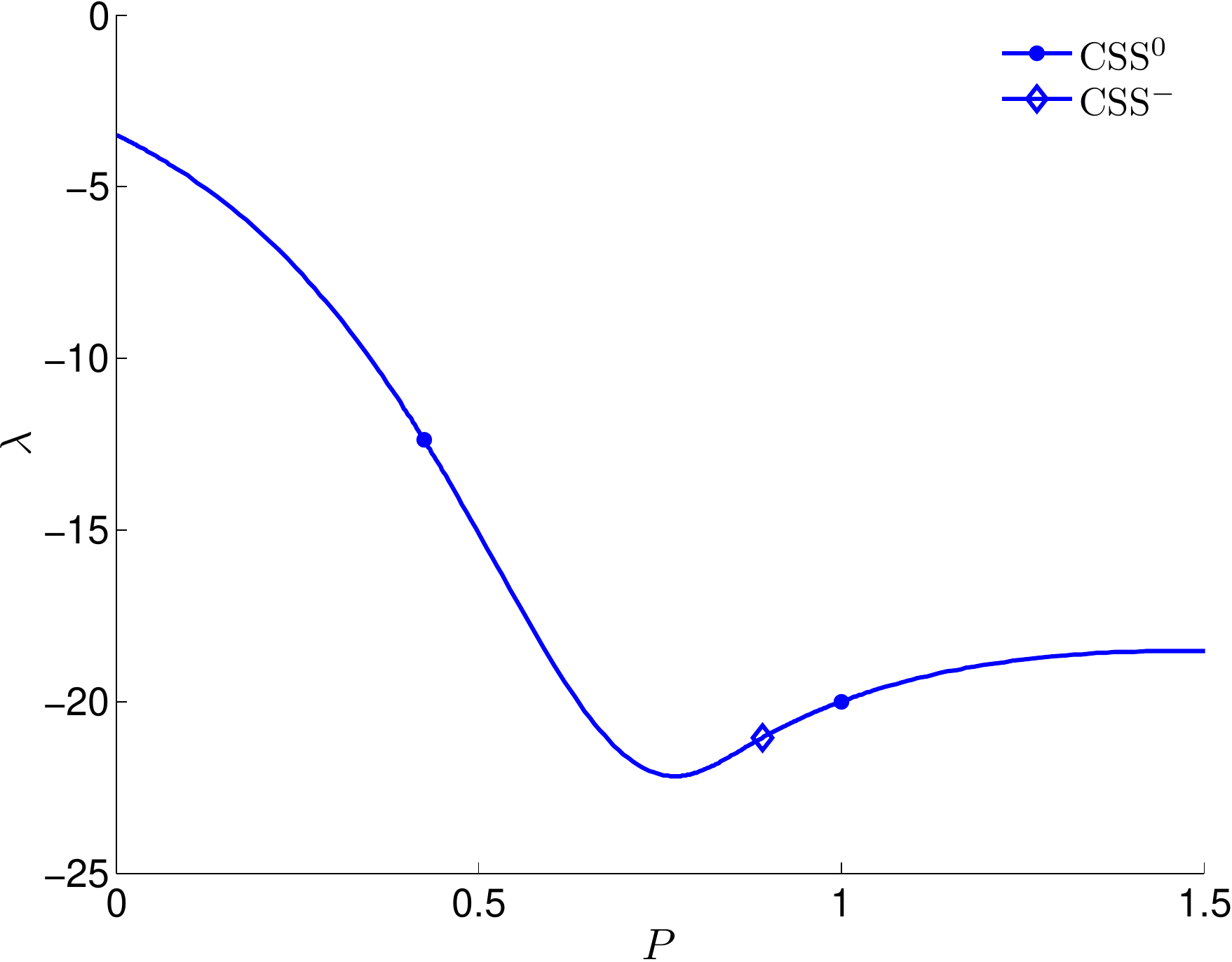}\label{fig:sl_optmultipleuniqueequa}}
	\subfloat[State-Objectivevalue]{\includegraphics[width=0.33\textwidth]{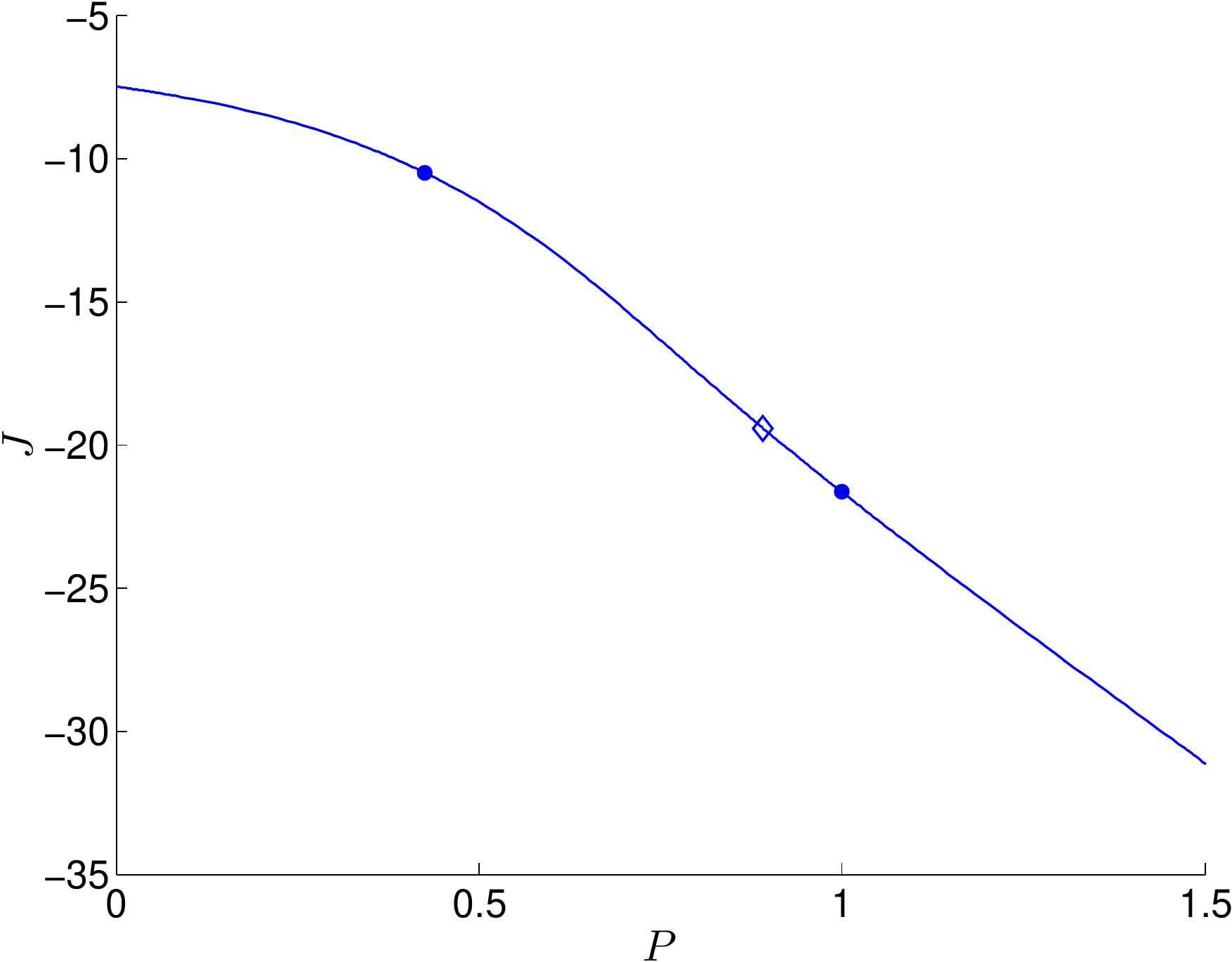}\label{fig:sl_optmultipleuniqueequb}}
	\subfloat[Optimal System]{\includegraphics[width=0.33\textwidth]{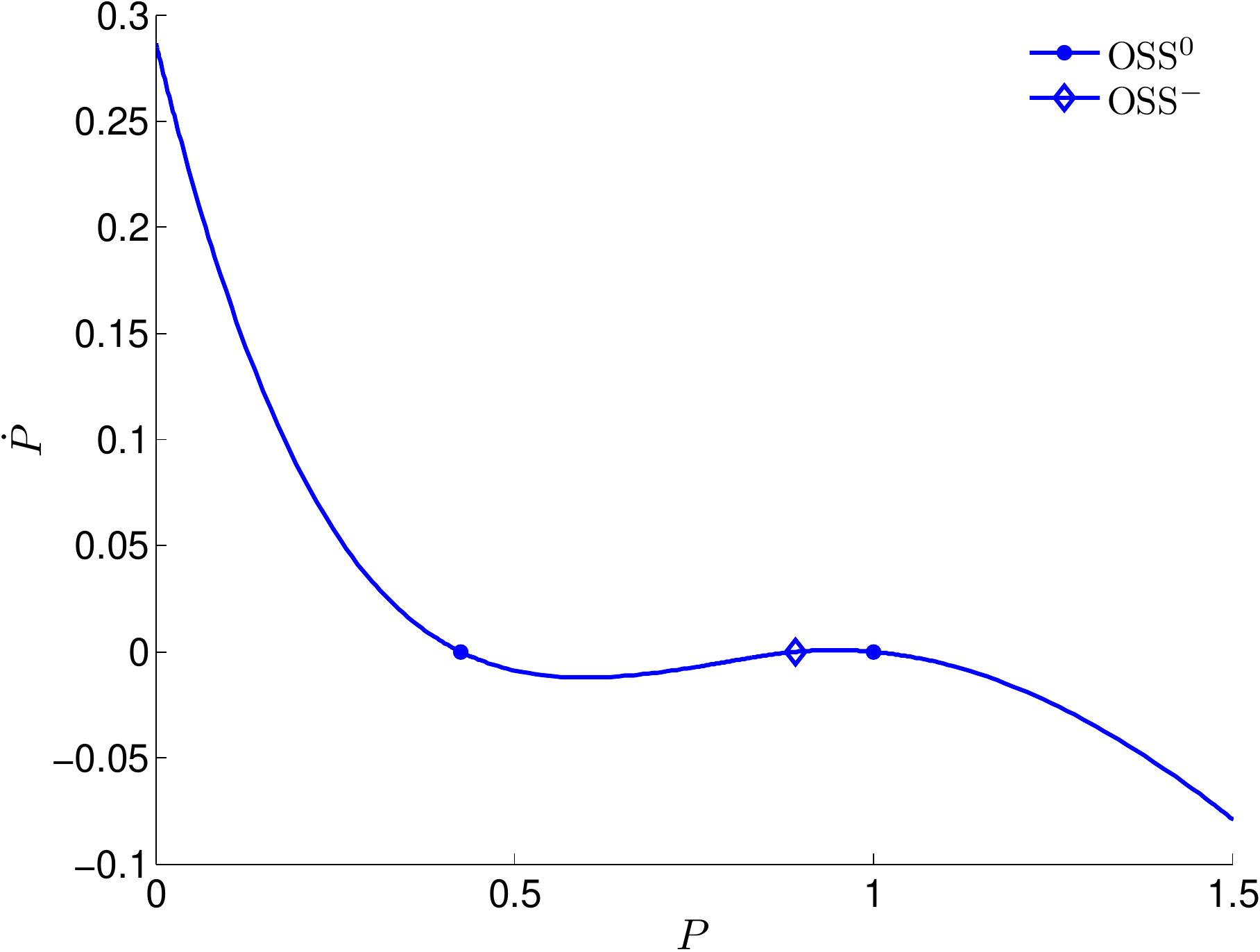}\label{fig:sl_optmultipleuniqueequc}}
\caption[]{For $\rho=0.3,\ c=3.5$ and $b=0.55$ there exist three equilibria in the canonical system \subref{fig:sl_optmultipleequa}. The optimal system \subref{fig:sl_optmultipleequc} consists of two locally optimal equilibria. The basins of attraction are separated by the optimal unstable equilibrium $\hat P_u$. In panel \subref{fig:sl_optmultipleuniqueequa} $\bullet$ denote saddles of the canonical system, $\diamond$ an unstable node. In panel \subref{fig:sl_optmultipleuniqueequc} $\bullet$ denote the locally stable optimal equilibria and $\diamond$ denotes the optimal unstable node.}
\label{fig:sl_optmultipleuniqueequ}
\end{figure}

\section{The shallow lake model with spatial diffusion}
\label{sec:sl_model}
An extension of the shallow lake model \labelcref{eq:sl_model} to the class of spatially distributed models, proposed in \citet{brockxepapadeas2008}, is given by
\begin{subequations}
\label{eq:sldiffusion_model}
\begin{align}
&\max_{u(\cdot,\cdot)}\int_0^\infty\E^{-\rho t}\int_\Omega \ln(u(x,t))-cP(x,t)\!\Dx\Dt\label{eq:sldiffusion_model_obj}\\
&\text{s.t.}\quad\frac{\partial}{\partial t} P(x,t)=u(x,t)-bP(x,t)+\frac{P(x,t)^2}{1+P(x,t)^2}+D\frac{\partial^2}{\partial x^2} P(x,t)\label{eq:sldiffusion_model_dyn}\\
&\quad\evalat{\partial_n P(x,t)}_{\partial\Omega}=0\label{eq:sldiffusion_model_bc}\\
&\quad\evalat{P(x,t)}_{t=0}=P_0(x),\quad x\in\Omega=[-L,L]\subset\R.\label{eq:sldiffusion_model_ic}
\end{align}
\end{subequations}
Contrary to \citet{brockxepapadeas2008} we formulated the problem for Neumann conditions \cref{eq:sldiffusion_model_bc}, the so called zero flux boundary condition, instead of the periodic boundary conditions. From an interpretational point of view we assume the lakes to be located consecutively in a row and there is no flow out of lake. Periodic boundary conditions refer to a ring of lakes. \citeauthor{brockxepapadeas2008} argue that they use periodic boundary conditions to exclude effects induced by the conditions at the end points. Anyhow, since this point does not touch our argument that it is necessary to analyze the global behavior of the optimally controlled system, we changed the model formulation in this respect.

Using the (\FDM) discretization proposed in \cref{sec:ModelsOfSpatialDimension11DModel} we find
\begin{subequations}
\label{eq:dslocm}
\begin{align}
& \max_{u_0(\cdot),\ldots,u_N(\cdot)}\left\{\int_0^\infty\E^{-rt}G(P_0(t),\ldots,P_N(t),u_0(t),\ldots,u_N(t))\,\Dt\right\}\\
\text{s.t.}\ &\dot P_i(t)=u_i(t)-bP_i(t)+\frac{P_i(t)^2}{1+P_i(t)^2}+\tilde D\left(P_{i-1}(t)-2P_i(t)+P_{i+1}(t)\right)\\
& P_1(t)-P_{-1}(t)=P_{N+1}(t)-P_{N-1}(t)=0,\quad t\ge0\\
& P_i(0) = P_{i,0}>0\label{eq:dslocm3}
\shortintertext{with}
&x_i=\frac{i}{N},\quad i=0,\ldots,N,\quad \tilde D\defin D\left(\frac{N}{2L}\right)^2\notag\\
&P_i(t)\defin P(x_i,t),\quad u_i(t)\defin u(x_i,t)\notag\\
&P_i\defin P(x_i,\cdot),\quad u_i\defin u(x_i,\cdot)\notag\\
&G(P_0,\ldots,P_N,u_0,\ldots,u_N)\defin\frac{1}{N}\sum_{i=1}^{N-1}g(P_i,u_i)+\frac{g(P_0,u_0)+g(P_N,u_N)}{2}\notag\\
&g(P_i,u_i)\defin\ln(u_i)-cP_i^2\notag
\end{align}
\end{subequations}
Applying \PMAXP\ on \crefsys{eq:discretization} yields the canonical system
\begin{subequations}
\label{eq:dslocmcansys}
\begin{align}
	&\dot P_i(t)=u^\ext_i(t)-bP_i(t)+\frac{P_i(t)^2}{1+P_i(t)^2}+\diffcoeff^{(P)}_i(t)\\
&\dot\lambda_i(t)=c_iP_i(t)+\lambda_i(t)\left(\rho+b-\frac{2P_i(t)}{\left(1+P_i(t)^2\right)^2}\right)-\diffcoeff^{(\lambda)}_i(t)\\
& P_i(0) = P_{i,0}>0,\quad i=0,\ldots,N\\
	&u_i^\ext(t)=\begin{cases}
	-\dfrac{1}{2\lambda_i(t)} & i=0,N\\
	-\dfrac{1}{\lambda_i(t)} & i=1,\ldots,N-1
	\end{cases}\\
& c_i\defin\begin{cases}
	c & i=0,N\\
	2c &i=1,\ldots,N-1
\end{cases}\notag\\
&\diffcoeff^{(P)}_i(t)\defin\begin{cases}
	2\tilde D(P_1(t)-P_0(t)) & i=0\\
	\tilde D(P_{i-1}(t)-2P_i(t)+P_{i+1}(t)) & i=1,\ldots,N-1\\
	2\tilde D(P_{N-1}(t)-P_N(t)) & i=N
\end{cases}\notag\\
&\diffcoeff^{(\lambda)}_i(t)\defin\begin{cases}
	\tilde D(\lambda_1(t)-2\lambda_0(t)) & i=0\\
	\tilde D(2\lambda_0(t)-2\lambda_1(t)+\lambda_2(t)) & i=1\\
	\tilde D(\lambda_{i-1}(t)-2\lambda_i(t)+\lambda_{i+1}(t)) & i=2,\ldots,N-2\\
	\tilde D(\lambda_{N-2}(t)-2\lambda_{N-1}(t)+2\lambda_N(t)) & i=N-1\\
	\tilde D(\lambda_{N-1}(t)-2\lambda_N(t)) & i=N
\end{cases}\notag
\end{align}
\end{subequations}
In \cref{sec:TheUsageOfOCMAT} the details for an implementation of the \modelcref{eq:dslocm} in \OCMAT\ are explained.

\subsection{Equilibria of the canonical/optimal system}
\label{sec:EquilibriaOfTheCanonicalSystem}
There is an intimate relation between the \css{} of the canonical system \cref{eq:sl_model_cansys} and \fcss{} of the canonical system \cref{eq:dslocmcansys}
\begin{corollary}
\label{cor:fosschar}
Let $(\hat P^d, \hat u^d)\in\R^{2N+2}$ be \foss{} then $(\hat P^d_0, 1/(2u_0))$ is an equilibrium of the canonical system \cref{eq:sl_model_cansys}.
\end{corollary}
\begin{corollary}
\label{cor:fosscharII}
Let $(\hat P, \hat\lambda)$ be an equilibrium of the canonical system \cref{eq:sl_model_cansys}. Then $\hat P^d\defin(\hat P,\ldots,\hat P)$ and $\hat\lambda^d\defin(2\hat\lambda,\hat\lambda,\ldots,\hat\lambda,2\hat\lambda)$ is a \fcss.
\end{corollary}
\begin{corollary}
\label{cor:fossspp}
Let $(\hat P,\hat\lambda)$ be a saddle of the canonical system \cref{eq:sl_model_cansys}. Then for $\tilde D$ small enough $(\hat P^d, \hat\lambda^d)\in\R^{2N+2}$ defined in \cref{cor:fosscharII} is \fcssspp. 
\end{corollary}
\hcss{} can in general not be calculated analytically therefore we have to resort to numerical methods. Using \cref{cor:fosscharII} we can start a bifurcation analysis of \cref{eq:dslocmcansys} with an equilibrium of \cref{eq:sl_model_cansys}. The \hcss{} emerge from branching points of the bifurcation curve. In \citet{grassuecker2015} the according bifurcation analysis is done using \PDETOPATH, a \MATL\ package for the bifurcation analysis of elliptic \PDE s, see \citet{ueckeretal2014} and \citet{dohnaletal2014}. Since the actual \modelcref{eq:sl_model} is a 0D optimal control model with a finite number of states ($N+1$) the bifurcation analysis of \cref{eq:dslocmcansys} is done using a modified version of \CLMATCONT\footnote{This modified version is available from the author.}.

We analyzed the two different scenarios specified in \cref{tb:base_par}. 
\paragraph{First Scenario}
\label{sec:FirstScenario}
The according bifurcation analysis with respect to $b$ is depicted in \cref{fig:sld_bifdiag1}. The black curves represent the bifurcation curves of the \fcss. These curves exhibit the same shape as the corresponding bifurcation curves for the 0D model in \cref{fig:sl0D_bifan1}. For the lower branch we additionally find four branching points $\circ$, where the branches of the \hcss{} emanate (red, green, magenta and cyan). Along the bifurcation curves of the \hcss{} we find additional branching points and calculated the according bifurcation curves (brown, dark green and orange). Thus, for $b=0.65$ we find in total two \fcssspp{} (correspondig to the oligotrophic and eutrophic equilibrium in the 0D model), one \fcssnspp{}, thirteen \hcssnspp{} and one \hcssspp{}. In fact the brown and dark green branch consists of two distinct bifurcation curves with equilibria that are spatially symmetric.

\paragraph{Second Scenario}
\label{sec:SecondScenario}
The according bifurcation analysis with respect to $c$ is depicted in \cref{fig:sld_bifdiag2}. The black bifurcation curve of the \fcss{} consists of one branch, exhibiting two fold bifurcations and six branching points $\circ$. These six branching points are connected by three bifurcation curves of the \hcss{} (red, magenta and green). The red and green curve consist of two spatially symmetric branches. From the magenta \hcss{} bifurcation curve two further \hcss{} branches emanate (brown and orange).

In \cref{sec:scenarion2} we consider two specific cases for $c=3.5$ and $c=3.0825$. In the first case there exist two \fcssspp{} (correspondig to the oligotrophic and eutrophic equilibrium in the 0D model), one \fcssnspp{}, ten \hcssnspp{} and two \hcssspp. In the latter case there exist two \fcssspp{}, one \fcssnspp{} and four \hcssnspp{}.
\begin{figure}
\centering
	\subfloat[First Scenario $b$, $\rho=0.03, c=0.5$]{\includegraphics[width=0.45\textwidth]{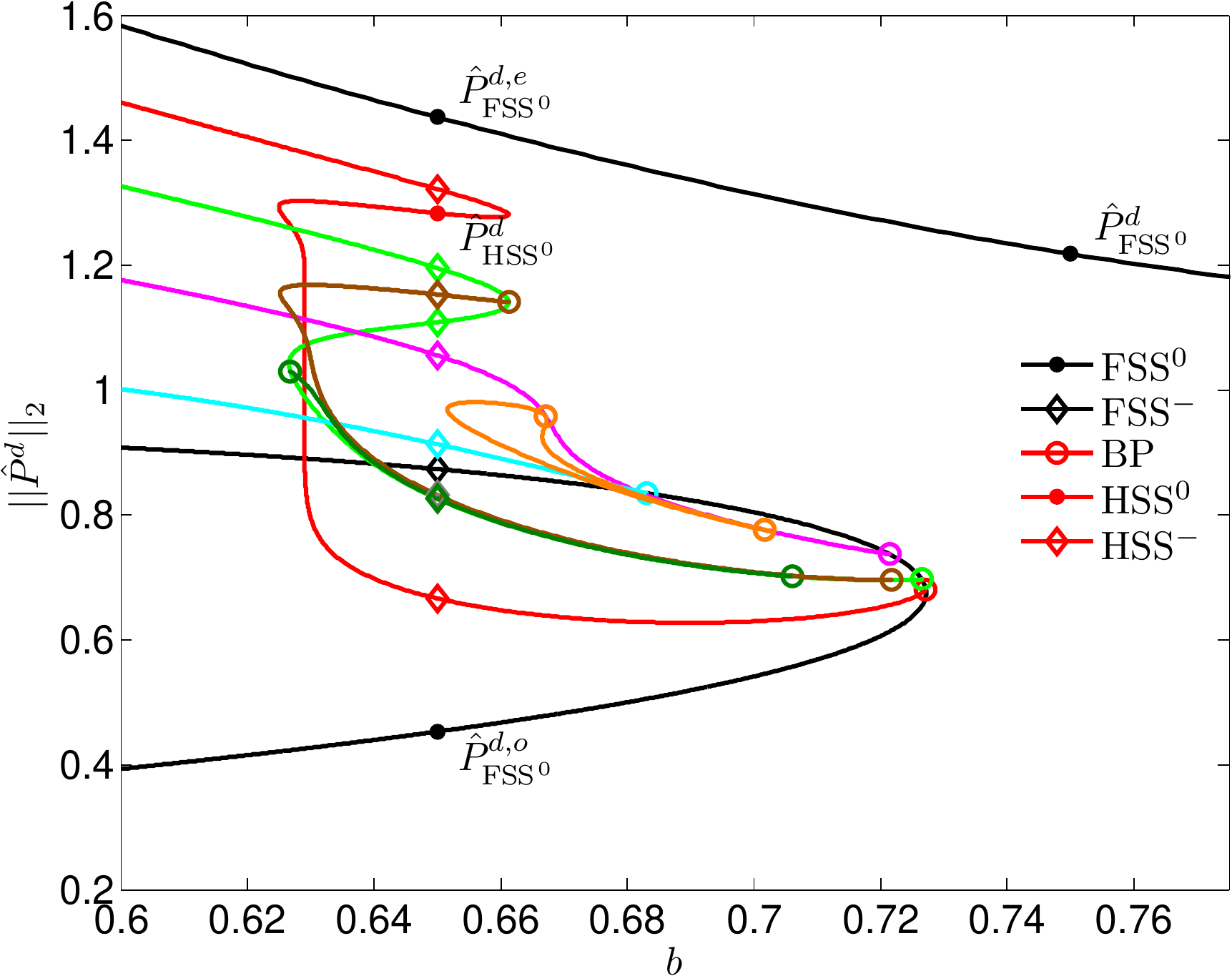}\label{fig:sld_bifdiag1}}\hfill
	\subfloat[Second Scenario $c$, $\rho=0.3, b=0.55$]{\includegraphics[width=0.45\textwidth]{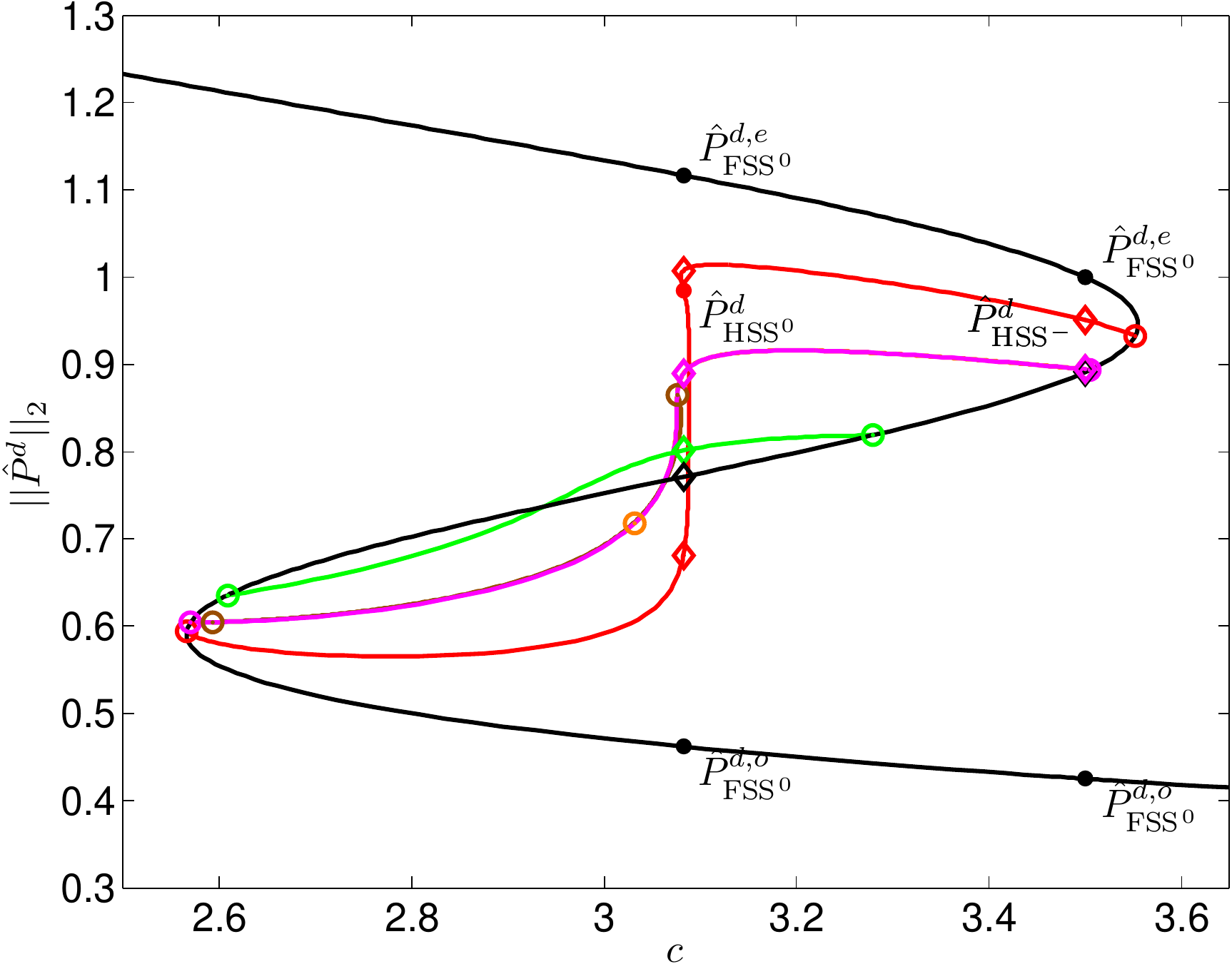}\label{fig:sld_bifdiag2}}
\caption[]{A bifurcation analysis for the canonical system \cref{eq:dslocmcansys} is depicted in \subref{fig:sld_bifdiag1} and \subref{fig:sld_bifdiag2}. The symbols: $\bullet$ denote equilibria satisfying the \spp, $\diamond$ equilibria not satisfying the \spp{} and $\circ$ branching points of the equilibria manifold. The FOSS curves are black and the HOSS curves are colored. Along the red HOSS curve two fold bifurcations occur, the HOSS candidates between these fold bifurcations satisfy the \spp.}
\label{fig:sld_bifdiag}
\end{figure}

\subsection{First scenario optimal solutions exemplified on two cases}
\label{sec:scenarion1}
Most of the results for the specific cases $b=0.75$ and $b=0.65$ have already been discussed in \citet{grassuecker2015}. Therefore we only give a brief summary and concentrate on some aspects that were easier to get using \OCMAT.

For $b=0.75$ there exists a single \foss{}, structurally reproducing the result of the according 0D model.

For $b=0.65$ the main results are
\begin{itemize}
	\item None of the \hcss{} are optimal.
	\item The two \fcssspp{} are optimal.
	\item The according basins of attractions are separated by an indifference threshold manifold.
	\item In \citet{grassuecker2015} we computed a homogeneous and patterned indifference threshold point. The homogeneous indifference threshold point coincides with the value of the indifference threshold point in \modelcref{eq:sl_model}.   
\end{itemize}
Next we explain in detail the calculation of a solution converging to an equilibrium satisfying \spp. Subsequently the necessary steps for the detection of an indifference threshold point are explained. Details for the calculation in \OCMAT{} can be found in \cref{sec:TheUsageOfOCMAT}.

\subsection{A locally optimal patterned equilibrium}
\label{sec:detectitp}
To prove that \hcssspp{} $(\phcssspp,\lhcssspp)$, see \cref{fig:sld_bifdiag1}, is an optimal equilibrium we have to show that there exist no other solution $(P^d(\cdot), u^d(\cdot))$ with $P^d(0)=\phcssspp$ yielding a larger objective value. There exist two other equilibria satisfying \spp, namely the two \fcssspp{}, the oligotrophic $(\pfocssspp,\lfocssspp)$ and the eutrophic \fcssspp{} $(\pfecssspp,\lfecssspp)$ (cf.~\cref{fig:sld_bifdiag1}). For each of these equilibria there may exist a path with $P^d(0)=\phcssspp$ and converging to the oligotrophic/eutrophic equilibrium. To determine these paths we consider the corresponding homotopy problems \cref{eq:homcontbvp} with $x_1$ replaced by $\phcssspp$, starting at the equilibrium $(\pfocssspp,\lfocssspp)$ and $(\pfecssspp,\lfecssspp)$, respectively.

\subsubsection{Comparison with equilibrium solution}
\label{sec:ComparisonWithEquilibriumSolution}
The result of these computations is depicted in \cref{fig:foss2hossb05}. \Cref{fig:foss2hossb051} illustrates the ``embedding'' \cref{eq:homcontbvp2}. The green manifold proceed from the flat eutrophic states $\pfecssspp$ and the blue manifold from the flat oligotrophic states $\pfocssspp$ to the patterned states $\phcssspp$. During the continuation of $\contpar$ from zero to one the initial states lie in these manifolds. In \cref{fig:foss2hossb052,fig:foss2hossb053} the states and their norm of the solution paths for $\contpar=0.5$ are shown (green and blue). Additionally the states and norm of the equilibrium solution $\phcssspp$ are depicted. 

\Cref{fig:foss2hossb054,fig:foss2hossb055,fig:foss2hossb057} display the final results for $\kappa=1$ also including the figure with the control paths. Moreover, \cref{fig:foss2hossb056} and \cref{fig:foss2hossb057} reveal that the values of the costates and hence the controls of the three solutions with $P^d(0)=\phcssspp$ are different.

\Cref{fig:foss2hossb058} show the slice manifolds (see~\cref{def:slicemanifold}) in the state-costate space. Each marker $\times$ denotes a continuation step. 

In \cref{fig:foss2hossb057} the objective values along the slice manifolds are plotted, i.e. the objective values are plotted against the (norm) of the initial points. We see that the final solutions for $\kappa=1$ both yield a higher objective value and the eutrophic solution converging to \fecssspp{} dominates the other solutions. On the other hand the gradients of the curves suggest that they eventually intersect. Such an intersection point characterizes an indifference threshold point (cf.~\cref{fig:sl_optmultipleequb} for the 0D model).

\subsubsection{Detection and Continuation of an indifference threshold point \itp}
\label{sec:DetectionOfAnIndifferenceThresholdPoint}
To find a possible intersection point of the objective values along the slice manifolds we have to assure that the slice manifolds are comparable and have to check if they intersect (see~\cref{def:consslicemanifold}). The slice manifolds depicted in \cref{fig:foss2hossb058} are, e.g., not comparable, simply because the manifolds
\begin{align*}
	&\{\pfocssspp+(1-\alpha_1)(\phcssspp-\pfocssspp):\ \alpha_1\in\R\}\\
	&\{\pfecssspp+(1-\alpha_2)(\phcssspp-\pfecssspp):\ \alpha_2\in\R\}.
\end{align*}
are different. See \cref{fig:foss2hossb051}), where the green and blue curves depict (parts) of these manifolds. Only if slice manifolds are comparable and have a non-empty intersection this means that the solutions corresponding to the intersection points of the slice manifold start at the same initial states and hence their objective values can be compared.

Anyhow, since the solution starting at $\phcssspp$ and converging to \fecssspp{} is known we can start the homotopy \bvpcref{eq:homcontbvp} with
\begin{equation*}
	X(0)=\phcssspp+(1-\contpar)(\pfocssspp-\phcssspp).
\end{equation*}
In that case the manifolds
\begin{align*}
	&\{\pfocssspp+(1-\alpha_1)(\phcssspp-\pfocssspp):\ \alpha_1\in\R\}\\
	&\{\phcssspp+(1-\contpar)(\pfocssspp-\phcssspp):\ \contpar\in\R\}.
\end{align*}
trivially coincide and the slice manifolds are comparable, see \cref{fig:fossskiba1}. Moreover the slice manifolds intersect and and the corresponding objective values curves intersect at an indifference threshold point $P^d_{I,1}$ (cf.~\cref{fig:fossskiba4}). 

In \cref{fig:fossskiba2} and \cref{fig:fossskiba5} the corresponding state $P_{e,o}^d(\cdot)$ and control paths $u_{e,o}^d(\cdot))$ are shown.

Repeating the same procedure for the three two \fcssspp{} and \hcssspp{} we find that \hcssspp{} is not optimal and again we find a further indifference threshold point $P^d_{I,2}$. The according solutions are depicted in \cref{fig:fossskiba3} and \cref{fig:fossskiba6}.

It is an interesting question to see how we can find indifference threshold points ``between'' $P^d_{I,1}$ and $P^d_{I,2}$. Let us recall that we found the indifference threshold points by the intersection of two $n$ dimensional manifolds (with $n=N+1$ the number of states $P^d$) in the $n+1$ dimensional space $(P^d,J(P^d(0)))$. Generically this yields an $n-1=N$ dimensional manifold. Already for the discretization $N=50$ its dimension is too large to recover the entire indifference threshold manifold (in the original \PDE{} problem the dimension actually increases to infinity). Anyhow, we can e.g. search for indifference threshold points along the linear connection $P^d_{I,2}+(1-\contpar)(P^d_{I,1}-P^d_{I,2})$. The result is depicted in an animation embedded into \cref{fig:fossskiba5}. The according \BVP{} for the numerical solution of this problem is presented in \cref{sec:ContinuationOfAnIndifferenceThresholdPoint}.

\begin{figure}
\centering
	\subfloat[Manifolds of initial distributions]{\includegraphics[width=0.33\textwidth]{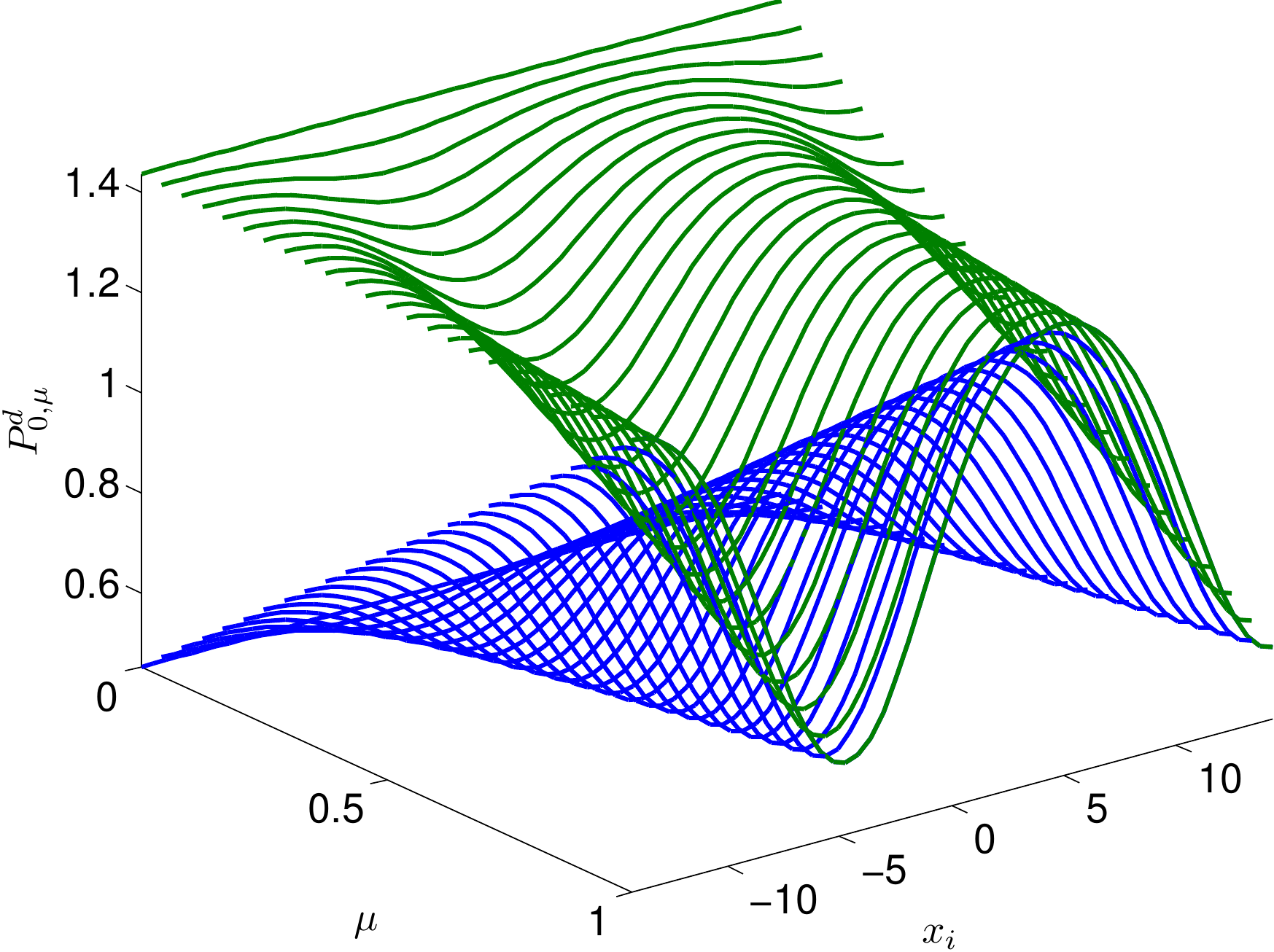}\label{fig:foss2hossb051}}\hfill
	\subfloat[State paths at $\contpar=0.5$]{\includegraphics[width=0.33\textwidth]{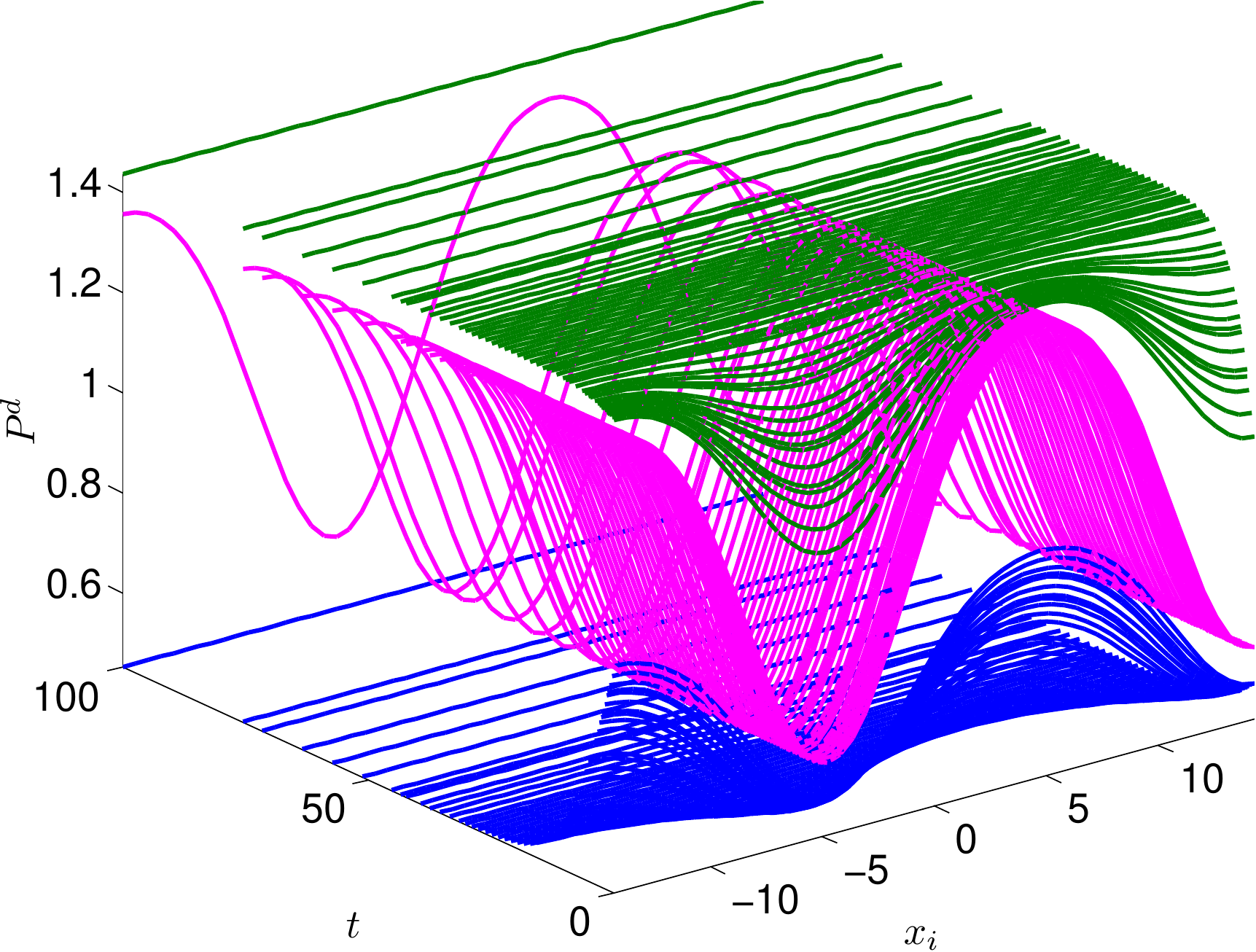}\label{fig:foss2hossb052}}\hfill
	\subfloat[Norm of the states at $\contpar=0.5$]{\includegraphics[width=0.33\textwidth]{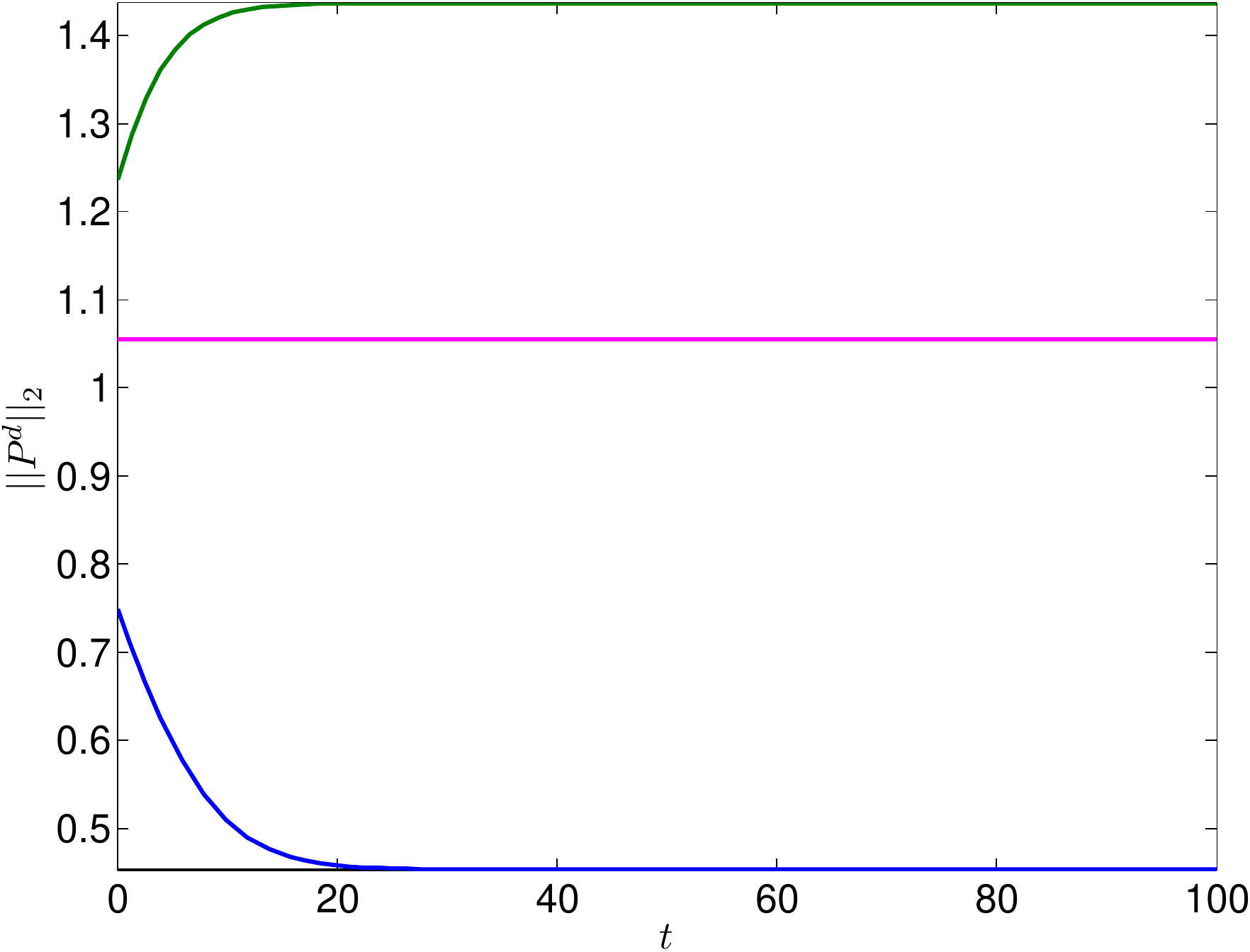}\label{fig:foss2hossb053}}\\
	\subfloat[State paths at $\contpar=1$]{\includegraphics[width=0.33\textwidth]{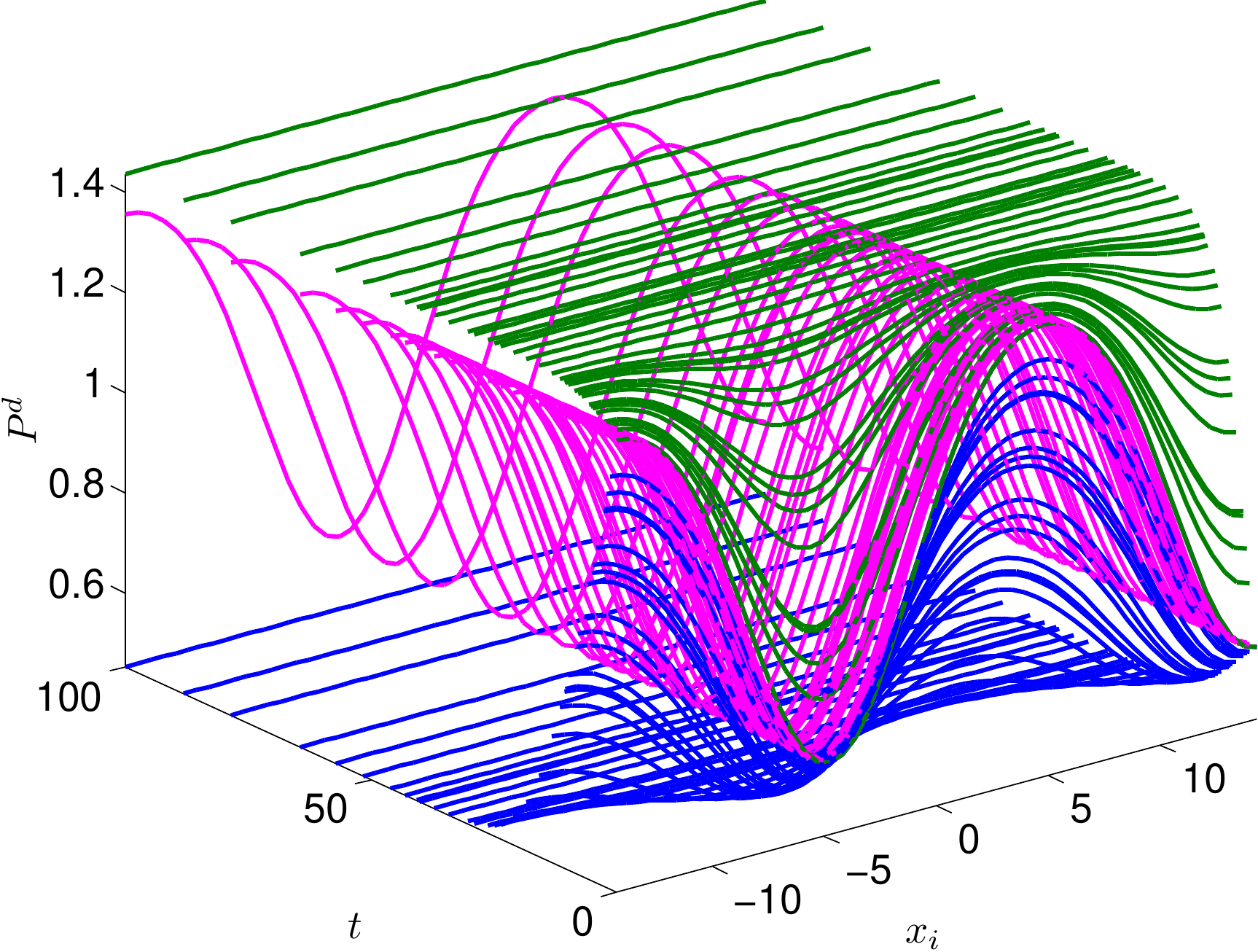}\label{fig:foss2hossb054}}\hfill
	\subfloat[Control paths at $\contpar=1$]{\includegraphics[width=0.33\textwidth]{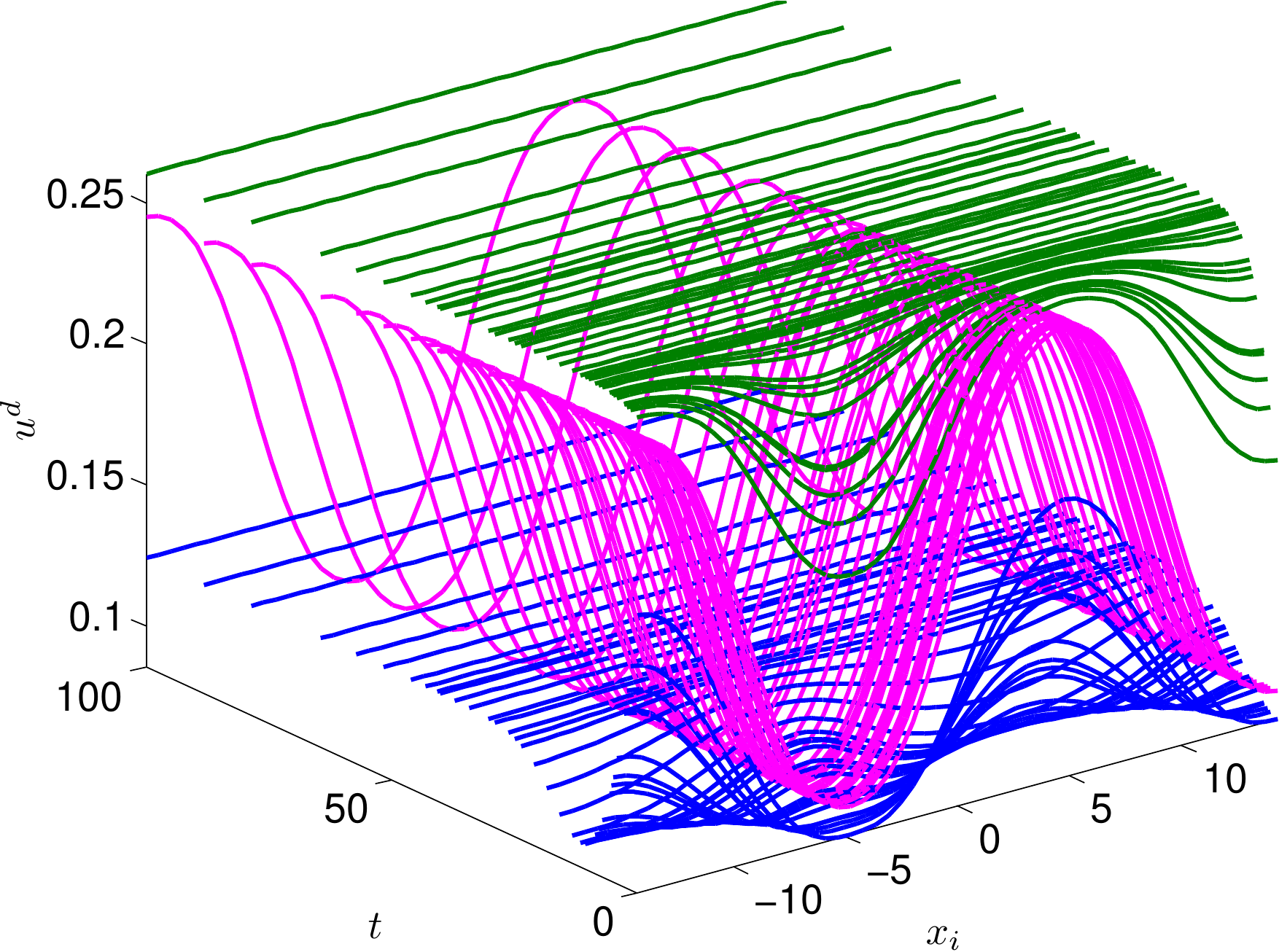}\label{fig:foss2hossb057}}\hfill
	\subfloat[Norm of the states at $\contpar=1$]{\includegraphics[width=0.33\textwidth]{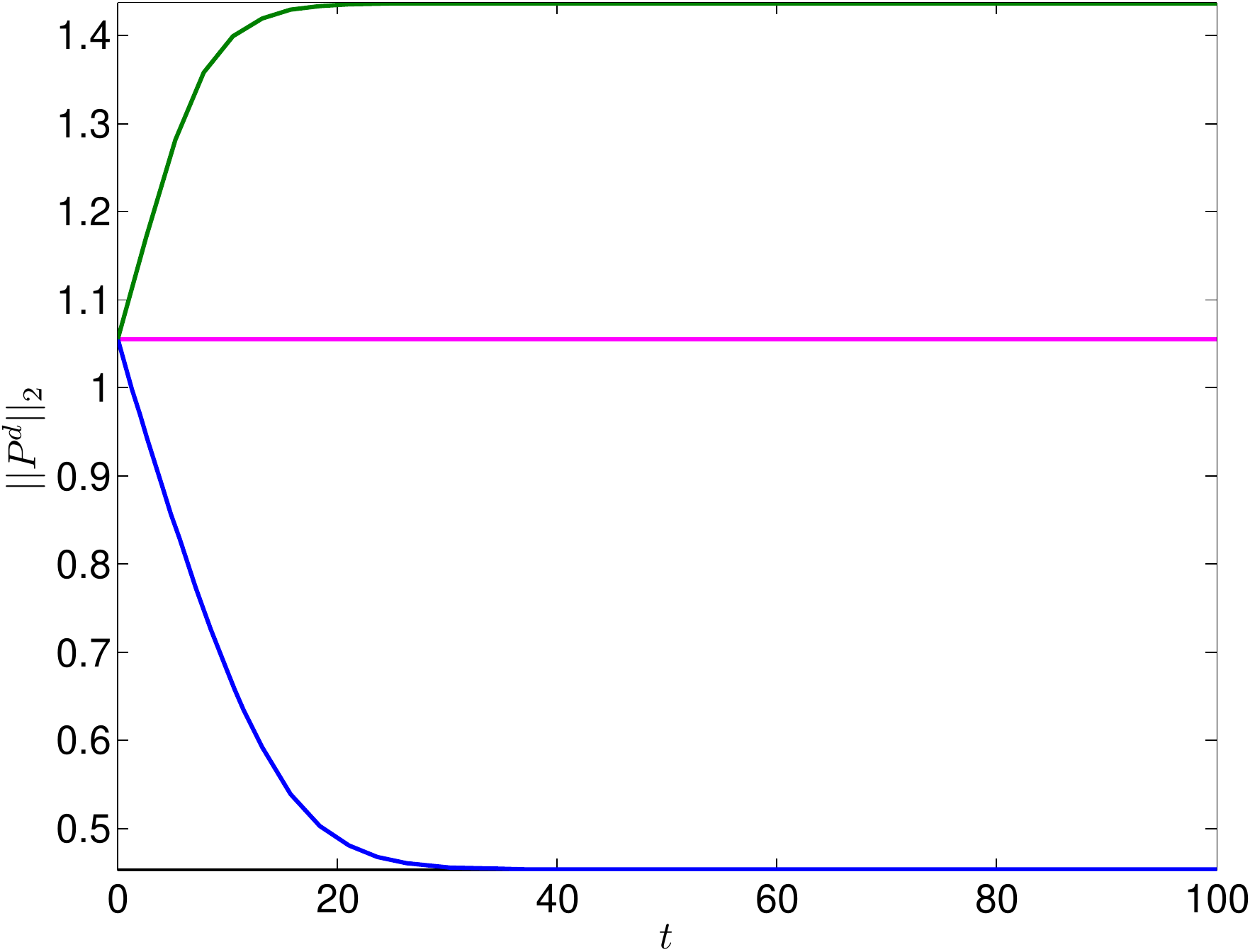}\label{fig:foss2hossb055}}\\
	\subfloat[Norm of the costates at $\contpar=1$]{\includegraphics[width=0.33\textwidth]{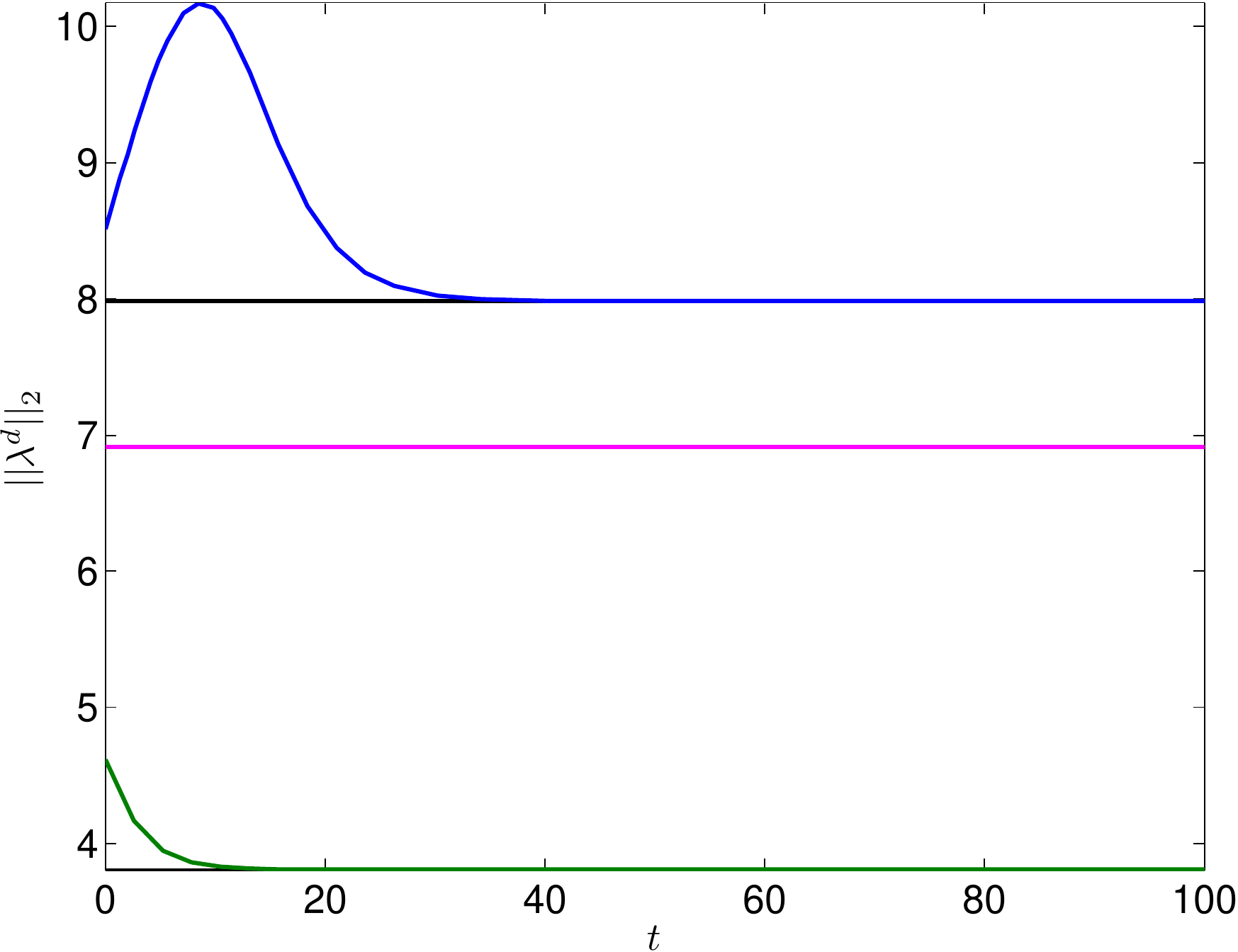}\label{fig:foss2hossb056}}\hfill
	\subfloat[Slice manifolds]{\includegraphics[width=0.33\textwidth]{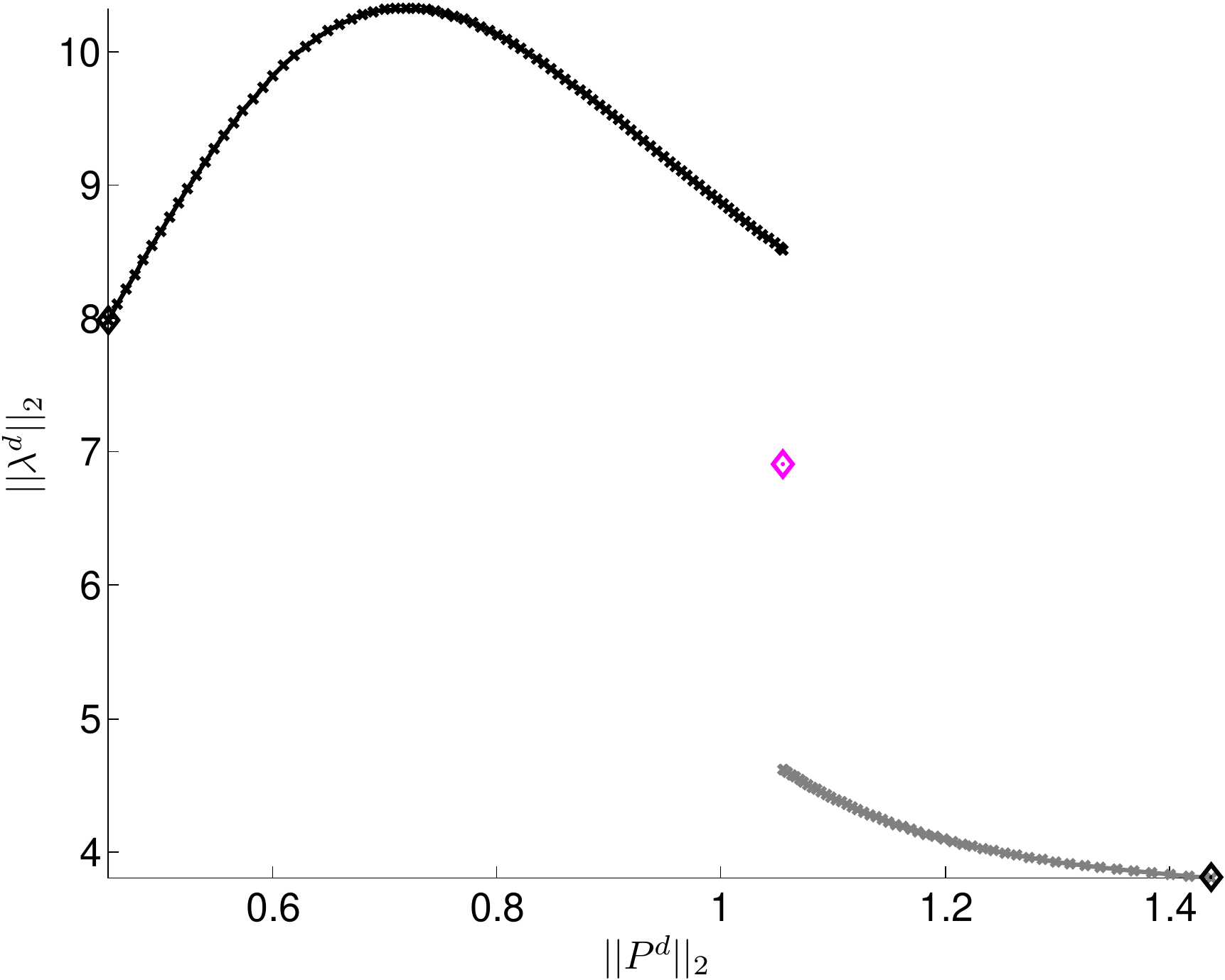}\label{fig:foss2hossb058}}\hfill
	\subfloat[Objective values at slice manifolds]{\includegraphics[width=0.33\textwidth]{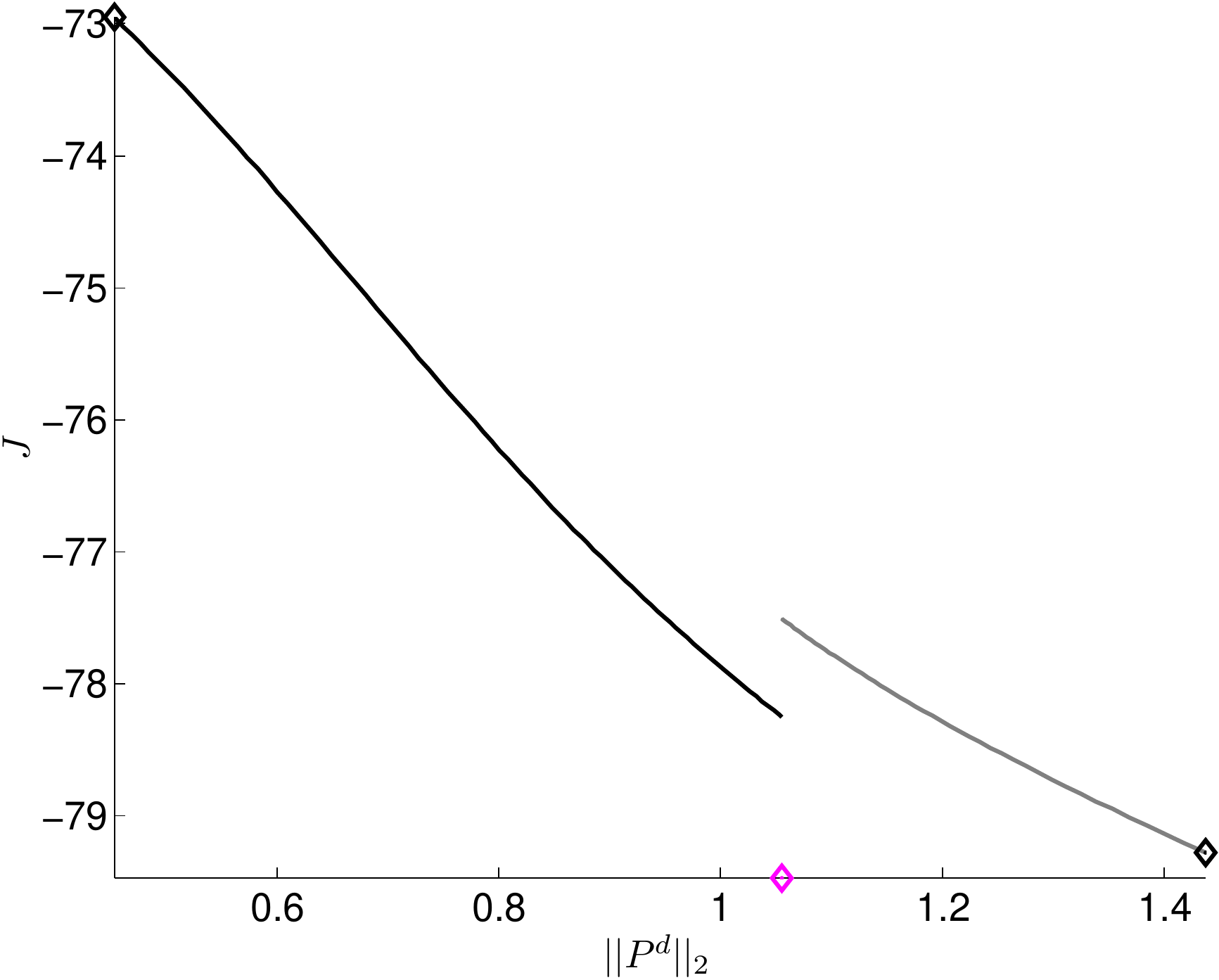}\label{fig:foss2hossb059}}
\caption[]{This figure presents the steps of the continuation process to find solutions starting at $P^d(0)=\phcssspp$ and converging to the oligotrophic and eutrophic \fcssspp. The colors refer to the eutrophic (blue), oligotrophic (green) and patterned (magenta) solutions. In \subref{fig:foss2hossb051} the manifolds of initial distributions, that are passed through the continuation, are depicted. For the continuation parameter $\contpar=0.5$ the state paths and corresponding norms are shown in \subref{fig:foss2hossb052} and \subref{fig:foss2hossb053}. The final results are illustrated in \subref{fig:foss2hossb054} (state paths), \subref{fig:foss2hossb057} (control paths), \subref{fig:foss2hossb055} (costate paths), and \subref{fig:foss2hossb056} (norms). In \subref{fig:foss2hossb058} the corresponding slice manifolds are shown in the state-costate space. The objective values for solutions of the slice manifolds in \subref{fig:foss2hossb059} show that the path converging to the oligotrophic equilibrium is optimal among all solutions that start in $P^d(0)=\phcssspp$.} 
	\label{fig:foss2hossb05}
\end{figure}
 
\begin{figure}
\centering
	\href{run:SLFDM_C30825_FindSkibaDistributionCombine.mp4}{\subfloat[Manifold of initial distributions]{\includegraphics[width=0.33\textwidth]{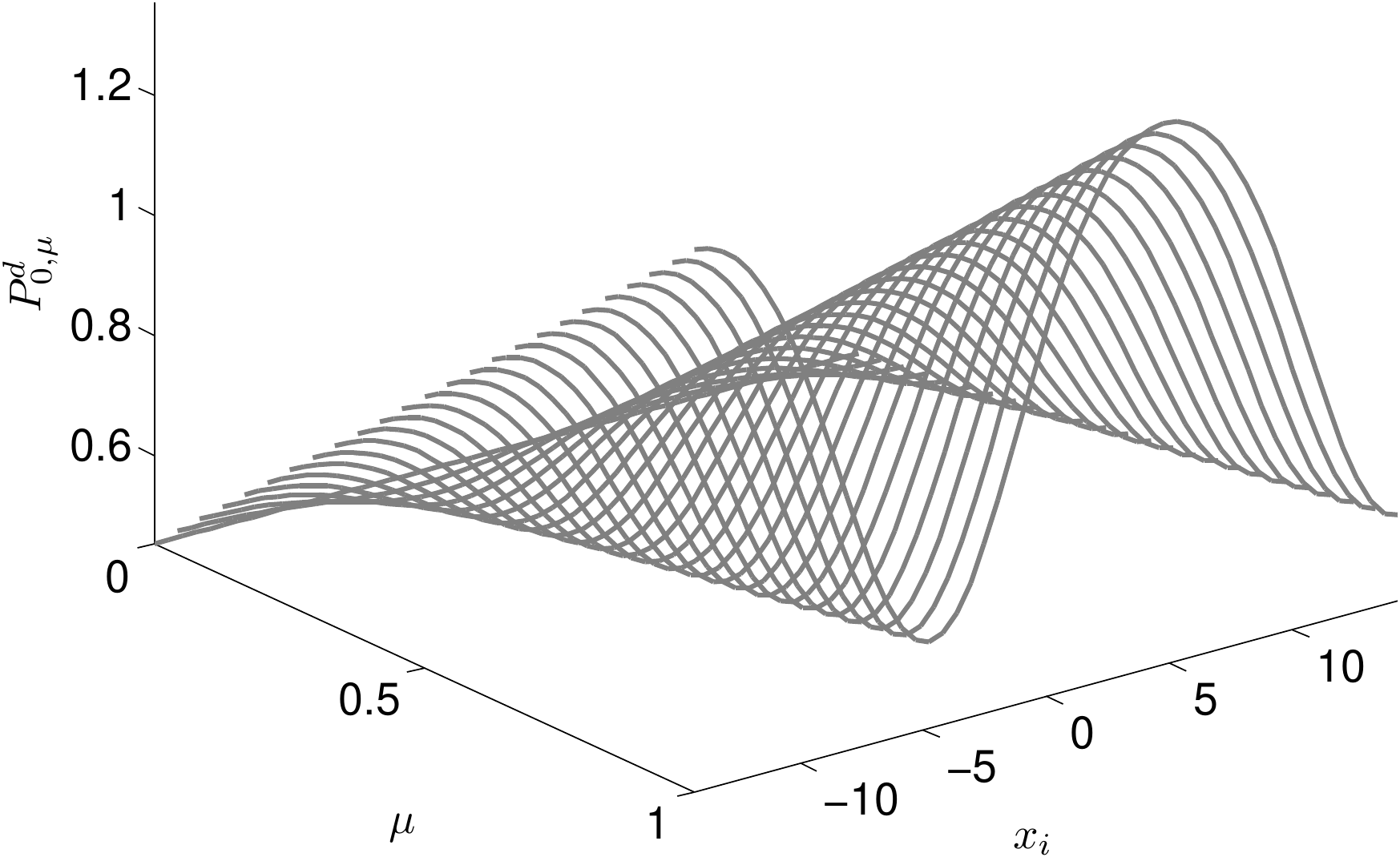}\label{fig:fossskiba1}}}\hfill
	\href{run:SLFDM_C30825_SkibaManifold_t_P.mp4}{\subfloat[State paths of Skiba solution]{\includegraphics[width=0.33\textwidth]{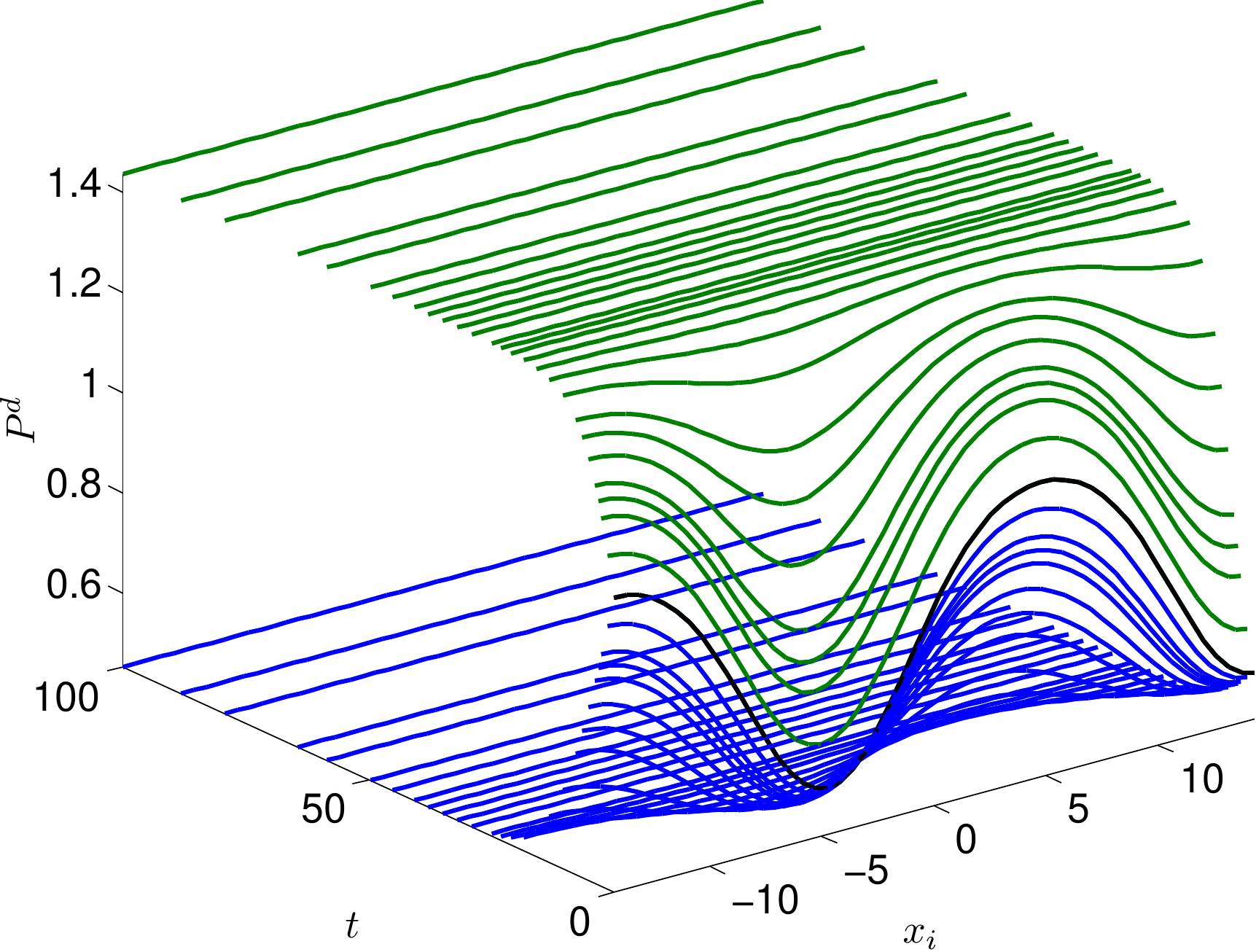}\label{fig:fossskiba2}}}\hfill
	\subfloat[State paths of Skiba solution]{\includegraphics[width=0.33\textwidth]{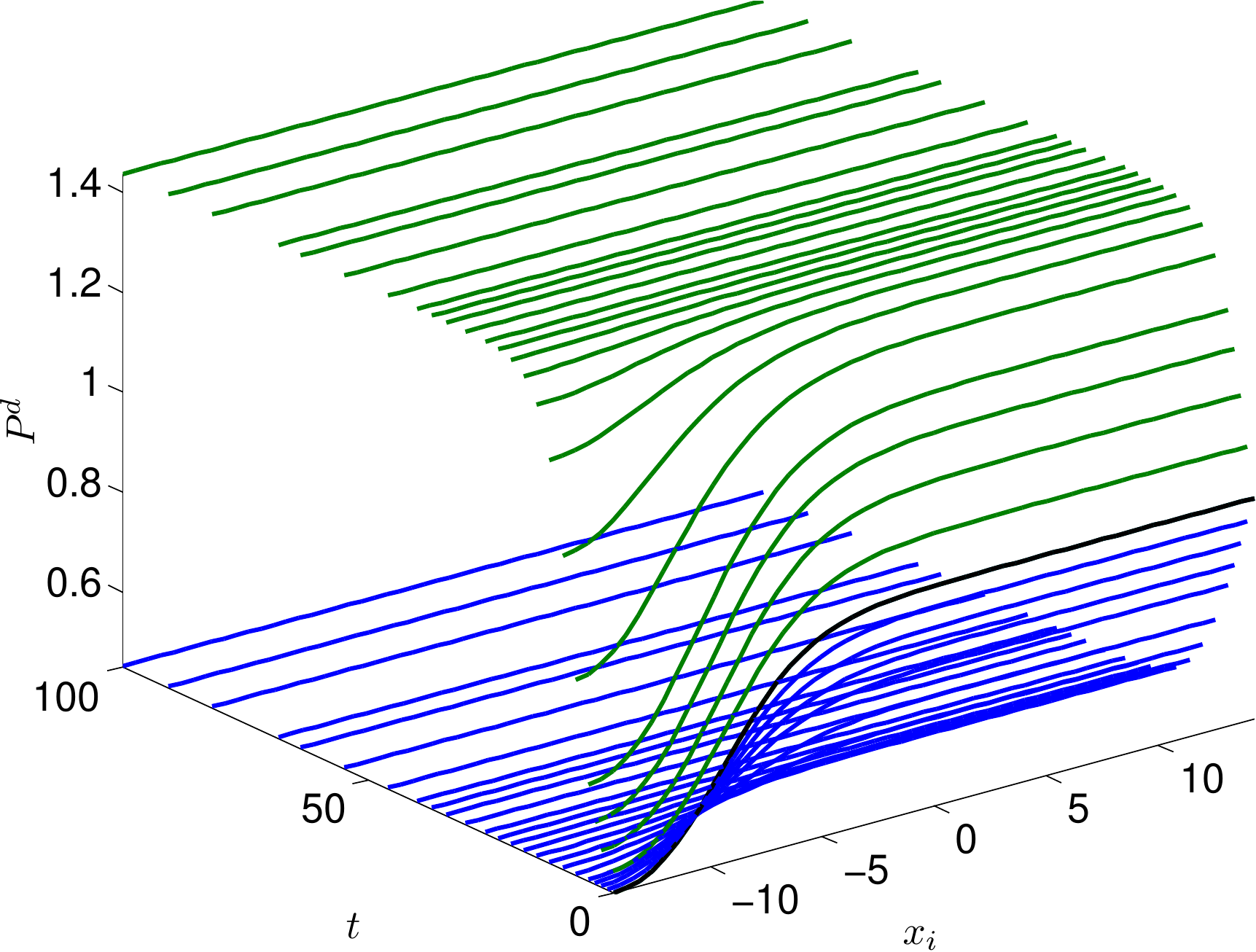}\label{fig:fossskiba3}}\\
	\subfloat[Crossing of objective values]{\includegraphics[width=0.33\textwidth]{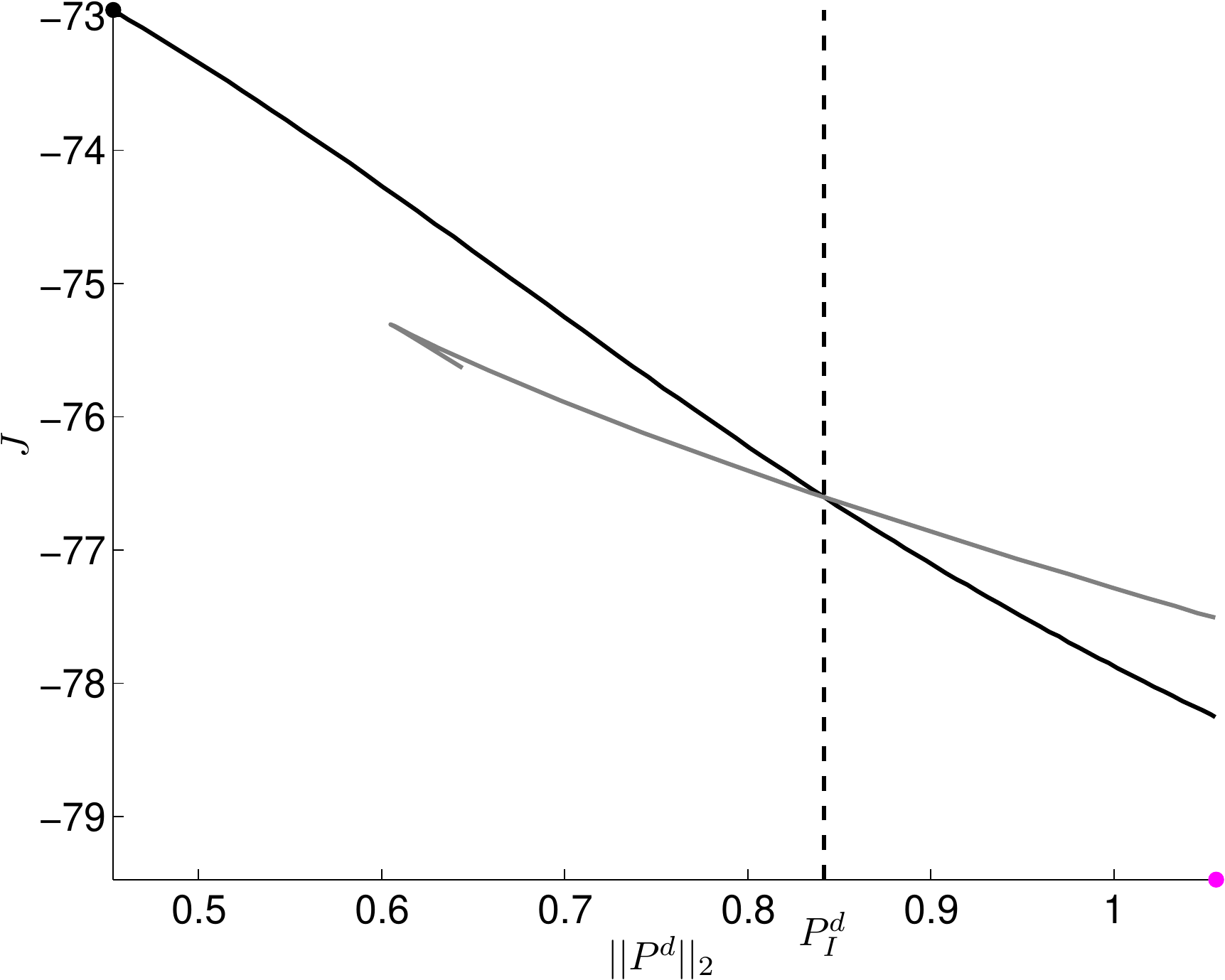}\label{fig:fossskiba4}}\hfill
	\href{run:SLFDM_C30825_SkibaManifold_t_u.mp4}{\subfloat[Control paths of Skiba solution]{\includegraphics[width=0.33\textwidth]{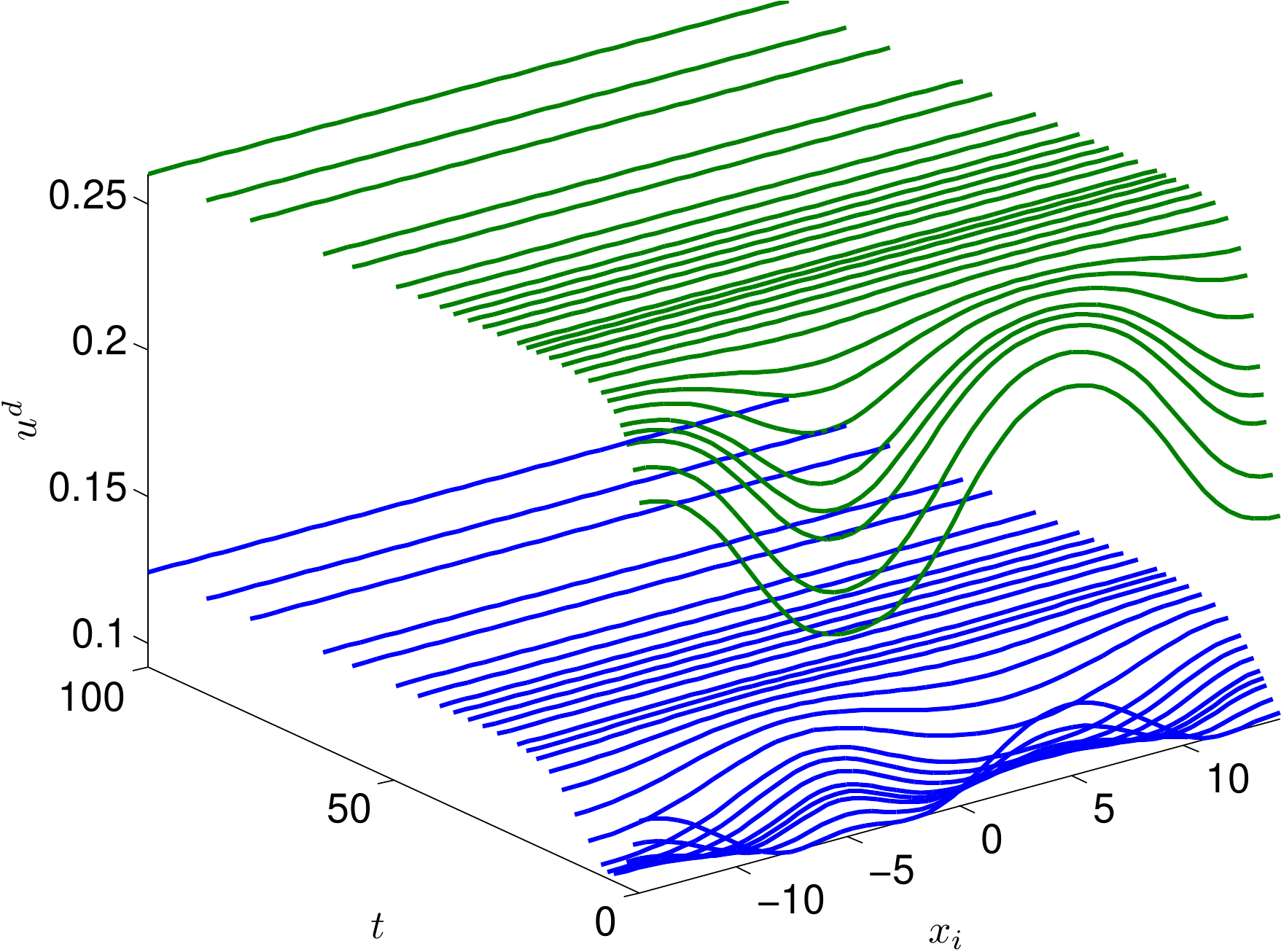}\label{fig:fossskiba5}}}\hfill
	\subfloat[Control paths of Skiba solution]{\includegraphics[width=0.33\textwidth]{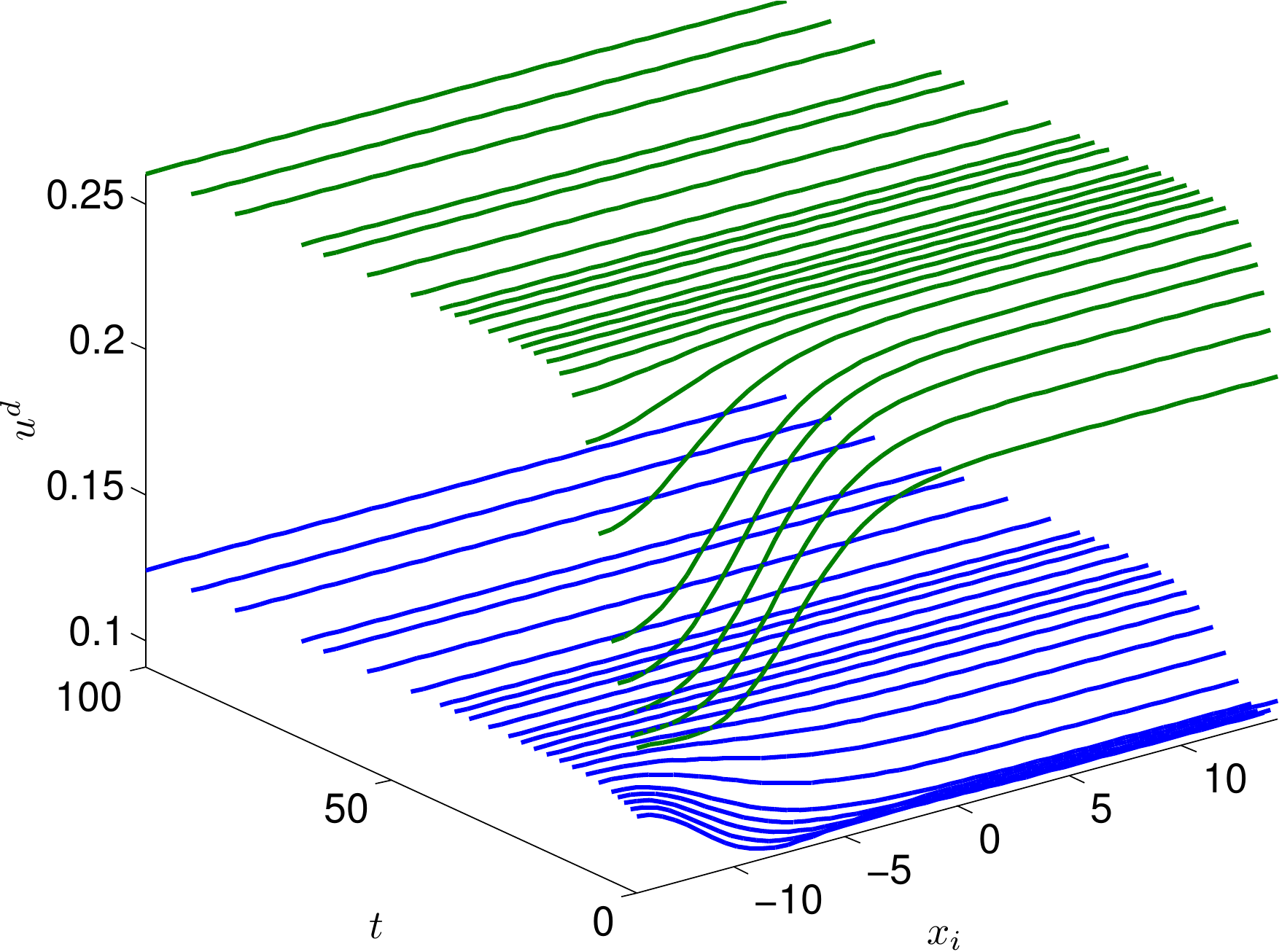}\label{fig:fossskiba6}}
\caption[]{This figure shows the detection of a patterned Skiba distribution and its continuation to a different Skiba distribution on the Skiba manifold. To receive the animation files associated to the panels \subref{fig:fossskiba1}, \subref{fig:fossskiba2} and \subref{fig:fossskiba5} please contact the author.}
	\label{fig:fossskiba}
\end{figure}

\subsection{Second scenario}
\label{sec:scenarion2}
In the first scenario with $b=0.65$ we only found the eutrophic and oligotrophic equilibria being \foss. This is somehow analogous to the 0D model, where the \css{} does not appear in the optimal system (see~\cref{fig:sl_optmultipleequ}). 

The result for the 0D model in the second scenario shows that beside the eutrophic and oligotrophic \cssspp{} the unstable node (\cssnspp) appears as limit of the regions of attractions for the two \cssspp. Thus, we can expect that \fcssnspp{} is optimal in \modelcref{eq:dslocm}. Therefore, also \hcssnspp{} cannot be excluded to be optimal.

We therefore analyze two different cases with $c=3.5$ where none of the \hcss{} satisfy \spp{} and $c=3.0825$, for which one of the \hcss{} satisfies \spp. For the original shallow lake \modelcref{eq:sl_model} these two cases are qualitatively the same (cf.~\cref{fig:sl_optmultipleuniqueequ}).

\subsubsection{Patterned equilibrium not satisfying \spp}
\label{sec:HOSSNotSatisfyingSPP}
For $c=3.5$ we try to find solutions that start at the states of the \fcssnspp{} ($\pfossmiddle$) and one of the \hcssnspp{} and converge to the eutrophic (oligotrophic) \fcss{}.

In \cref{fig:foss2unstI} the main results of this analysis are depicted. The example depicted in the first column (\subref*{fig:foss2unstI1},\subref*{fig:foss2unstI4} and \subref*{fig:foss2unstI7}) is analogous to the case in the 0D model shown in \cref{fig:sl_optmultipleuniqueequ}. Thus, it is not possible to find a solution starting at the initial states of the \fcssnspp{} and converging to the eutrophic or oligotrophic \fcss. Instead, during the continuation process the initial states of $\pfossmiddle$ are approached but cannot be reached. In \cref{fig:foss2unstI1} the phase portrait of the final result is plotted, which is analogous to the state-costate space in 0D (cf.~\cref{fig:sl_optmultipleuniqueequa}). In \cref{fig:foss2unstI3} we see that the objective value for the homogeneous initial distributions is continuous (cf.~\cref{fig:sl_optmultipleuniqueequb}). Consequently, for spatially homogeneous initial distributions, the optimal path is unique, where $\pfossmiddle$ separates the regions of attractions for the eutrophic (oligotrophic) equilibrium $\pfecssspp$ ($\pfocssspp$). The optimal state paths for solutions starting exactly at $\pfossmiddle$ (the equilibrium solution, black) and in the near vicinity (blue and green) are depicted in \cref{fig:foss2unstI4}.
Repeating these steps for each of the \hcss, two examples are depicted in the last two columns of \cref{fig:foss2unstI}, yields that each of these equilibria is optimal, i.e. are \hoss. Since none of the \hcss{} satisfy \spp{} these equilibria and their stable manifolds separate the regions of attractions of the \foss. 

At this point a few questions remain unsettled.
\begin{itemize}
	\item Is the defect of an equilibrium not satisfying \spp{} constant for state discretizations that are fine enough?
	\item What does this mean for the original \PDE{} problem?
	\item Can we say that equilibria not satisfying \spp{} and their stable manifolds separate the regions of attractions for the solutions satisfying \spp?
	\item What is the state space of the \PDE{} problem?
\end{itemize}

\subsubsection{Patterned equilibrium satisfying \spp}
\label{sec:HOSSSatisfyingSPP}
In this section we numerically check, whether the unique patterned equilibrium \hcssspp{} for $c=3.0825$ is \hoss (cf.~\cref{fig:sld_bifdiag2}). For that reason we have to show that there exists no other solution path of the canonical system \cref{eq:dslocmcansys} that starts at $\phcssspp$ yielding a larger objective value. The only other candidates are stable paths of the eutrophic and oligotrophic \fcssspp. The result of the numerical comparison for the oligotrohic versus the patterned equilibrium is depicted in \cref{fig:fossnhoss3} and \cref{fig:fossnhoss2}.

To find a feasible path $(P_1^d(\cdot,\contpar),\lambda_1^d(\cdot,\contpar))$ that satisfies $P^d(0,1)=\phcssspp$ and 
\begin{equation*}
	\lim_{t\to\infty}(P_1^d(t,\contpar),\lambda_1^d(t,\contpar))=(\pfocssspp,\lfocssspp)\quad\text{with}\quad\contpar\in[0,1]
\end{equation*}
we solve the homotopy problem \cref{eq:homcontbvp}, starting with the constant equilibrium solution $(\pfocssspp,\lfocssspp)$. The continuation process revealed that it is not possible to find a feasible path for $\kappa=1$, instead some value $\kappa_0<1$ was approached. The last computed path $(P_1^d(\cdot,\contpar_0),\lambda_1^d(\cdot,\contpar_0))$ is shown in \cref{fig:fossnhoss310} together with the corresponding slice manifold (dashed black).

Next we repeated the procedure for the reversed homotopy problem, starting  with the constant patterned solution $(\phcssspp,\lhcssspp)$ and trying to find a feasible path $(P_2^d(\cdot,1-\contpar),\lambda_2^d(\cdot,1-\contpar))$ that satisfies $P^d(0,0)=\pfocssspp$ and 
\begin{equation*}
	\lim_{t\to\infty}(P_2^d(t,1-\contpar),\lambda_2^d(t,1-\contpar))=(\phcssspp,\lhcssspp)\quad\text{with}\quad\contpar\in[0,1].
\end{equation*}
Again the continuation process revealed that it was not possible to find a feasible path for $\kappa=0$, instead some value approximately $1-\kappa_0$ was approached. This solution $(P_2^d(\cdot,1-\contpar),\lambda_2^d(\cdot,1-\contpar))$ is represented in \cref{fig:fossnhoss301} by the blue solution path and black slice manifold.

The two last solution paths from both continuation processes suggest that there exists a separating manifold for the regions of attractions of the oligotrophic \fcssspp{} and \hcssspp. A possible candidate for this separating manifold is the stable manifold of the 
\hcssnspp{} with defect $-1$ (see the \textcolor{red}{$\diamond$} in \cref{fig:fossnhoss300}). To test this conjecture we solved the homotopy problem \cref{eq:homcontbvpnspp} for defective equilibria. For $x_1^{(1)}$ we took $P_1^d(0,\contpar_0)$ (the initial states of the last continuation step of the first homotopy problem) and set $V_1\defin(1,\ldots,1)\in\R^{N+1}$, which satisfies the rank condition \cref{eq:homcontbvpnspp4}. The last solution $(P_3^d(\cdot,1),\lambda_3^d(\cdot,1))$ of this homotopy problem is depicted as dashed blue curve in \cref{fig:fossnhoss3} and gives a strong numerical argument for our conjecture. 

The overall picture, \cref{fig:fossnhoss3}, suggests that for every $\varepsilon>0$ there exists $\kappa_1$ and $\kappa_2$ such that there exists solutions $(P_1^d(\cdot,\contpar_1),\lambda_1^d(\cdot,\contpar_1))$ and $(P_2^d(\cdot,1-\contpar_2),\lambda_2^d(\cdot,1-\contpar_2))$ of the homotopy problems with
\begin{align*}
	&\Norm{(P_1^d(0,\contpar_1),\lambda_1^d(0,\contpar_1))-(P_3^d(0,1),\lambda_3^d(0,1))}_2<\varepsilon
	\shortintertext{and}
	&\Norm{(P_2^d(0,1-\contpar_2),\lambda_2^d(0,1-\contpar_2))-(P_3^d(0,1),\lambda_3^d(0,1))}_2<\varepsilon
	\shortintertext{or even stronger}
	&\Norm{(P_1^d(\cdot,\contpar_1),\lambda_1^d(\cdot,\contpar_1))-(P_3^d(\cdot,1),\lambda_3^d(\cdot,1))}_{L_2}<\varepsilon
	\shortintertext{and}
	&\Norm{(P_2^d(\cdot,1-\contpar_2),\lambda_2^d(\cdot,1-\contpar_2))-(P_3^d(\cdot,1),\lambda_3^d(\cdot,1))}_{L_2}<\varepsilon.
\end{align*}
Plotting the objective values evaluated along the solutions of the corresponding slice manifolds shows that the objective function is continuous in the vicinity of $P_3^d(0,1)$, see \cref{fig:fossnhoss2}. An analogous result holds for the comparison of the stable paths converging to the eutrophic \fcssspp{} and the \hcssspp. This proves the optimality of \fcssspp{} and \hcssspp{} as well. Also according to the case $c=3.5$ the stable manifolds of defective equilibria separate the regions of attractions of the equilibria satisfying \spp.

\Cref{fig:fossnhoss1} shows (part of) a phase portrait in the state-costate space for $c=3.0825$. The subscripts of the equilibria denote the defect of the according equilibrium. Thus, there exist two \fcssspp{} ($\pfecssspp$ and $\pfocssspp$), and one \hcssspp{} ($\phcssspp$). Additionally a few paths are plotted that converge to $\pfecssspp,\ \pfocssspp$, and $\phcssspp$ (solid blue) and to \hcssnspp{} with defect $-1$ (dashed blue). The specific case for solutions that converge to the oligotrophic equilibrium $(\pfocssspp,\lfocssspp)$ and patterned equilibrium $(\phcssspp,\lhcssspp)$ is separately illustrated in \cref{fig:fossnhoss3} and \cref{fig:fossnhoss2}. 

\begin{figure}
\centering
	\subfloat[Continuation to $P^d(0)\approx\pfcssnspp$]{\includegraphics[width=0.33\textwidth]{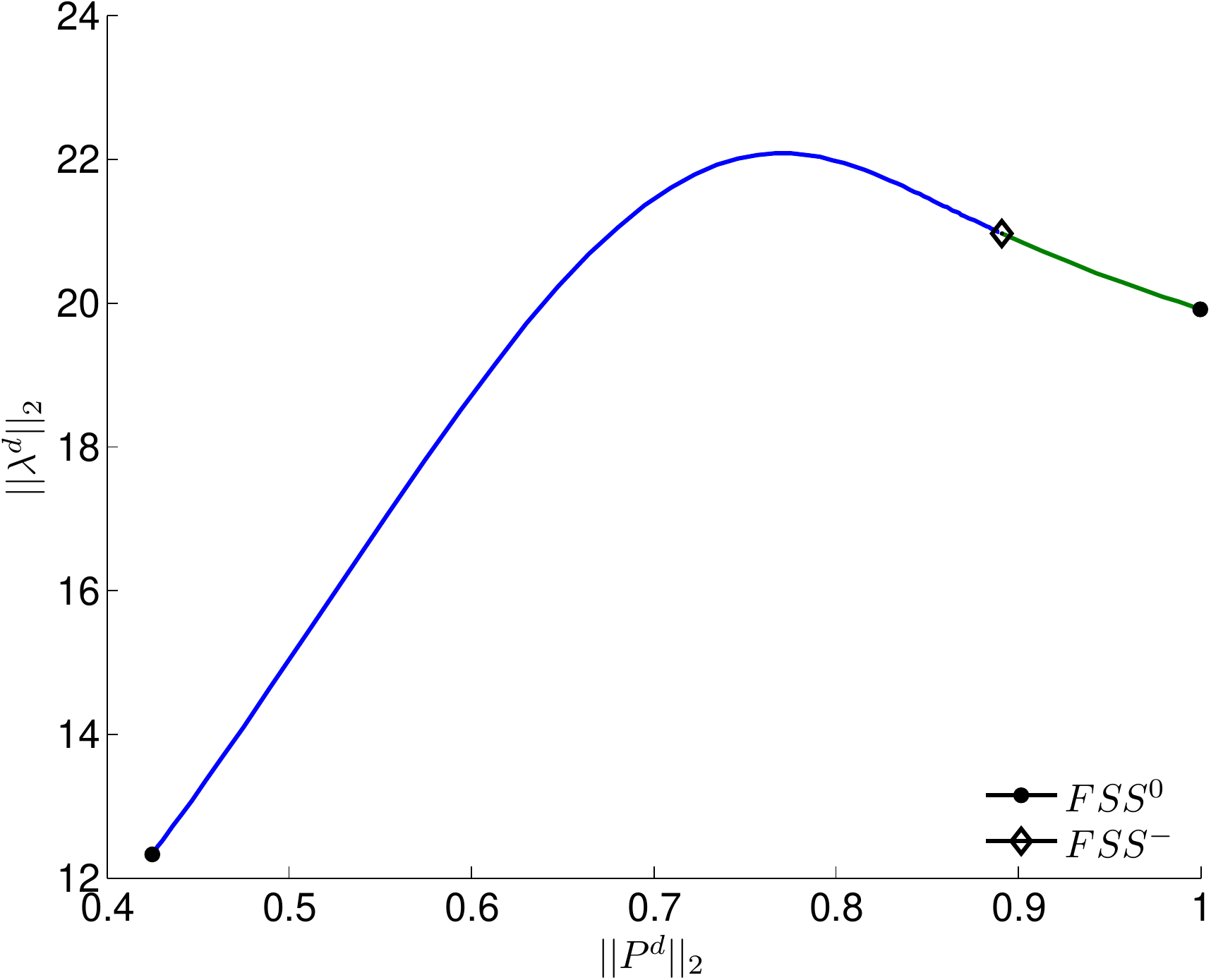}\label{fig:foss2unstI1}}\hfill
	\subfloat[Continuation to $P^d(0)\approx$\textcolor{magenta}{$\phcssnspp$}]{\includegraphics[width=0.33\textwidth]{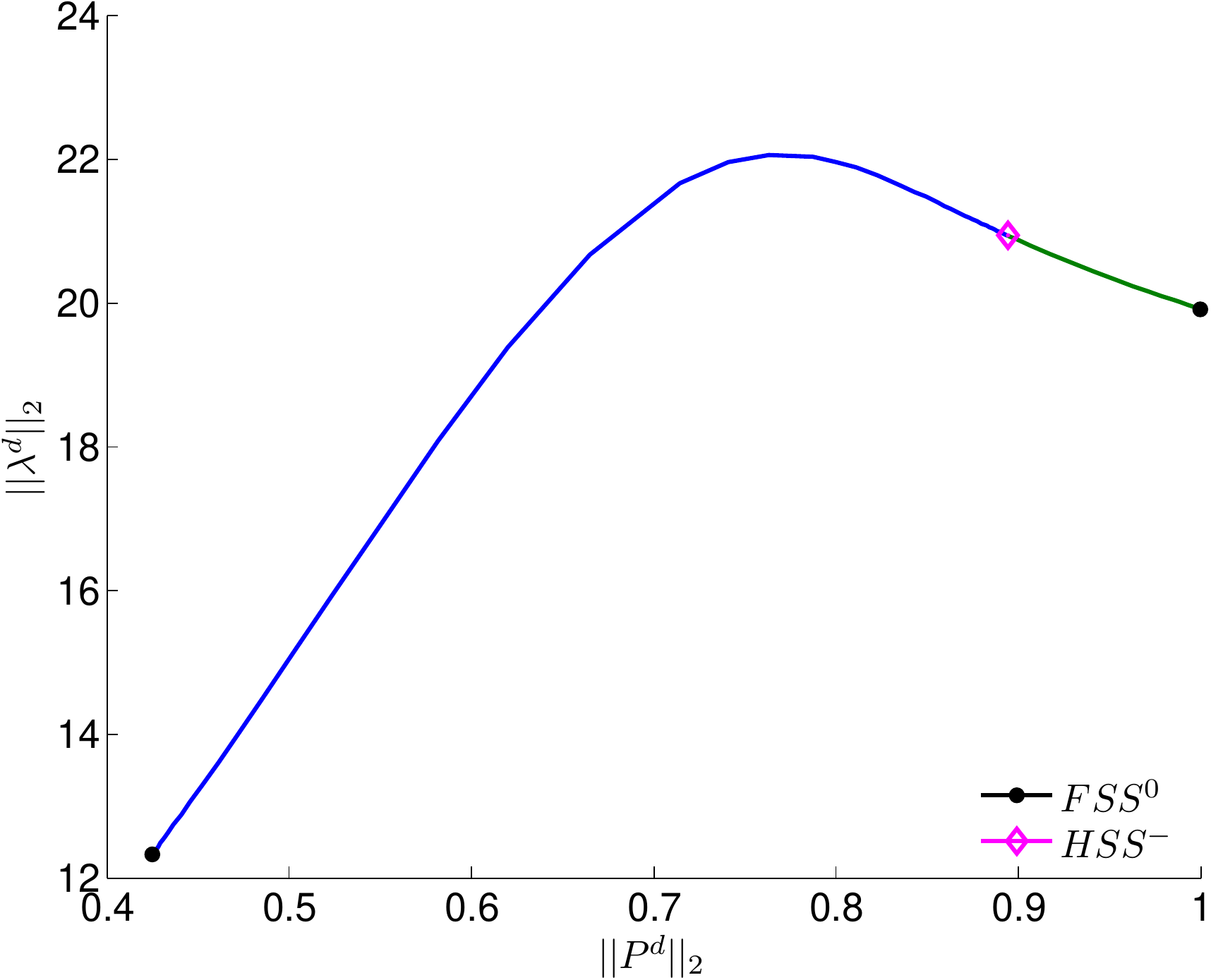}\label{fig:foss2unstI2}}\hfill
	\subfloat[Continuation to $P^d(0)\approx$\textcolor{red}{$\phcssnspp$}]{\includegraphics[width=0.33\textwidth]{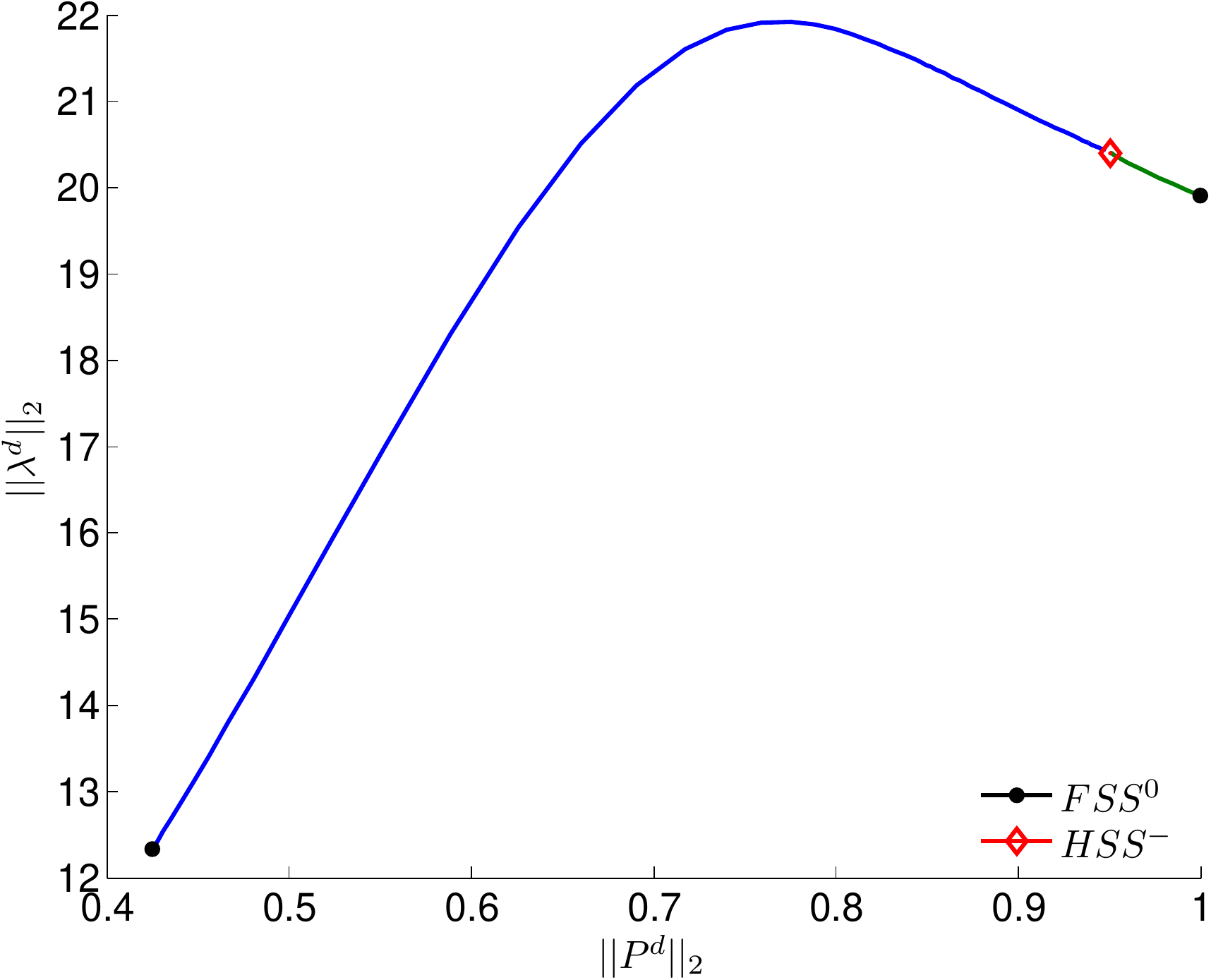}\label{fig:foss2unstI3}}\\
	\subfloat[State paths starting near $\pfcssnspp$]{\includegraphics[width=0.33\textwidth]{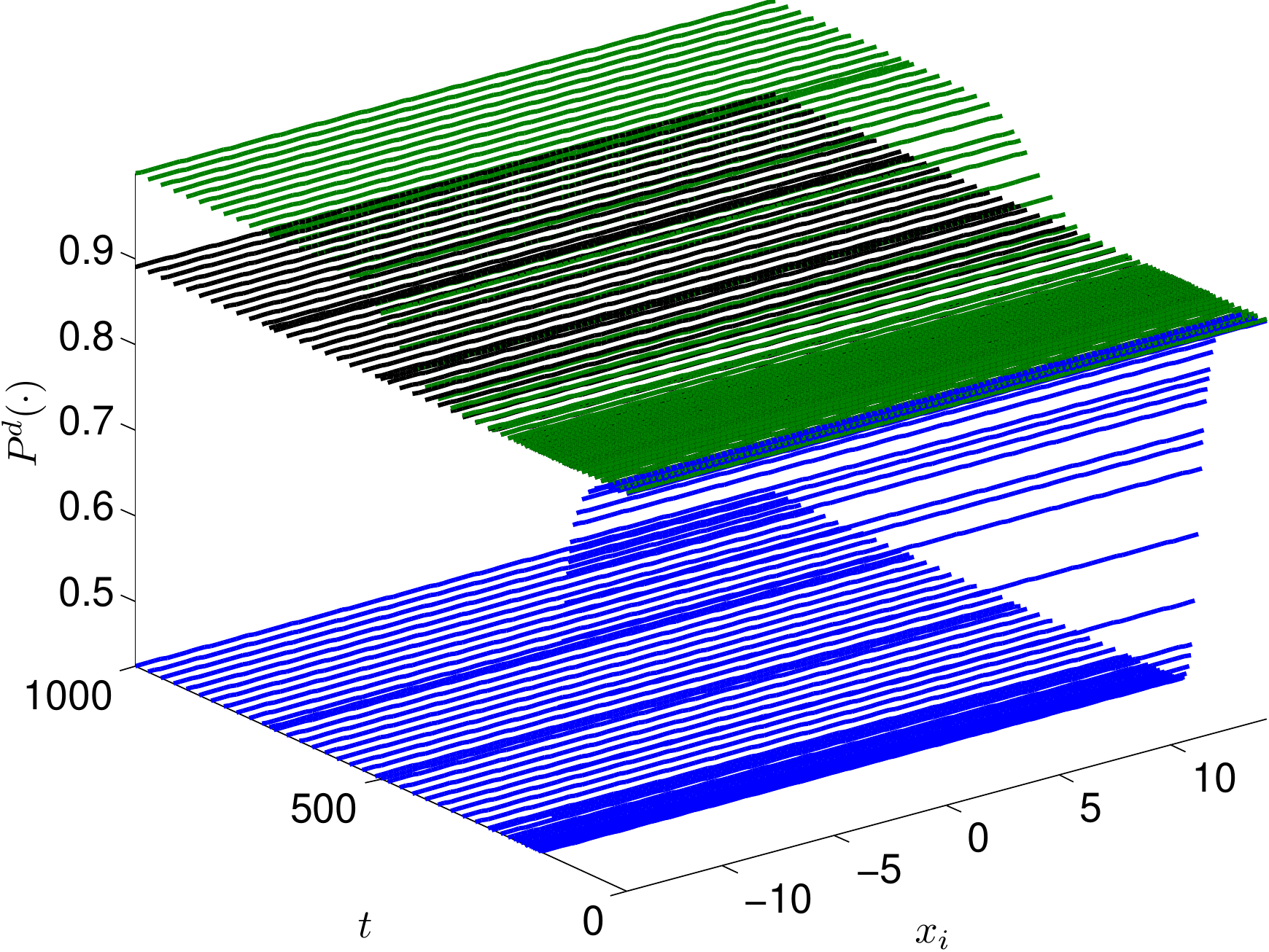}\label{fig:foss2unstI4}}\hfill
	\subfloat[State paths starting near \textcolor{magenta}{$\phcssnspp$}]{\includegraphics[width=0.33\textwidth]{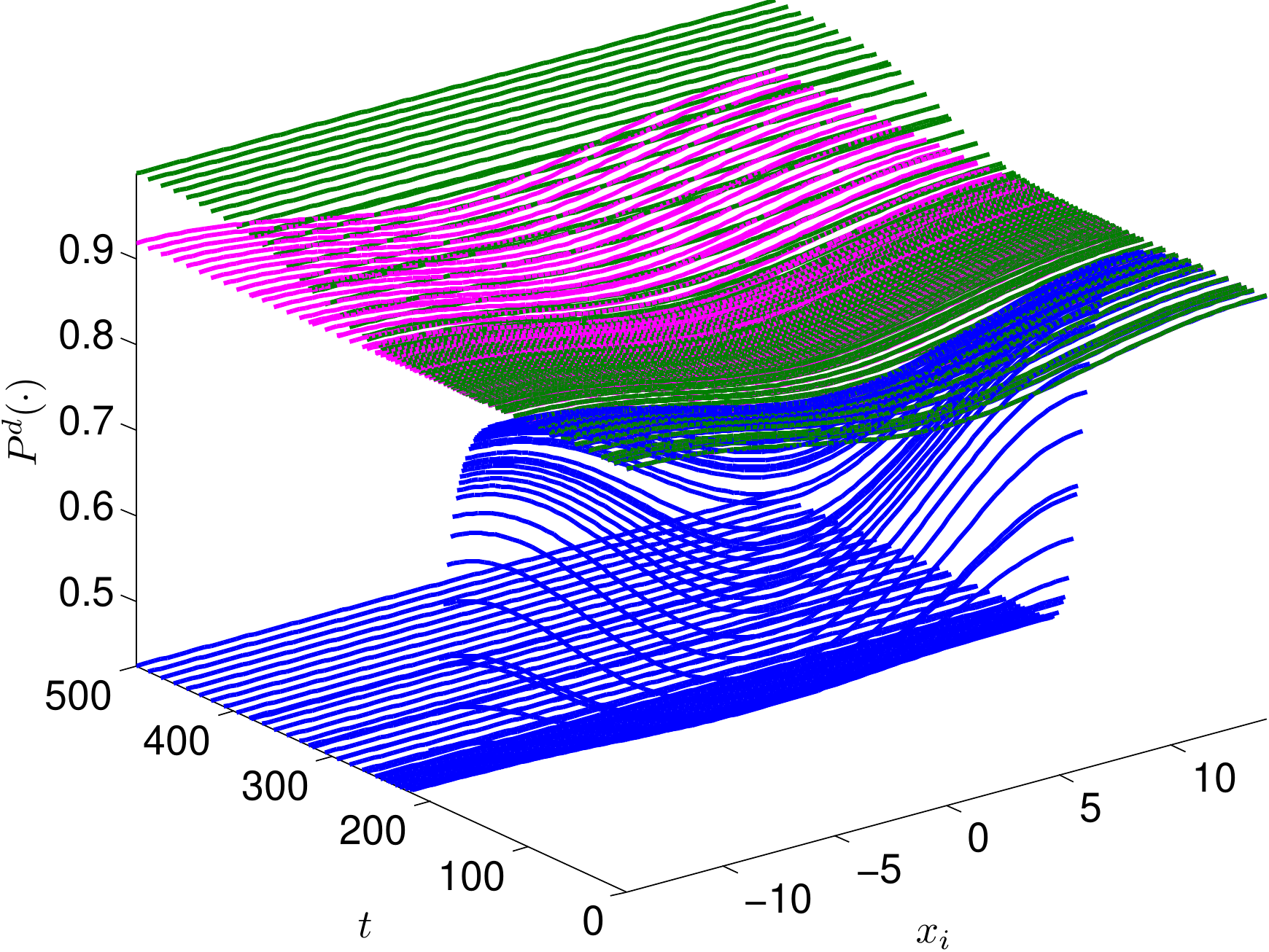}\label{fig:foss2unstI5}}\hfill
	\subfloat[State paths starting near \textcolor{red}{$\phcssnspp$}]{\includegraphics[width=0.33\textwidth]{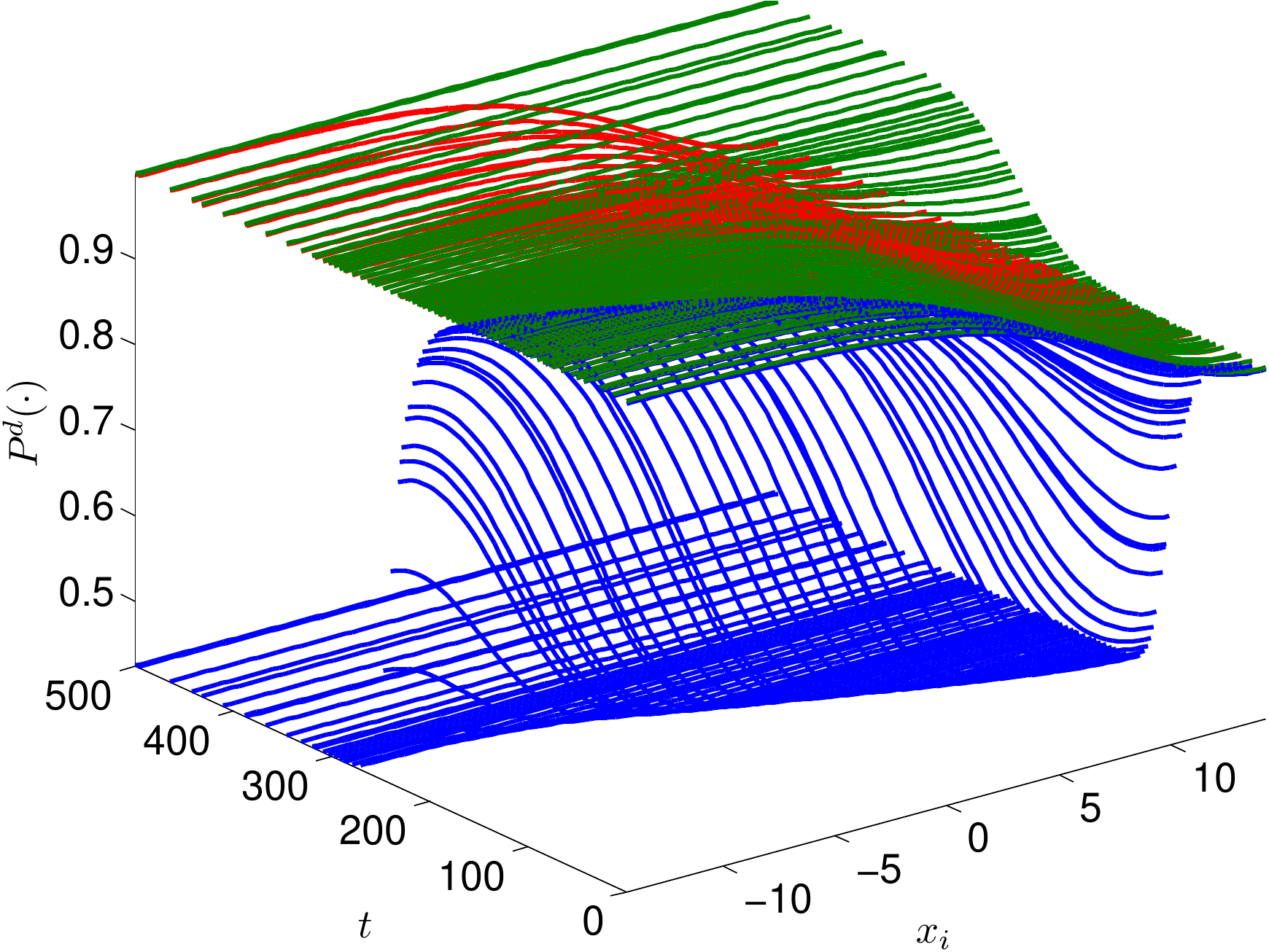}\label{fig:foss2unstI6}}\\
	\subfloat[Objective value along slice manifolds]{\includegraphics[width=0.33\textwidth]{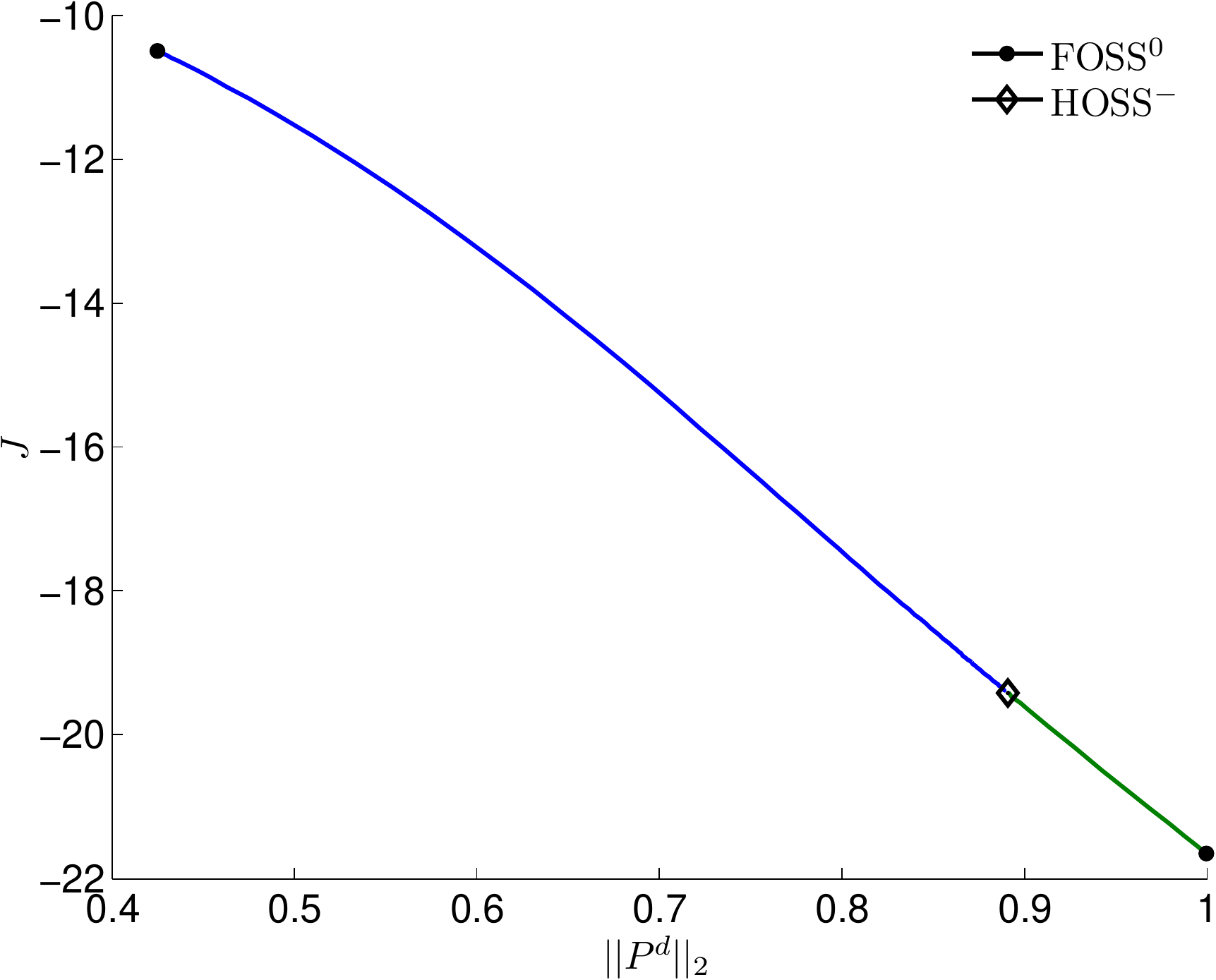}\label{fig:foss2unstI7}}\hfill
	\subfloat[Objective value along slice manifolds]{\includegraphics[width=0.33\textwidth]{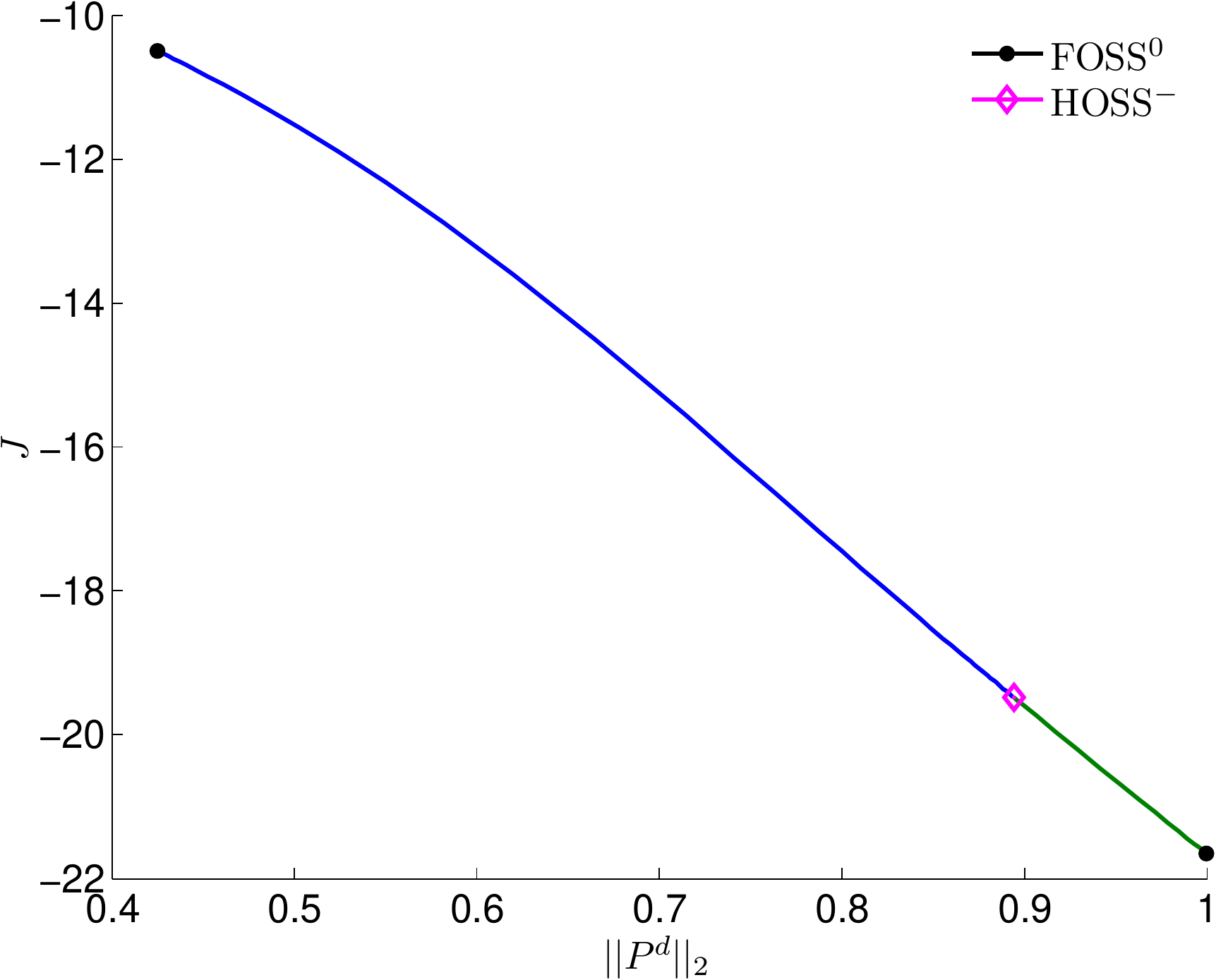}\label{fig:foss2unstI8}}\hfill
	\subfloat[Objective value along slice manifolds]{\includegraphics[width=0.33\textwidth]{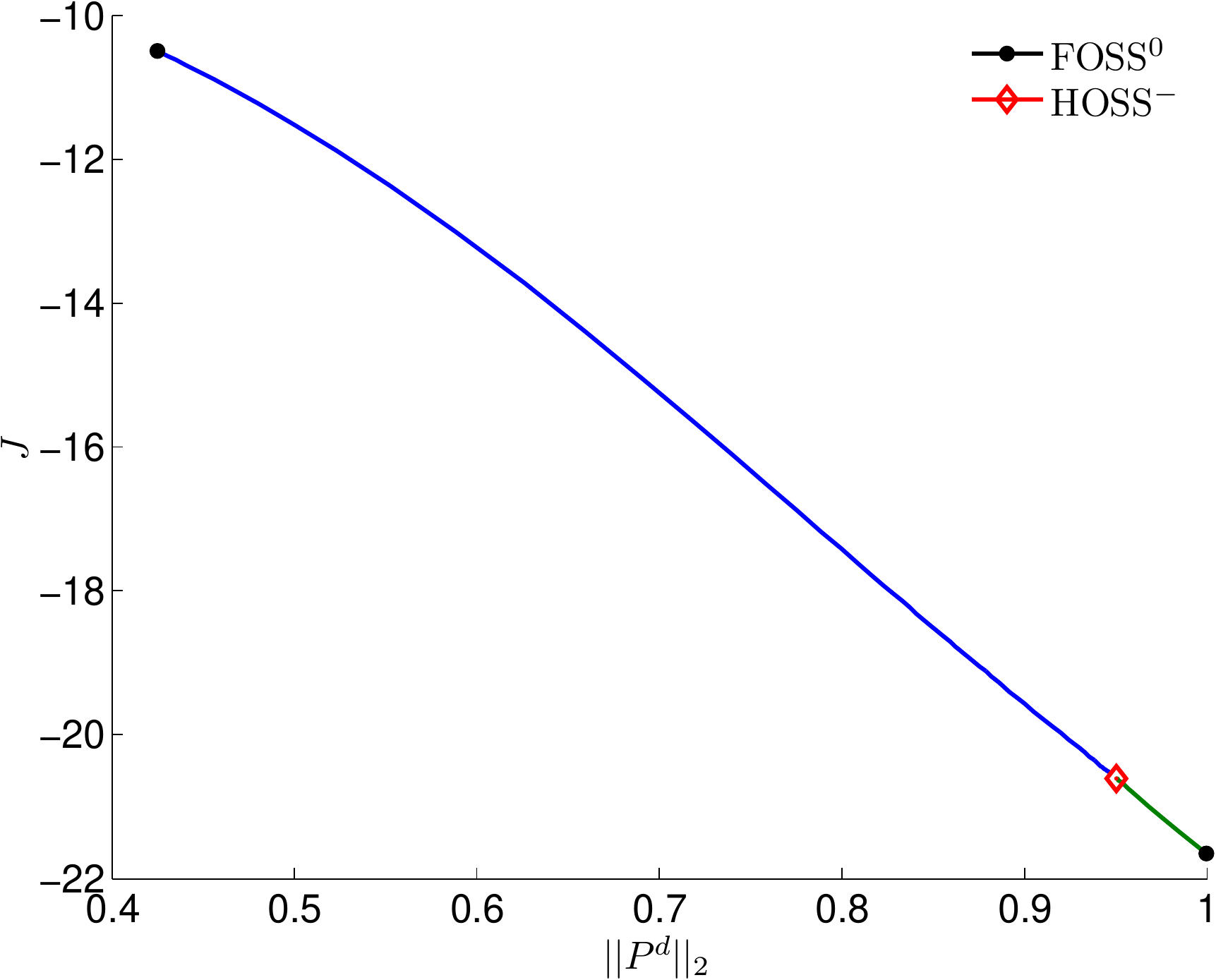}\label{fig:foss2unstI9}}
\caption[]{This figure depicts the numerical proof for the optimality of \fcssnspp{} and \hcssnspp{} at $c=3.5$, i.e., the optimality of the equilibria that do not satisfy \spp, exemplified for three equilibria. In the first row \subref{fig:foss2unstI1}-\subref{fig:foss2unstI3} the slice manifolds for the according continuation processes are depicted in the normed state-costate space. The second row \subref{fig:foss2unstI4}-\subref{fig:foss2unstI6} shows the state paths for the solutions starting at and near the \fcssnspp{} and \hcssnspp, respectively. The last row \subref{fig:foss2unstI7}-\subref{fig:foss2unstI9} illustrates that the objective function is continuous in the vicinity of the constant equilibria solutions for the \fcssnspp{} and \hcssnspp. Therefore the equilibria not satisfying \spp{} are optimal as well.}
	\label{fig:foss2unstI}
\end{figure}


\begin{figure}
\centering
	\subfloat[Homotopy \fcssspp{} to \hcssspp{}]{\includegraphics[width=0.33\textwidth]{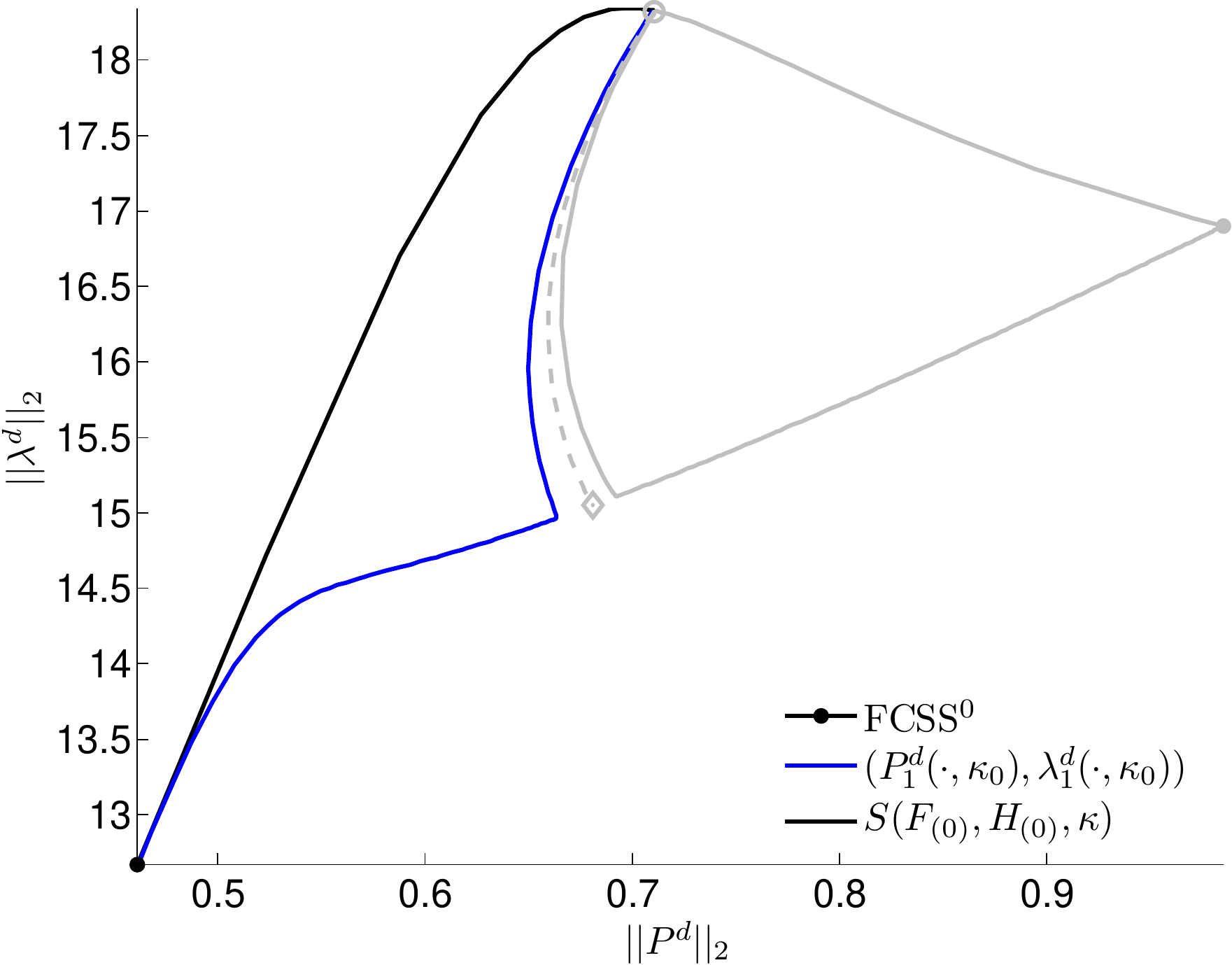}\label{fig:fossnhoss310}}\hfill
	\subfloat[Homotopy \hcssspp{} to \fcssspp{}]{\includegraphics[width=0.33\textwidth]{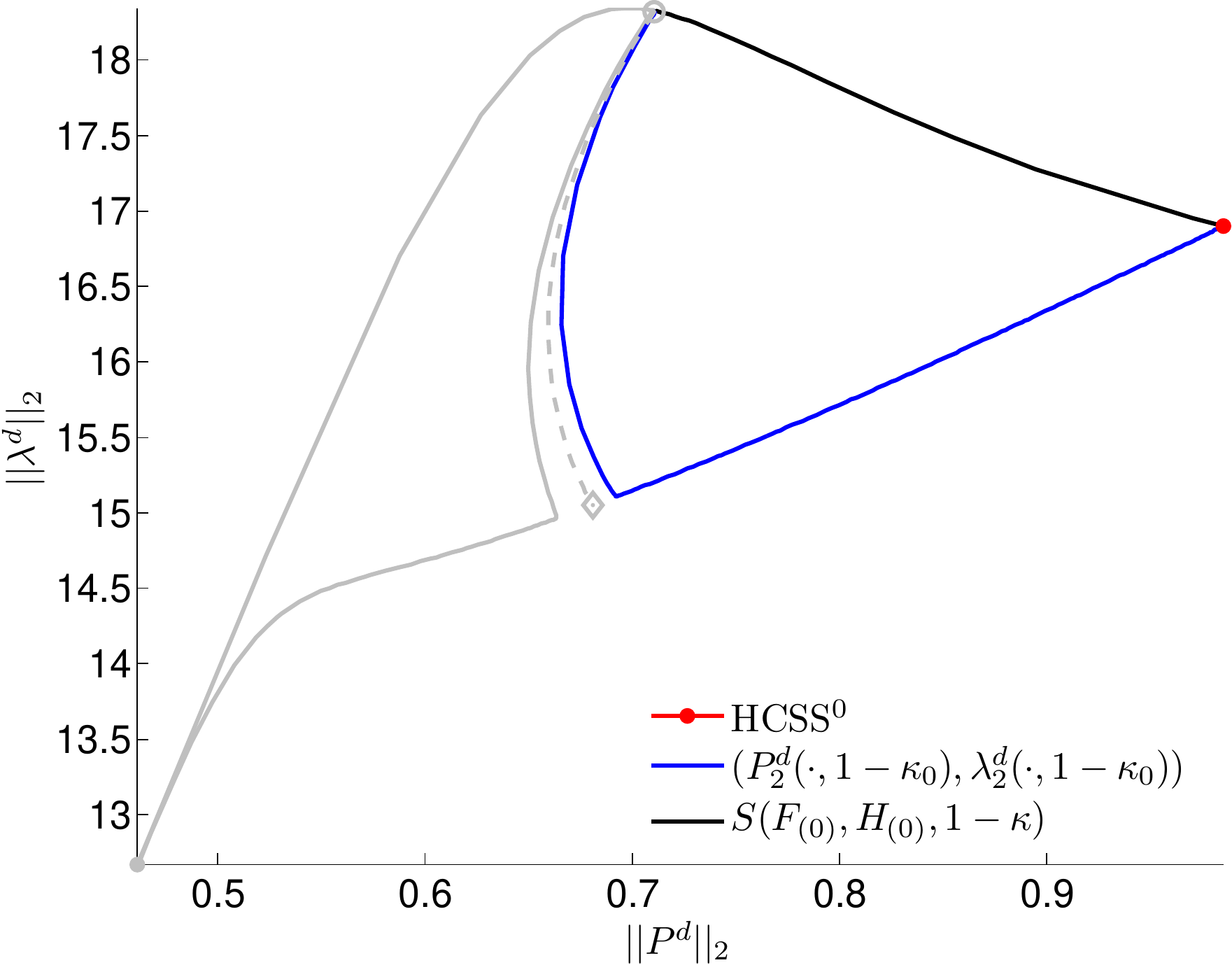}\label{fig:fossnhoss301}}\hfill
	\subfloat[Homotopy \hcssnspp{} to $P_1(0,\kappa_0)$]{\includegraphics[width=0.33\textwidth]{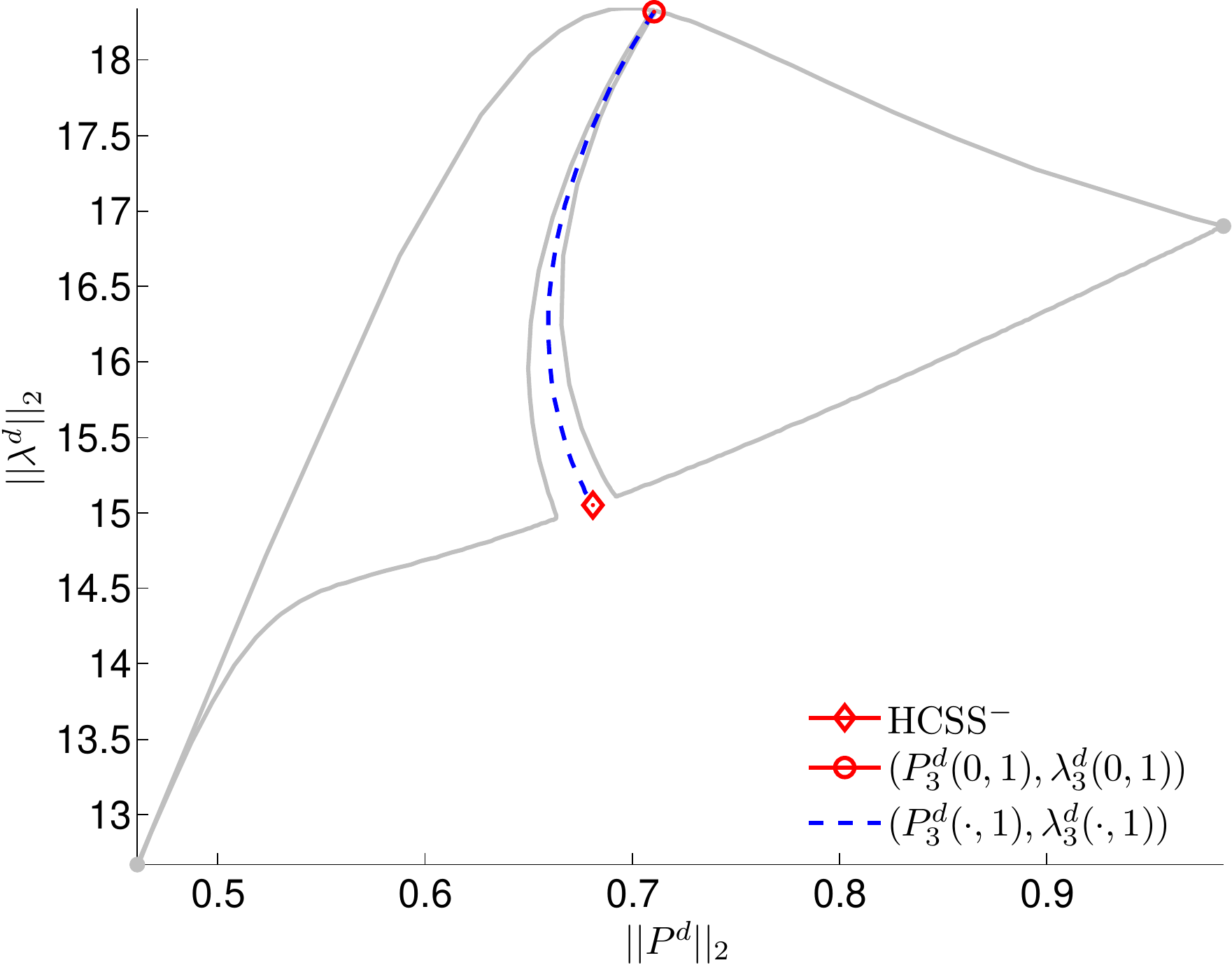}\label{fig:fossnhoss300}}\\
	\href{run:SLFDM_C30825_AnalogonUnstableNodeCont_t_P.mp4}{\subfloat[Slice manifolds and solution paths]{\includegraphics[width=0.33\textwidth]{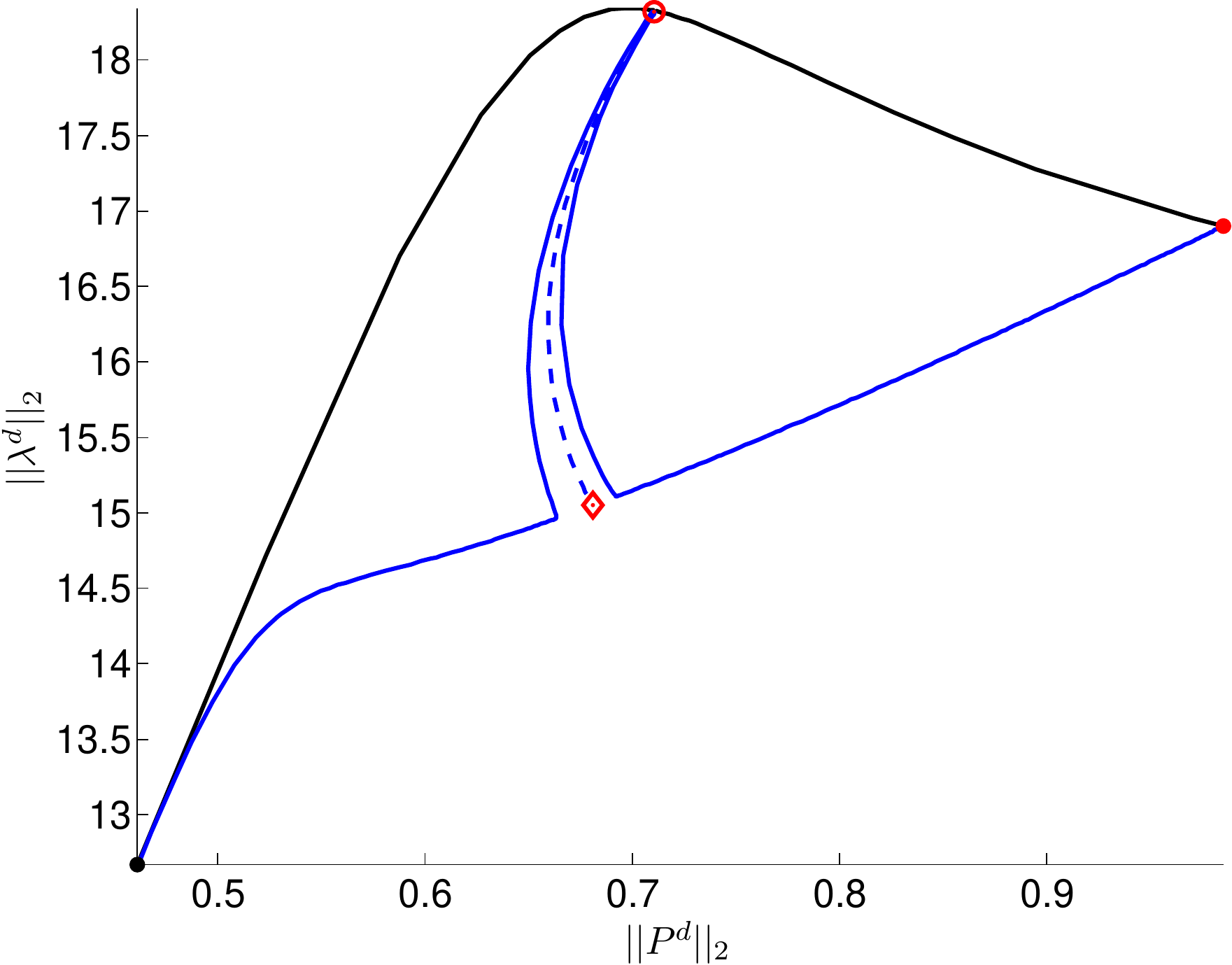}\label{fig:fossnhoss3}}}\hfill
	\subfloat[Objective value]{\includegraphics[width=0.33\textwidth]{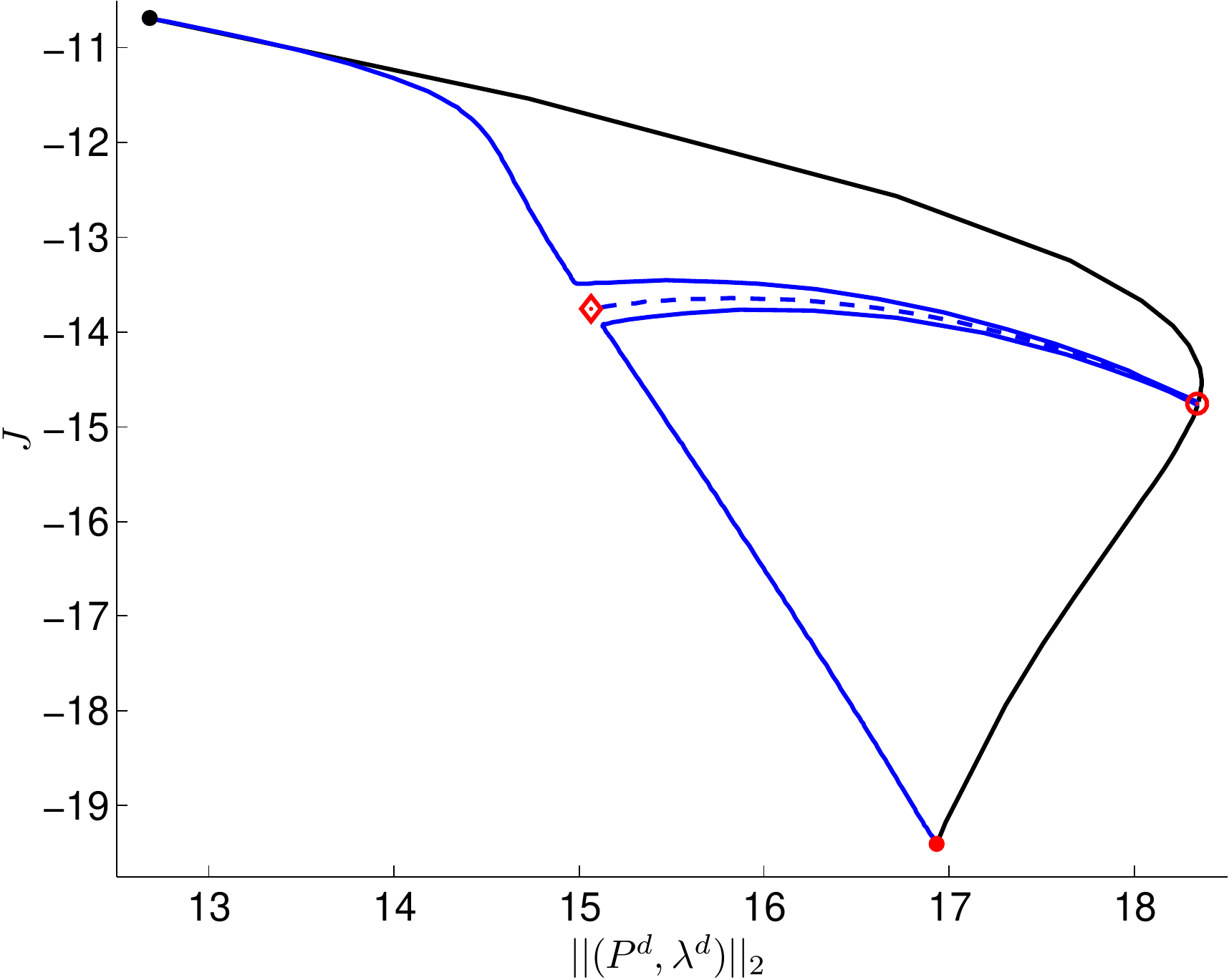}\label{fig:fossnhoss2}}\hfill
	\subfloat[Phase portrait]{\includegraphics[width=0.33\textwidth]{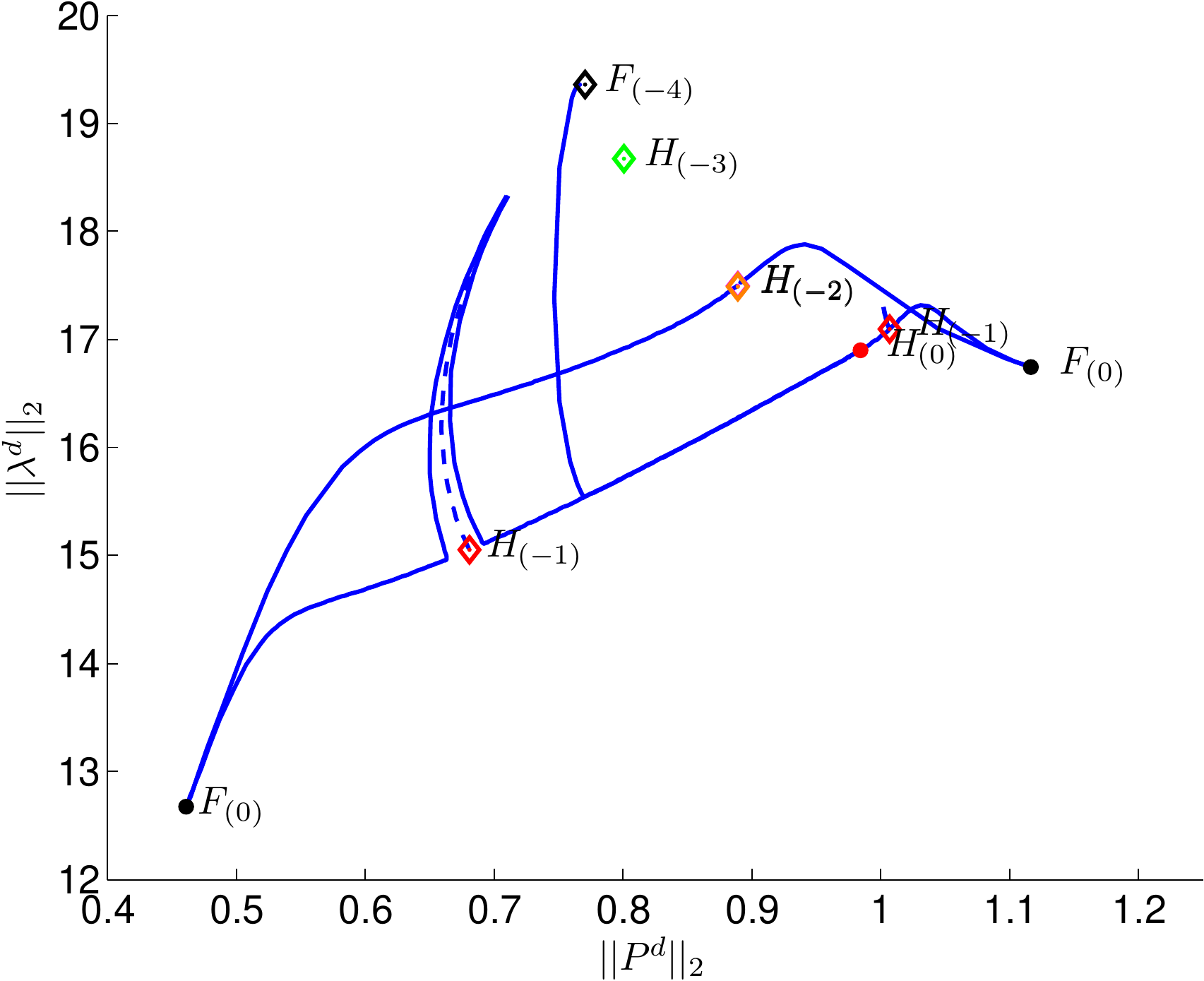}\label{fig:fossnhoss1}}
\caption[]{In \subref{fig:fossnhoss1} some of the solutions paths and equilibria in the state-costate space are depicted. The solutions shown in \subref{fig:fossnhoss3} and \subref{fig:fossnhoss2} are the result of the continuation processes, when we tried to find a solution starting at the states of  \hoss{} and converging to the oligotrophic \foss, and vice versa. The continuation process approached a state lying on the stable manifold (dashed blue) of the \hcssnspp{} with defect $-1$. Thus, for initial states coinciding with the states of the stable manifold with defect $-1$ it is optimal to converge along the stable manifold to the \hoss{} with defect $-1$. In the vicinity of these initial states it is optimal to converge either to the oligotrophic \foss{} or the \hoss{} satisfying \spp{}. To receive the animation file associated to the panel \subref{fig:fossnhoss3} please contact the author.}
	\label{fig:fossnhoss}
\end{figure}

\section{Conclusion}
\label{sec:Conclusion}
The purpose of this article is the presentation of a numerical framework, that allows us to numerically experiment with 1D spatially distributed optimal control problems. Numerical experiment is meant in the sense of Heaviside, who claimed mathematics as an experimental science, cf.~\citet{heaviside1893}.

We were thus able to find numerical evidence for the occurrence of (indifference) threshold distributions in distributed models. Moreover, this points towards a possible generalization of the saddle point property, which in case of multiple canonical steady states (\css) is intimately connected to the existence of multiple optimal solutions. Even though this simple \FDM{} could also be applied to spatially 2D models, such an approach would immediately become numerically intractable. Therefore, it is an intermediate step to the method of a finite element discretization, that is presented in \citet{grassuecker2015} and extended in \citet{uecker2015}.

Different directions for further research result from the presented approach. The obvious next step is the previously mentioned application of \FEM{} discretization. This is realized as an add-on toolbox (\ptopoc) to the \MATL{} package \pdetopath, which is a numerical tool for continuation and bifurcation in 2D elliptic systems, cf.~\citet{ueckeretal2014}.\footnote{It can be downloaded for free from \httppdetopath.}

A main drawback of the actual approach is difficulty to handle (inequality) constraints. For a correct usage of the used \BVP solver the rhs of the dynamics has to at least continuously differentiable. This property is violated, when constraints become active. We solved that problem by considering different arcs of the solution path, where each of the arcs satisfied the differentiability condition. For non-distributed models this is in general a practicable way but hardly to realize for distributed models. Simply because for each transition from one arc to the other the according switching conditions have to be stated. This ansatz quickly becomes intractable for the high dimensional discretized system.

Therefore, we will put effort in the development of a solver based on finite element discretization and conjugate gradient method combined with a continuation step.
\appendix
\section{Numerical method implemented in \OCMAT}
\label{sec:NumericalMethodImplementedInOCMAT}
In this section we formulate the basic method for the calculation of paths that converge to an equilibrium of the canonical system.

\subsection{The core problem}
\label{sec:TheCoreProblem}

The core problem that has to be solved is the following. Given an equilibrium $(\hat x,\hat\lambda)$ and an initial distribution $x_0$ we want to find a path $(x(\cdot),\lambda(\cdot))$ satisfying \cref{eq:cansys0D} together with the boundary conditions
\begin{equation}
\label{eq:bdstabpath}
	x(0)=x_0\quad\text{and}\quad\lim_{t\to\infty}(x(t),\lambda(t))=(\hat x,\hat\lambda).
\end{equation}
The dimension of the according eigenspace
\begin{equation*}
	0<\dim\eigspace{s}(\hat x,\hat\lambda)=n_s\le n.
\end{equation*}
Then, the boundary condition at infinity in \cref{eq:bdstabpath} can be approximated by the so called asymptotic boundary condition \citep[cf.~][]{lentini1978,lentinikeller1980}
\begin{equation}
\label{eq:bdstabpathas}
	\asymbcmat^\transp\left(\begin{pmatrix}
	\hat x\\
	\hat\lambda
\end{pmatrix}-\begin{pmatrix}
	x(T)\\
	\lambda(T)
\end{pmatrix}\right)=\zerovec\in\R^{2n-n_s},\quad\asymbcmat\in\R^{2n\times(2n-n_s)}\quad\text{and}\quad\asymbcmat\bot\eigspace{s}(\hat x,\hat\lambda)
\end{equation}
with $T>0$ large enough.

For a compact notation we introduce
\begin{equation*}
	X\defin\begin{pmatrix}
	x\\
	\lambda
\end{pmatrix}\quad\text{and \cref{eq:dslocmcansys} is written as}\quad\dot X(t)=F(X(t)).
\end{equation*}
Then the previous \BVP{} writes as
\begin{subequations}
\label{eq:contbvp}
\begin{align}
	&\dot X(t)=TF(X(t)),\quad t\in[0,1]\label{eq:contbvp1}\\
	&X^{(i_1,\ldots,i_{n_s})}(0)=x_0^{(i_1,\ldots,i_{n_s})}\label{eq:contbvp2}\\
	&\asymbcmat^\transp(\hat X-X(1))=0.
\intertext{If $\hat X$ satisfies \spp{} then $n_s=n$ and \cref{eq:contbvp2} simplifies to $X^{(1,\ldots,n)}(0)=x_0$. To keep notation simple we assume that the coordinates of $x$ are sorted, such that $(i_1,\ldots,i_{n_s})=(1,\ldots,n_s)\definr(n_s)$. Then \cref{eq:contbvp2} can be rewritten as}
&X^{(n_s)}(0)=x_0^{(n_s)}\tag{\theparentequation b'}\label{eq:contbvp2p}
\end{align}
\end{subequations}
In general \bvpcref{eq:contbvp} cannot be solved analytically, hence numerical methods have to be applied. These numerical methods request some initial function $\tilde X(\cdot)$. Such an initial function need not satisfy the \bvpcref{eq:contbvp} but, depending on the problems properties, it has to a more or less good approximation. What can we do if such an initial function is not at hand?

\subsection{Embedding into a homotopy problem}
\label{sec:EmbeddingIntoAHomotopyProblem}
Given that a solution $Y(\cdot)$ of \cref{eq:contbvp} with $Y^{(n)}(0)=x_1^{(n)}\ne x_0^{(n)}$ is available we can embed \cref{eq:contbvp} into the according homotopy problem
\begin{subequations}
\label{eq:homcontbvp}
\begin{align}
	&\dot X(t)=TF(X(t),\modelpar),\quad t\in[0,1]\label{eq:homcontbvp1}\\
	&X^{(n)}(0)=x_0^{(n)}+(1-\contpar)(x_1^{(n)}-x_0^{(n)})\label{eq:homcontbvp2}\\
	&\asymbcmat^\transp(\hat X-X(1))=0.
\end{align}
Then $Y(\cdot)$ solves \cref{eq:homcontbvp} for $\contpar=0$ and $\contpar=1$ yields a solution of \bvpcref{eq:contbvp}. \OCMAT{} solves the homotopy \bvpcref{eq:homcontbvp} using arclength continuation \citep{kuznetsov1998,allgowergeorg2003}. Thus, at each homotopy step $i>0$ with previous solution $(X_{(i-1)}(\cdot),\contpar_{i-1})$ the \bvpcref{eq:homcontbvp} is solved for $(X_{(i)}(\cdot),\contpar_{i})$ together with the additional equation
\begin{align}
	&\int_0^1 (X_{(i)}(t)-X_{(i-1)}(t))^\transp V_{(i-1)}(t)\!\Dt+(\contpar_i-\contpar_{i-1})V_{(i-1),\contpar}=0\label{eq:homcontbvp3}
	\intertext{where $(V_{(i-1)}(\cdot),V_{(i-1),\contpar})$ satisfies the linearized \BVP}
	&\dot V_{(i-1)}(t)=TF_X(X_{(i-1)}(t),\modelpar)V_{(i-1)}(t),\quad t\in[0,1]\label{eq:homcontbvp4}\\
	&V_{(i-1)}^{(n)}(0)=V_{(i-1),\contpar}(x_1^{(n)}-x_0^{(n)})\\
	&\asymbcmat^\transp V(1)=0.\label{eq:homcontbvp6}
\end{align}
\end{subequations}
The solution $(V_{(i-1)}(\cdot),V_{(i-1),\contpar})$ of \bvpcrefrange{eq:homcontbvp4}{eq:homcontbvp6} is called the \emph{tangent} solution at step $i-1$. In the actual \OCMAT{} implementation the \bvpcrefrange{eq:homcontbvp1}{eq:homcontbvp6} is discretized providing different discretization schemes, relying on the two native \MATL{} \BVP{} solvers \lstinline+bvp4c+, \lstinline+bvp5c+ \citep[cf.][]{kierzenkashampine2001,kierzenkashampine2008}, and the adapted solver \lstinline+bvp6c+ \citep[cf.][]{hale2006,halemoore2008}. The discretized tangent is computed at each Newton step. This is a specific arclength continuation, called Moore-Penrose continuation, cf.~\citet{kuznetsov1998}.

Subsequently we introduce some terminology to make a clear distinction between the stable paths and the set of initial points computed during a continuation process.
\begin{definition}[Slice Manifold]
\label{def:slicemanifold}
Let $X(\cdot,\contpar(s)),\ s\in I\subset\R$ with $I$ a non-empty interval and $\contpar(\cdot)\in\Cset[0](I,\R)$ be a solution of \cref{eq:homcontbvp} for every $s\in I$. Then 
\begin{equation}
\label{eq:slicemanifold}
	\smf(\hat X,x_0^{(n_s)},x_1^{(n_s)},\contpar(\cdot))\defin\{X(0,\contpar(s)):\ s\in I\}
\end{equation}
is called the \emph{slice manifold} along $x_0^{(n_s)},x_1^{(n_s)}$ for $\hat X$ and $\contpar(\cdot)$.
\end{definition}
\begin{definition}[Comparable Slice Manifolds]
\label{def:consslicemanifold}
Let $\hat X$ satisfy \spp{} and $\smf(\hat X_j,x_0^j,x_1^j,\contpar_j(s_j)),\ s_j\in I_j,\ j=1,2$ be two slice manifolds along $x_0^j,x_1^j$ for $\hat X_j$ and $\contpar_j(\cdot)$ with $x_0^j\neq x_1^j,\ j=1,2$. Then $\smf(\hat X_j,x_0^j,x_1^j,\contpar_j(s_j)),\ j=1,2$ are called \emph{comparable} $\iff$
\begin{equation}
\label{eq:intslicemanifold}
	\{x_0^1+(1-\alpha_1)(x_1^1-x_0^1):\ \alpha_1\in\R\}=\{x_0^2+(1-\alpha_2)(x_1^2-x_0^2):\ \alpha_2\in\R\}
\end{equation}
holds. 

If comparable slice manifolds satisfy
\begin{equation}
\label{eq:consslicemanifold}
	\{x_0^1+(1-\contpar_1(s))(x_1^1-x_0^1):\ s\in I_1\}\cap\{x_0^2+(1-\contpar_2(s))(x_1^2-x_0^2):\ s\in I_2\}\neq\emptyset
\end{equation}
it is said that the slice manifolds are \emph{intersecting}.
\end{definition}
\begin{remark}
A slice manifold is a linear cut through the stable manifold. At the intersection of two different slice manifolds the cuts are given for the same (initial) states $x_0$. Hence the according paths are (different) solution candidates for the optimal control problem with $x(0)=x_0$. For one state autonomous optimal control problems with the stable path $(x(\cdot),\lambda(\cdot))$ converging to $(\hat x,\hat\lambda)$
\begin{equation*}
	\smf(x_0,x(T),\hat x,\id_{[0,1]})=\{(x(t),\lambda(t)):\ t\in[0,T]\}.
\end{equation*}
Thus, the orbit of the stable path coincides with the stable manifold. Moreover two slice manifolds for different saddles are trivially comparable.
\end{remark}
There are good reasons to consider \bvpcref{eq:homcontbvp} instead of \bvpcref{eq:contbvp}. For an arbitrary initial point $x_0$ it is often hard to provide a ``good'' guess of an initial function for \bvpcref{eq:contbvp}. Since in general the solution of an \BVP{} is not unique it may not be guaranteed that a computed solution is the searched for solution. 

On the other hand the equilibrium solution trivially satisfies \bvpcref{eq:contbvp} for the initial point $\hat x$. Hence the homotopy \bvpcref{eq:homcontbvp} can be started with an exact solution. Then the existence of a unique solution is guaranteed by the \implicth{} as long as some rank condition is satisfied. A careful inspection of the linearization of \bvpcref{eq:homcontbvp}, which is a byproduct of the arclength continuation, then yields important information about the behavior of the solution paths. Moreover, these intermediate solution paths can be used to find, e.g. an indifference threshold point. 

\subsubsection{The stable manifold for an equilibrium not satisfying \spp}
\label{sec:TheStableManifoldForAnEquilibriumNotSatisfyingSpp}
In the previous section we assumed that $\hat X$ satisfies \spp, i.e., $\dim\eigspace{s}(\hat J)=n$. Next we adapt \bvpcref{eq:homcontbvp} for the case $\dim\eigspace{s}(\hat J)=n_s<n$. Therefore, we assume that $\hat X$ is hyperbolic with $\dim\eigspace{u}(\hat J)=n_u$ and hence $2n=n_s+n_u$. Furthermore, two points $x_i^{(n)}$ in the state space are fixed and $n-n_s$ vectors $v_i\in\R^n$ are chosen. Then, the according \BVP{} for the calculation of stable paths converging to $\hat X$ becomes
\begin{subequations}
\label{eq:homcontbvpnspp}
\begin{align}
	&\dot X(t)=TF(X(t),\modelpar),\quad t\in[0,1]\\
	&X^{(n)}(0)=x_1^{(n)}+(1-\contpar_0)(x_0^{(n)}-x_1^{(n)})+\sum_{i=1}^{n-n_s}\contpar_iv_i\label{eq:homcontbvpnspp2}\\
	&\asymbcmat^\transp(\hat X-X(1))=0\in\R^{n_u}
	\shortintertext{with}
	&\rank\begin{pmatrix}
	(x_0^{(n)}-x_1^{(n)}) & v_1 & \cdots & v_{n-n_s}
	\end{pmatrix}=n-n_s+1\label{eq:homcontbvpnspp4}\\
	&\asymbcmat\bot\eigspace{s}(\hat J)=0\quad\text{and}\quad\asymbcmat\in\R^{2n\times n_u}.\notag
\end{align}
\end{subequations}
Let us assume that a solution $Y(\cdot)$ for $\contpar_j=0,\ j=0,\ldots,n-n_s$ is known. Then \cref{eq:homcontbvpnspp2} can be interpreted as the condition that we search for a solution along the direction $x_0^{(n)}-x_1^{(n)}$. Additionally we have to take care  letting enough freedom, guaranteed by \ref{eq:homcontbvpnspp4}, to start at the stable manifold.
\subsubsection{Continuation of an indifference threshold point}
\label{sec:ContinuationOfAnIndifferenceThresholdPoint}
A useful relation between the Hamiltonian and objective value is given for a \modelcref{eq:gen_0Dmodel} with $\rho>0$ and finite objective value \citep[cf.][]{michel1982}. Then we find for any solution $x(\cdot),\lambda(\cdot)$ of the canonical system \cref{eq:cansys0D}
\begin{equation}
\label{eq:hamobj}
	J(x_0)=\frac{1}{\rho}\Ha(x(0),\lambda(0))
\end{equation}
where $\Ha$ is defined according to \cref{eq:cansys0D5} and the bar is omitted.

Let $\hat Y_i,\ i=1,2$ be two \css{} of \cref{eq:cansys0D} and $x^1_I$ and $x^2_I$ be two distinct indifference threshold points of \modelcref{eq:gen_0Dmodel}. Furthermore, let $Z_{1,2}(\cdot)$ be two solutions corresponding to $x^1_I$ and $Z_{3,4}(\cdot)$ two solutions corresponding to $x^2_I$. To continue the indifference threshold point from $x^1_I$ to $x^2_I$ we solve the following homotopy problem
\begin{subequations}
\label{eq:homconitptbvp}
\begin{align}
	&\dot X_1(t)=T_1F(X_1(t)),\quad t\in[0,1]\label{eq:homconitptbvp1}\\
	&\dot X_2(t)=T_2F(X_2(t)),\quad t\in[0,1]\label{eq:homconitptbvp2}\\
	&X_1^{(n)}(0)=X_2^{(n)}(0)\in\R^n\label{eq:homconitptbvp3}\\
	&\Ha(X_1(0))-\Ha(X_2(0))=0\in\R\label{eq:homconitptbvp4}\\
	&\asymbcmat_1^\transp(\hat Y_1-X_1(1))=0\in\R^n\label{eq:homconitptbvp5}\\
	&\asymbcmat_2^\transp(\hat Y_2-X_2(1))=0\in\R^n\label{eq:homconitptbvp6}\\
	&X_1^{(n)}(0)=x^2_I+(1-\contpar_1)(x^1_I-x^2_I)+\contpar_2V\in\R^n\label{eq:homconitptbvp7}
	\shortintertext{with}
	&a_1V+a_2(x^1_I-x^2_I)=0\quad\text{and}\quad \norm{a_1}+\norm{a_2}\ne0\notag\\
	&\asymbcmat_i\bot\eigspace{s}(\hat J_i),\ i=1,2.\notag
\end{align}
\end{subequations}
\Cref{eq:homconitptbvp1,eq:homconitptbvp2} denote the dynamics for the two distinct paths, starting at the same initial states, \cref{eq:homconitptbvp3}. The truncation times $T_1>0$ and $T_2>0$ may be chosen differently. The two paths $X_1(\cdot)$ and $X_2(\cdot)$ yield the same objective value which according to \cref{eq:hamobj} can be stated as \cref{eq:homconitptbvp4}. \Cref{eq:homconitptbvp5,eq:homconitptbvp6} denote the asymptotic boundary conditions for the path $X_1(\cdot)$ converging to $\hat Y_1$ and $X_2(\cdot)$ converging to $\hat Y_2$, respectively. Finally \cref{eq:homconitptbvp7} specifies the continuation in the state space, where the first part $x^2_I+(1-\contpar_1)(x^1_I-x^2_I)$ describes the change into the direction to the target $x^1_I$ and $\contpar_2V$ is a correction term, since the stable manifold has one dimension less than the state space.
 
Counting the number of unknowns and equations, we find $4n+2$ unknowns, two times the states and costates $X_i(\cdot)$ and the two free parameters $\contpar_i,\ i=1,2$ and $4n+1$ equations. Moreover \cref{eq:homconitptbvp} is solved for $(Z_{1,2}(\cdot),0,0)$ and $(Z_{3,4}(\cdot),1,0)$. Thus we can start a continuation process starting with one of these previously detected solutions.

\section{The usage of \OCMAT}
\label{sec:TheUsageOfOCMAT}
In this section the basic steps for the numerical analysis of model \labelcref{eq:dslocm} are explained in detail. This enables the user to reproduce the presented results and learn the basic commands and structure of \OCMAT.

\subsection{The initialization file}
\label{sec:TheInitializationFile}
To get results that are comparable smooth in the spatial dimension as the discretization used in \citet{grassuecker2015} we choose $N=51$. In the syntax logic of \OCMAT\ we have to provide an initialization file consisting of $N+1$ \ODE s, the entry of $N+1$ state and control variables $P_i, u_i, i=0,\ldots, N$ and the appropriate objective function. Doing that by hand can become a boring and error-prone task. We therefore wrote a \MATL\ file \lstinline+'makeinitfile'+ that generates the initialization file, providing $N$\footnote{If this \FDM\ approach turns out to be an important tool by itself, this step could be directly implemented within \OCMAT.}
\begin{matlab}
function makeinitfile(N,modelname)

if nargin==1
    modelname='shallowlakeline';
end
dotPxi='uxi-b*Pxi+Pxi^2/(1+Pxi^2)+D*(N/2/L)^2*(Pxim1-2*Pxi+Pxip1)';
intvar='';
for ii=0:N
    ode{ii+1}=['DPx' num2str(ii) '=' strrep(strrep(strrep(strrep(strrep(dotPxi,'i',num2str(ii)),['x' num2str(ii) 'm1'],['x' num2str(ii-1)]),['x' num2str(ii) 'p1'],['x' num2str(ii+1)]),'Px-1','Px1'),['Px' num2str(N+1)],['Px' num2str(N-1)])];
    ode{ii+1}=char(simple(sym(ode{ii+1})));
    if ii>0
        controlvar=[controlvar ',ux' num2str(ii)];
        statevar=[statevar ',Px' num2str(ii)];
        if ii<N
            intvar=[intvar '+log(ux' num2str(ii) ')-c*Px' num2str(ii) '^2' ];
        else
            intvar=[intvar '+(log(ux0)-c*Px0^2+log(ux' num2str(ii) ')-c*Px' num2str(ii) '^2)/2' ];
        end
    else
        controlvar=['ux' num2str(ii)];
        statevar=['Px' num2str(ii)];
    end
end
intvar(1)=[];
intvar=['(' char(collect(simple(sym(intvar)),'c')) ')'];

par={'rho::0.03','b::0.65','c::0.5','D::0.5','L::2*pi/0.44',['N::' num2str(N)]};
initfilefid=fopen([modelname '.ocm'],'w');
fprintf(initfilefid,'Type\nstandardmodel\n\nVariable\nstate::
for ii=1:length(ode)
    fprintf(initfilefid,'ode::
end
fprintf(initfilefid,'\nObjective\nexpdisc::rho\nint::
for ii=1:length(par)
    fprintf(initfilefid,'
end
fclose(initfilefid);
\end{matlab}
As an example we call
\begin{matlab}
>> makeinitfile(5,'shallowlakelinetest')
\end{matlab}
at the \MATL\ workspace and the file \lstinline+'shallowlakelinetest.ocm'+ is generated
\begin{matlab}
Type
standardmodel

Variable
state::Px0,Px1,Px2,Px3,Px4,Px5
control::ux0,ux1,ux2,ux3,ux4,ux5

Statedynamics
ode::DPx0=ux0-b*Px0+Px0^2/(1+Px0^2)+D*(N/2/L)^2*(Px1-2*Px0+Px1)
ode::DPx1=ux1-b*Px1+Px1^2/(1+Px1^2)+D*(N/2/L)^2*(Px0-2*Px1+Px2)
ode::DPx2=ux2-b*Px2+Px2^2/(1+Px2^2)+D*(N/2/L)^2*(Px1-2*Px2+Px3)
ode::DPx3=ux3-b*Px3+Px3^2/(1+Px3^2)+D*(N/2/L)^2*(Px2-2*Px3+Px4)
ode::DPx4=ux4-b*Px4+Px4^2/(1+Px4^2)+D*(N/2/L)^2*(Px3-2*Px4+Px5)
ode::DPx5=ux5-b*Px5+Px5^2/(1+Px5^2)+D*(N/2/L)^2*(Px4-2*Px5+Px4)

Objective
expdisc::rho
int::((-Px1^2-Px2^2-Px3^2-Px4^2-1/2*Px0^2-1/2*Px5^2)*c+log(ux1)+log(ux2)+log(ux3)+log(ux4)+1/2*log(ux0)+1/2*log(ux5))

Parameter
rho::0.03
b::0.65
c::0.5
D::0.5
L::2*pi/0.44
N::5
\end{matlab}
This file is placed at the actual \MATL\ directory, to make it visible for \OCMAT\ it has to be moved to the folder \lstinline+ocmat\model\initfiles+. After that the initialization process of \OCMAT\ can be started
\begin{matlab}
>> ocStruct=processinitfile('shallowlakelinetest');
ocmat\model\usermodel\shallowlakelinetest does not exist. Create it?  (y)/n: y
ocmat\model\usermodel\shallowlakelinetest\data does not exist. Create it?  (y)/n: y
ocmat\model\usermodel\shallowlakelinetest\data is not on MATLAB path. Add it?  (y)/n: y
ocmat\model\usermodel\shallowlakelinetest is not on MATLAB path. Add it?  (y)/n: y
>> modelfiles=makefile4ocmat(ocStruct);
>> moveocmatfiles(ocStruct,modelfiles)
\end{matlab}
For the following analysis we use the initialization file \lstinline+shallowlakelinecoarse+ with $N=51$. Thus, the previous initialization commands have to be repeated with \lstinline+shallowlakelinecoarse+ instead of \lstinline+shallowlakelinetest+. Note that even though the discretization parameter $N$ appears as a parameter of the model \lstinline+shallowlakelinecoarse+ it must not be changed. The initialization steps can be very time consuming, depending on the chosen value of $N$ and the computer capacity. For $N=51$ the initialization, on a PC with the specifications Intel Core i7-3820 CPU@3.60GHz, Windows 7, \MATL\ 7.6.0, took around half an hour.

Subsequently we will also make use of the 0D shallow lake model. Therefore, we have to run through the initialization process for the model \lstinline+shallowlake+\footnote{The according initialization file \lstinline+shallowlake.ocm+ can be found in the folder \lstinline+ocmat/mode/initfiles+.} as well.
\begin{matlab}
>> ocStruct=processinitfile('shallowlake');
ocmat\model\usermodel\shallowlake does not exist. Create it?  (y)/n: y
ocmat\model\usermodel\shallowlake\data does not exist. Create it?  (y)/n: y
ocmat\model\usermodel\shallowlake\data is not on MATLAB path. Add it?  (y)/n: y
ocmat\model\usermodel\shallowlake is not on MATLAB path. Add it?  (y)/n: y
>> modelfiles=makefile4ocmat(ocStruct);
>> moveocmatfiles(ocStruct,modelfiles)
\end{matlab}

\subsection{Detection and continuation of equilibria}
\label{sec:cont_equi}
For the calculation of the \fcss{} we use the equilibria of the \lstinline+shallowlake+ model. A detailed explanation of the commands can be found in the \OCMAT{} manual.
\begin{matlab}
>> m0=stdocmodel('shallowlake');
>> m0=changeparametervalue(m0,'b,c',[0.65 0.5]);
>> ocEP0=calcep(m0);b=isadmissible(ocEP0,m0,[],'UserAdmissible');ocEP0(~b)=[];
\end{matlab}
The values of the equilibria are used to generate the values of \fcss.
\begin{matlab}
>> m=stdocmodel('shallowlakelinecoarse');
>> N=parametervalue(m,'N');
>> y0=[ocEP0{1}.y([1 end]) ocEP0{2}.y([1 end]) ocEP0{3}.y([1 end])];
>> Y=[y0([ones(N+1,1);2*ones(N+1,1)],1) y0([ones(N+1,1);2*ones(N+1,1)],2) y0([ones(N+1,1);2*ones(N+1,1)],3)];
>> Y([N+2 2*N+2],:)=Y([N+2 2*N+2],:)/2;
>> opt=setocoptions('EQ','TolFun',1e-12,'MaxFunEvals',50000,'MaxIter',50000);
>> ocEP=calcep(m,Y,[],opt);b=isadmissible(ocEP,m);ocEP(~b)=[];
>> [b dfct]=isspp(m,ocEP{:})
b =
     1     0     1
dfct =
     0    -5     0
\end{matlab}
Thus, the first and third equilibrium satisfy \spp, the second equilibrium has defect $-5$. For the subsequent step an adapted version of \CLMATCONT{} has been used.
\begin{matlab}
>> opt0=setocoptions('MATCONT','MaxNumPoints',750,'MaxStepsize',1e-1,'Backward',0,'IgnoreSingularity',2,'OCCONTARG','CheckAdmissibility','off');
>> opt1=setocoptions(opt0,'MATCONT','Backward',1);
>> contpar='b';epidx=1;[x0 v0 s0 f0 h0]=contep(m,ocEP{epidx},contpar,opt0);
first point found
tangent vector to first point found
elapsed time  = 107.2 secs
npoints curve = 750
>> store(m,'modelequilibrium')
\end{matlab}
During the continuation of the first equilibrium four branching points are detected. These branching points are the initial solutions for heterogeneous equilibria. Repeating the previous steps for the third equilibrium \lstinline+epidx=3+ in both continuation directions yields the eutrophic arc. On this arc no branching point exists. With \lstinline+store(m,'modelequilibrium')+ the results of the continuation process are stored into the results of model \lstinline+m+.
\begin{matlab}
>> contpar='b';epidx=3;[x0 v0 s0 f0 h0]=contep(m,ocEP{epidx},contpar,opt0);
first point found
tangent vector to first point found
elapsed time  = 72.9 secs
npoints curve = 750
>> store(m,'modelequilibrium')
>> contpar='b';epidx=3;[x1 v1 s1 f1 h1]=contep(m,ocEP{epidx},contpar,opt1);
first point found
tangent vector to first point found
elapsed time  = 76.2 secs
npoints curve = 750
>> store(m,'modelequilibrium')
\end{matlab}
The calculation of the arcs emanating from the branching points are exemplified for the second branching point.
\begin{matlab}
>> n=3;m1=changeparametervalue(m,'b',[x0(end,s0(n).index)]);
>> ocEP1=calcep(m1,x0(1:end-1,s0(n).index),0,opt);
>> opt0=setocoptions(opt0,'MATCONT','MaxNumPoints',1500,'MaxStepsize',2e-1,'CheckClosed',50);
>> epidx=1;[xbp0 vbp0 sbp0 fbp0]=contbp(m1,ocEP1{epidx},s0(n),0.01,opt0);
first point found
tangent vector to first point found
label = BP, x = ( 0.619422 0.619661 0.620389 ... -5.917192 -5.917341 -2.958687 0.721718 )
label = BP, x = ( 0.554176 0.554239 0.554455 ... -6.412585 -6.407783 -3.203096 0.706011 )
label = BP, x = ( 0.398031 0.398863 0.401801 ... -9.309244 -9.141784 -4.542734 0.626785 )

elapsed time  = 249.5 secs
npoints curve = 1500
>> store(m,'modelequilibrium')
\end{matlab}
During the continuation starting from the second branching point, i.e. the bifurcation point stored in the structure \lstinline+s0+ at index \lstinline+n=3+, further branching points are detected. The branches emanating from these \hcss{} can be computed in an analogous way. Repeating the previous steps for \lstinline+n=2,4,5+ finally yields all branches of \fcss{} and \hcss, see \cref{fig:sld_bifdiag1}.

Using the solutions of the previous calculations can now be used to find the all \fcss{} and \hcss{} equilibria for a specific value of $b$, e.g. $b=0.65$
\begin{matlab}
>> matRes=matcontresult(m);counter=0;b=0.65;
>> for ii=1:length(matRes), ...
	x0=[matRes{ii}.ContinuationSolution.y;matRes{ii}.ContinuationSolution.userinfo.varyparametervalue]; ...
	idx0=cont2idx(x0(end,:),b); ...
	for jj=1:length(idx0); ...
	counter=counter+1;ocEP=calcep(m,x0(1:end-1,idx0(jj)),0,opt); ...
	store(m,ocEP);end,end
\end{matlab}
For $b=0.65$ there exist $17$ equilibria.\footnote{With \lstinline+load(m,'\%1.2f',[],'b,c')+ the model where these equilibria and other numerical results are already stored can be loaded into the \MATL{} workspace.}

We will give a two examples for the calculation of a stable path to an equilibrium that satisfies \spp{} and that does not satisfy \spp.
\subsection{Saddle path calculation}
\label{sec:saddle_path}
One of the main tasks of \OCMAT{} is the calculation of a stable path converging to an equilibrium of saddle-type. The computation is done by solving the homotopy problem \cref{eq:homcontbvp} or \cref{eq:homcontbvpnspp}. We start with a problem of the first type.

\subsubsection{Stable path when \spp{} is satisfied}
\label{sec:StablePathWhenSppIsSatisfied}
As an example we compute, for the parameter values $b=0.65$ and $c=0.5$, the stable path that starts at the states of a defective \hcss{} (\lstinline+ocEP{9}+) and converge to the eutrophic \fcss{} (\lstinline+ocEP{17}+).\footnote{The order refers to the results stored in the file \lstinline+shallowlakelinecoarse_b_0.65_c_0.50.mat+.} First we load the model with the stored data of the equilibria.
\begin{matlab}
>> m=stdocmodel('shallowlakelinecoarse');m=changeparametervalue(m,'b,c',[0.65 0.5]);
>> load(m,'
\end{matlab}
Next we determine the flat equilibria satisfying \spp{} and their size of the states. 
\begin{matlab}
>> idx=find(isflat(m,ocEP{:}) & isspp(m,ocEP{:}))
idx =
     1    17
>> P=state(m,ocEP{1});P(1)
ans =
    0.4530
>> P=state(m,ocEP{17});P(1)
ans =
    1.4370
\end{matlab}
To start the continuation process we change some of the default options and call the initialization function \lstinline+initocmat_AE_EP+. The truncation time \lstinline+T+ is determined by
\begin{equation*}
	T=\frac{T_0}{\min_{\eigval\in\eigspace{s}}\norm{\Re\eigval}},
\end{equation*}
where $T_0$ usually is set to $10$.
\begin{matlab}
>> opt=setocoptions('OCCONTARG','MaxStepWidth',1,'InitStepWidth',5e-1,'CheckAdmissibility','off','SBVPOC','MeshAdaptAbsTol',1e-4,'MeshAdaptRelTol',1e-3,'GENERAL','TrivialArcMeshNum',20);
>> eval=real(eig(ocEP{17}));eval(eval>0)=[];T=10/min(abs(eval));
>> sol=initocmat_AE_EP(m,ocEP{17},1:N+1,ocEP{9}.y(1:N+1),opt,'TruncationTime',T);
\end{matlab}
After the initialization process the continuation process is started calling \lstinline+bvpcont+.
\begin{matlab}
>> c=bvpcont('extremal2ep',sol,[],opt);
first solution found
tangent vector to first solution found

 Continuation step No.: 1
 stepwidth: 0.5
 Newton Iterations: 1
 Mesh size: 21
 Continuation parameter: 0.0952864


 Continuation step No.: 15
 stepwidth: 1
 Newton Iterations: 1
 Mesh size: 27
 Continuation parameter: 1.01767

 Target value hit.
 label=HTV
 Continuation parameter=1

elapsed time  = 38.4 secs
>> store(m,'extremal2ep');
>> save(m)
\end{matlab}
Finally the result of the continuation is stored in the model-object \lstinline+m+, for further use. The \lstinline+save+ command allows to store the entire object, i.e. with the stored results, as a \MATL\ data file.

\subsubsection{Stable path when \spp{} is not satisfied}
\label{sec:StablePathWhenSppIsNotSatisfied}
The next example was used for the computation of the second and third homotopy process in \cref{sec:HOSSSatisfyingSPP}, cf.~\cref{fig:fossnhoss300}. The parameter values are $b=0.55$ and $c=3.0825$. Thus, in a first step we try to find a solution that starts at the states of the flat oligotrophic equilibrium (\lstinline+ocEP{1}+) and converge to the heterogeneous equilibrium satisfying \spp{} (\lstinline+ocEP{7}+). For that reason we repeat the steps of the previous \cref{sec:StablePathWhenSppIsSatisfied}.
\begin{matlab}
>> opt=setocoptions(opt,'OCCONTARG','MaxStepWidth',1e1,'InitStepWidth',5e-1,'MaxContinuationSteps',30,'SBVPOC','MeshAdaptAbsTol',1e-3,'MeshAdaptRelTol',1e-2);
>> m=changeparametervalue(m,'b,c',[0.55 3.0825]);
>> load(m,'
The results in the actual model 'm' are overwritten. Proceed? y/(n) : y
>> ocEP=equilibrium(m);
>> eval=real(eig(ocEP{7}));eval(eval>0)=[];T=10/min(abs(eval))
>> sol=initocmat_AE_EP(m,ocEP{7},1:N+1,ocEP{1}.y(1:N+1),opt,'TruncationTime',T);
>> c=bvpcont('extremal2ep',sol,[],opt);
>> store(m,'extremal2ep');ocEx=extremalsolution(m);n=length(ocEx);
\end{matlab}
During the continuation process we encounter that the detected solution does not end ``near'' the equilibrium, i.e. the used truncation time becomes too short. The reason becomes obvious when we have a look on \cref{fig:fossnhoss3}. The computed solution path approaches a stable path of the defective \hcssnspp. Therefore, the time it takes to stay in the vicinity of this equilibrium increases. To overcome this problem we extend the homotopy problem \cref{eq:homcontbvp} by letting the truncation time $T$ be a free parameter value and adding a further constraint, that guarantees that the solution $X(\cdot)$ not only ends at the the (linearized) stable manifold, but also satisfies
\begin{equation}
\label{eq:nearequilib}
	\Norm{X(1)-\hat X}=\varepsilon,
\end{equation}
with some fixed $\varepsilon>0$. To start this extended continuation process we use the solution after $30$ continuation steps\footnote{We therefore set the option \lstinline+'MaxContinuationSteps',30+.} and use the initialization argument \lstinline+'movinghorizon'+
\begin{matlab}
>> sol=initocmat_AE_AE(m,ocEx{n},[1:N+1],ocEP{1}.y(1:N+1),opt,'movinghorizon',1)
>> opt=setocoptions(opt,'OCCONTARG','MaxStepWidth',1e2,'InitStepWidth',5e0,'MaxContinuationSteps',70);
>> c=bvpcont('extremal2ep',sol,[],opt);
>> store(m,'extremal2ep');ocEx=extremalsolution(m);n=length(ocEx);
\end{matlab}
Next we compute a stable path that converges to the defective equilibrium \hcssnspp. The stable manifold of the defective equilibrium (it has defect $-1$) is of dimension $(N+1)-1$. Taking the initial states $P^d(0)$ of the solution of the last continuation process, in \OCMAT{} notation \lstinline+ocEx{n}.y(1:N+1,1)+, there exists $\contpar_1$ such that $P_0=P^d(0)+\contpar_1v_1$ with $v_1=(1,\ldots,1)^\transp\in\R^{N+1}$ is lying in the $N$-dimensional stable manifold. Thus, we solve the according homotopy problem.
\begin{matlab}
>> opt=setocoptions(opt,'OCCONTARG','MaxStepWidth',1e1,'InitStepWidth',1e-1,'SBVPOC','MeshAdaptAbsTol',1e-4,'MeshAdaptRelTol',1e-3);
>> V1=ones(N+1,1);eval=real(eig(ocEP{8}));eval(eval>0)=[];T=10/min(abs(eval))
>> sol=initocmat_AE_EP(m,ocEP{8},[1:N+1],ocEx{n}.y(1:N+1,1),opt,'TruncationTime',T,'freevector',V1);
>> c=bvpcont('extremal2ep',sol,[],opt);
\end{matlab}
The result of the computations can be plotted using \OCMAT{\ plotting commands.
\begin{matlab}
>> clf;xcoord=1;ycoord=2;xvar='spatialnorm';yvar='spatialnorm';
>> plotcont(m,xvar,xcoord,yvar,ycoord,'contfield','ExtremalSolution','Index',[2 4 5]);hold on,
>> plotlimitset(m,xvar,xcoord,yvar,ycoord,'Index',[1 7 8],'Marker','.','MarkerSize',10,'showspp',1,'showflat',1);
>> hold off;figure(gcf)
\end{matlab}

\subsubsection{Indifference threshold point and manifold}
\label{sec:SkibaPointAndManifold}
This example presents in detail the computation of the results from \cref{sec:DetectionOfAnIndifferenceThresholdPoint}. First we load the models data, retrieve the equilibria and remove the already stored continuation results.
\begin{matlab}
>> m=stdocmodel('shallowlakelinecoarse');m=changeparametervalue(m,'b,c',[0.65 0.5]);
>> load(m,'
>> removeresult(m,'Continuation');
\end{matlab}
The solutions that start at the states of the seventh equilibrium, a \hcssspp, and converge to the eutrophic and oligotrophic equilibrium are computed.
\begin{matlab}
>> idx=find(~isflat(m,ocEP{:}) & isspp(m,ocEP{:}))
idx =
     4     7
>> opt=setocoptions('OCCONTARG','MaxStepWidth',1e1,'InitStepWidth',5e-1,'CheckAdmissibility','off','SBVPOC','MeshAdaptAbsTol',1e-4,'MeshAdaptRelTol',1e-3,'GENERAL','TrivialArcMeshNum',20);		
>> eval=real(eig(ocEP{1}));eval(eval>0)=[];T=10/min(abs(eval));
>> sol=initocmat_AE_EP(m,ocEP{1},1:N+1,ocEP{7}.y(1:N+1),opt,'TruncationTime',T);
>> c=bvpcont('extremal2ep',sol,[],opt);
>> store(m,'extremal2ep');ocEx=extremalsolution(m);n=length(ocEx);
>> eval=real(eig(ocEP{17}));eval(eval>0)=[];T=10/min(abs(eval));
>> sol=initocmat_AE_EP(m,ocEP{17},1:N+1,ocEP{7}.y(1:N+1),opt,'TruncationTime',T);
>> c=bvpcont('extremal2ep',sol,[],opt);
>> store(m,'extremal2ep');ocEx=extremalsolution(m);n=length(ocEx);
\end{matlab}
Plotting the objective value (Hamiltonian) shows that the eutrophic solution is the optimal solution. Next the last eutrophic solution is continued (for ten steps) in direction of the oligotrophic equilibrium. 
\begin{matlab}
>> clf;xcoord=1;ycoord=1;xvar='spatialnorm';yvar='hamiltonian';
>> plotcont(m,xvar,xcoord,yvar,ycoord,'contfield','SliceManifold','Index',[1 2]);figure(gcf)
>> opt=setocoptions(opt,'OCCONTARG','InitStepWidth',2.5e0,'MaxContinuationSteps',10);
>> sol=initocmat_AE_AE(m,ocEx{2},1:N+1,ocEP{1}.y(1:N+1),opt);
>> c=bvpcont('extremal2ep',sol,[],opt);
>> store(m,'extremal2ep');ocEx=extremalsolution(m);n=length(ocEx);
\end{matlab}
Comparing the objective value of the solutions for the first and third continuation process yields an (heterogeneous) indifference threshold point. Finally the solutions starting at the indifference threshold point and converging to the \foss{} are computed.
\begin{matlab}
>> ipt0=findindifferencepoint(m,1,3);
>> opt=setocoptions(opt,'OCCONTARG','InitStepWidth',1e1,'MaxContinuationSteps',inf);
>> eval=real(eig(ocEP{1}));eval(eval>0)=[];T=10/min(abs(eval));
>> sol=initocmat_AE_EP(m,ocEP{1},1:N+1,ipt0,opt,'TruncationTime',T);
>> c=bvpcont('extremal2ep',sol,[],opt);
>> store(m,'extremal2ep');ocEx=extremalsolution(m);n=length(ocEx);
>> sol=initocmat_AE_AE(m,ocEx{2},1:N+1,ipt0,opt);
>> c=bvpcont('extremal2ep',sol,[],opt);
>> store(m,'extremal2ep');ocEx=extremalsolution(m);n=length(ocEx);
\end{matlab}
The analogous computations are done for the second \hcssspp{} (\lstinline+ocEP{4}+), yielding a second (heterogeneous) indifference threshold point. 
\begin{matlab}
>> opt=setocoptions(opt,'OCCONTARG','MaxStepWidth',1e1,'InitStepWidth',5e-1,'MaxContinuationSteps',30);
>> eval=real(eig(ocEP{1}));eval(eval>0)=[];T=10/min(abs(eval));
>> sol=initocmat_AE_EP(m,ocEP{1},1:N+1,ocEP{4}.y(1:N+1),opt,'TruncationTime',T);
>> c=bvpcont('extremal2ep',sol,[],opt);
>> store(m,'extremal2ep');ocEx=extremalsolution(m);n=length(ocEx);
>> opt=setocoptions(opt,'OCCONTARG','MaxContinuationSteps',inf);
>> eval=real(eig(ocEP{17}));eval(eval>0)=[];T=10/min(abs(eval));
>> sol=initocmat_AE_EP(m,ocEP{17},1:N+1,ocEP{4}.y(1:N+1),opt,'TruncationTime',T);
>> c=bvpcont('extremal2ep',sol,[],opt);
>> store(m,'extremal2ep');ocEx=extremalsolution(m);n=length(ocEx);
>> clf;xcoord=1;ycoord=1;xvar='spatialnorm';yvar='hamiltonian';
>> plotcont(m,xvar,xcoord,yvar,ycoord,'contfield','SliceManifold','Index',[6 7]);figure(gcf)
>> opt=setocoptions(opt,'OCCONTARG','InitStepWidth',2.5e0,'MaxContinuationSteps',10);
>> sol=initocmat_AE_AE(m,ocEx{7},1:N+1,ocEP{1}.y(1:N+1),opt);
>> c=bvpcont('extremal2ep',sol,[],opt);
>> store(m,'extremal2ep');ocEx=extremalsolution(m);n=length(ocEx);
>> ipt=findindifferencepoint(m,6,8);
>> opt=setocoptions(opt,'OCCONTARG','InitStepWidth',1e1,'MaxContinuationSteps',inf);
>> eval=real(eig(ocEP{1}));eval(eval>0)=[];T=10/min(abs(eval));
>> sol=initocmat_AE_EP(m,ocEP{1},1:N+1,ipt,opt,'TruncationTime',T);
>> c=bvpcont('extremal2ep',sol,[],opt);
>> store(m,'extremal2ep');ocEx=extremalsolution(m);n=length(ocEx);
>> sol=initocmat_AE_AE(m,ocEx{7},1:N+1,ipt,opt);
>> c=bvpcont('extremal2ep',sol,[],opt);
>> store(m,'extremal2ep');ocEx=extremalsolution(m);n=length(ocEx);
\end{matlab}
Finally, we use continuation to find the intermediate indifference threshold points from the transformation of the first to the second indifference threshold point.
\begin{matlab}
>> opt=setocoptions(opt,'OCCONTARG','MaxStepWidth',1e1,'InitStepWidth',5e-1,'SBVPOC','BCJacobian',0);
>> v=ones(N+1,1);
>> sol=initocmat_AE_IDS(m,ocEx(9:10),v,ipt0,opt);
>> c=bvpcont('indifferencedistribution',sol,[],opt);
>> store(m,'indifferencedistribution');
\end{matlab}

\bibliographystyle{plainnat}
\bibliography{diffusion_instability}
\end{document}